\documentclass[10pt,journal,comsoc]{IEEEtran}
\usepackage{cite}
\usepackage[T1]{fontenc}\usepackage{bm}
\usepackage{amsmath}
\usepackage[top=0.7in, bottom=0.55in, left=0.65in, right=0.65in]{geometry}
\usepackage[table,xcdraw]{xcolor}

\usepackage{amsthm}
\usepackage{float}
\usepackage{setspace}
\usepackage{amssymb}
\usepackage{stfloats}
\usepackage{cite}
\usepackage{ragged2e}
\usepackage[ruled,vlined,linesnumbered]{algorithm2e}
\usepackage{amsfonts}
\usepackage{mathrsfs}
\usepackage{amsmath,amsthm}
\usepackage{array,booktabs}
\usepackage{subfigure}
\usepackage{multirow}
\usepackage{cuted}
\usepackage{multicol}
\usepackage{graphicx}
\usepackage{subfigure}
\usepackage{graphicx,xcolor,bm}
\usepackage{hyperref}
\usepackage{threeparttable}
\usepackage{dcolumn}
\usepackage{setspace}
\usepackage{makecell}
\usepackage{lipsum}
\usepackage{enumerate}
\usepackage{mathrsfs}
\usepackage{bbm}

\usepackage{enumitem}
\usepackage[protrusion=true,expansion=true]{microtype}
\pdfoutput=1

\newtheorem{Defn}{Definition}

\newtheorem{Prop}{Proposition}

\usepackage{xfrac}

\usepackage{cite,bm}
\graphicspath{{figures/}}

\def\BibTeX{{\rm B\kern-.05em{\sc i\kern-.025em b}\kern-.08em
    T\kern-.1667em\lower.7ex\hbox{E}\kern-.125emX}}
\usepackage{balance}
\begin{document}
\title{FUSION: Forecast-Embedded Agent Scheduling with Service Incentive Optimization over Distributed Air-Ground Edge Networks }
\author{Houyi Qi, \IEEEmembership{Graduate Student Member}, \IEEEmembership{IEEE}, Minghui Liwang, \IEEEmembership{Senior Member}, \IEEEmembership{IEEE}, \\Seyyedali Hosseinalipour, \IEEEmembership{Senior Member}, \IEEEmembership{IEEE}, Liqun Fu, \IEEEmembership{Senior Member}, \IEEEmembership{IEEE}, Sai Zou, \IEEEmembership{Senior Member}, \IEEEmembership{IEEE}, 
	Xianbin Wang, \IEEEmembership{Fellow}, \IEEEmembership{IEEE}, Wei Ni, \IEEEmembership{Fellow}, \IEEEmembership{IEEE},
	and Yiguang Hong, \IEEEmembership{Fellow}, \IEEEmembership{IEEE}
\thanks{H. Qi (houyiqi@tongji.edu.cn), M. Liwang (minghuiliwang@tongji.edu.cn), and Y. Hong (yghong@iss.ac.cn) are with the Shanghai Research Institute for Intelligent Autonomous Systems, with the State Key Laboratory of Autonomous Intelligent Unmanned Systems, Frontiers Science Center for Intelligent Autonomous Systems, Ministry of Education, Shanghai Key Laboratory of Intelligent Autonomous Systems, and Department of Control Science and Engineering, 
Tongji University, China. S. Hosseinalipour (alipour@buffalo.edu) is with the Department of Electrical Engineering, University at Buffalo-SUNY, USA. 
L. Fu (liqun@xmu.edu.cn) is with the School of Informatics, Xiamen University, China.
S. Zou (dr-zousai@foxmail.com) is with the College of Big Data and Information Engineering, Guizhou University, China. 
X. Wang (xianbin.wang@uwo.ca) is with the Department of Electrical and Computer Engineering, Western University, Canada. W. Ni (Wei.Ni@ieee.org) is with School of Engineering, Edith Cowan University, Perth, Australia.

Corresponding author: Minghui Liwang
}
}

\IEEEtitleabstractindextext{
\begin{abstract} 
	\justifying 
	This paper introduces a forecasting-driven, incentive-aware service provisioning framework for distributed air--ground integrated networks with human--machine coexistence. Agent pairs (APs), each comprising a vehicle and its carried uncrewed aerial vehicles (UAVs), are proactively dispatched to overloaded hotspots to augment the computing capacity of edge servers (ESs). This design introduces four coupled challenges: uncertain spatio-temporal workloads, coupling between vehicular mobility and UAV capacity, forecast-driven contracting risks, and heterogeneous quality-of-service (QoS) requirements of human users (HUs) and machine users (MUs). To address these challenges, we propose FUSION, a two-stage framework with offline service preparation and online task scheduling. In the offline stage, a liquid neural network forecasts multi-step ES demand, an enhanced ant colony optimization scheme constructs AP service routes, and an auction-based mechanism establishes ES--AP contracts. In the online stage, we formulate congestion-aware scheduling as an exact-potential game among service demanders (SDs) and develop a potential-guided best-response dynamics algorithm. For a fixed online state, the algorithm converges to an $\varepsilon$-Nash equilibrium (NE) under a positive improvement threshold and to a pure-strategy NE when the threshold is zero. Within the considered contracting model, we theoretically establish that the offline mechanism satisfies individual rationality, near-truthfulness, and weak budget balance. Experiments on synthetic data and real-world load traces show that FUSION achieves higher social welfare while maintaining interaction delay and signaling energy overheads comparable to the considered benchmarks.
\end{abstract}
	
	\begin{IEEEkeywords}
Distributed edge computing, liquid neural network, auction, potential games, heterogeneous network users.
	\end{IEEEkeywords}
}

\maketitle
\IEEEdisplaynontitleabstractindextext

\IEEEpeerreviewmaketitle

\section{Introduction}\label{sec:Intro}

\IEEEPARstart{R}{e}cent years have witnessed a new generation of intelligent applications such as real-time distributed video analytics, intelligent transportation systems, and large-scale mobile crowdsensing, flourishing across smart cities \cite{Survey-edge1,Survey-edge2}. Although edge servers (ESs) can provide computing capabilities for these intelligent applications at the network edge, their limited local resources make them susceptible to resource shortages and task congestion during peak demand periods \cite{Survey-edge3,Survey-edge4,Survey-edge5}. {To enhance the capabilities of ESs, low-altitude wireless and edge networking enabled by uncrewed aerial vehicles (UAVs) has received significant attention~\cite{UAV-edge1,UAV-edge2,UAV-edge3,UAV-edge4}. Unlike static ESs, UAVs offer high mobility and flexibility, enabling rapid deployment to hotspot regions. Existing UAV-assisted edge frameworks mainly optimize post-deployment operations, such as trajectory and resource scheduling, inference offloading, data uploading, and autonomous navigation~\cite{RW-UAV-Li,RW-UAV-Xu,RW-UAV-Zeng,RW-UAV-Wang}. However, how UAVs are transported from deployment bases to geographically distributed hotspots before service provision remains largely unexplored. To fill this gap, we propose an \emph{air--ground cooperative edge computing paradigm}: during periods of intensive demand, ground vehicles transport UAVs to the vicinity of overloaded ESs, with each vehicle and its carried UAVs collectively referred to as an ``agent pair (AP).''}
UAVs are then launched from strategically selected locations to provide temporary computing capacity to the ESs. Operating this hybrid AP--ES deployment paradigm introduces four main challenges:

\noindent\textbf{(A) Spatio-temporal demand uncertainty.}
Edge computing demands exhibit strong spatio-temporal variability, and the coexistence of \underline{h}uman \underline{u}sers (HUs) and \underline{m}achine \underline{u}sers (MUs) further complicates the management of distributed workloads. This variability makes it difficult to determine when, where, and how UAV services should be dispatched. In such networks, static provisioning or purely reactive scheduling often leads to mismatches between supply and demand, resulting in congestion at emerging hotspots and resource underutilization in low-demand regions~\cite{RW-PG-Zhou,RW-PG-Ding}. The key challenge therefore lies in \textit{accurately forecasting spatio-temporal demand patterns from historical requests and contextual information}, so as to enable proactive UAV deployment.

\noindent\textbf{(B) Spatio-temporal coupling of vehicular mobility and UAV deployment.}
Supporting geo-distributed ESs requires APs to rely on vehicular mobility to reach target locations, which is constrained by road topology and traffic conditions. Moreover, the number of UAVs that can be deployed by a vehicle is tightly coupled with its visiting sequence and dwell time at each ES, since UAV dispatching, service execution, and recovery jointly consume time and onboard resources. Consequently, vehicle--UAV joint deployment becomes \textit{a highly coupled optimization problem that must simultaneously account for travel costs, service time windows, UAV dispatching logistics, and spatially distributed service rewards.}

\noindent\textbf{(C) Predictive contracting and incentive-compatible trading.}
In the ES--AP hybrid network, APs operated by private entities must be compensated by ES operators through \textit{advance contracts} that commit service capacity and pricing for future peak-demand periods before demand realization. However, forecasting errors and dynamic supply–demand variations introduce significant risks for both parties: ESs may suffer unmet demand, while APs may incur revenue loss due to resource misallocation. The central challenge is thus to design \textit{a reliable offline trading mechanism that is individually rational, incentive compatible, and budget feasible, while operating prior to demand realization and service execution}~\cite{AUC-survey}.

\noindent\textbf{(D) Online network coordination under human--machine coexistence.}
HUs and MUs compete for shared resources provided by local ESs and leased UAVs, yet exhibit fundamentally different QoS requirements. HUs are typically delay-sensitive due to interactive applications, whereas MUs emphasize timely task completion before deadlines. The resulting challenge is to design \textit{fair and efficient distributed resource scheduling mechanisms that jointly serve HUs and MUs while ensuring stable operation and provable convergence guarantees.}

To address these challenges, we develop a \underline{f}orecasting-driven agent \underline{s}ched\underline{u}ling and \underline{i}ncentive optimizati\underline{on} (FUSION) framework for distributed service provisioning. For challenge (A), we develop \underline{pro}active spatio-temporal demand prediction via \underline{l}iquid \underline{n}eural \underline{n}etworks (\emph{Pro-LNN}) to transform historical service traces and contextual features into multi-step workload forecasts. For challenge (B), we design \underline{e}nhanced \underline{a}nt \underline{c}olony \underline{o}ptimization for \underline{v}ehicular \underline{r}oute \underline{p}lanning (\emph{eACO-VRP}) to construct AP service routes based on the predicted demand. For challenge (C), we introduce \underline{off}line \underline{a}uction-based \underline{i}ncentive-\underline{c}ompatible \underline{c}ontracting (\emph{Off-AIC$^2$}) to establish ES--AP contracts and determine the corresponding payments and rewards. For challenge (D), we formulate online task scheduling as a potential game among service demanders (SDs) and develop \underline{p}otential \underline{g}ame-based \underline{b}est-\underline{r}esponse \underline{d}ynamics (\emph{PG-BRD}). Its asynchronous payoff-improving updates lead to an $\varepsilon$-Nash equilibrium (NE), or a pure-strategy NE when $\varepsilon=0$.
Key contributions are:
\begin{itemize}[leftmargin=3.5mm]
    \item We develop an air--ground integrated edge network model that jointly captures vehicular--UAV APs, ESs, and heterogeneous human--machine computing demands.

    \item We formulate ES--AP--user service provisioning under spatio-temporal mobility constraints, stochastic workload evolution, and strategic agent behaviors. The offline ES--AP contract design leads to a combinatorial mixed-integer formulation, whereas the online SD-side interactions over heterogeneous ES/UAV resources form a finite potential game with a combinatorial action space. We develop a two-stage solution consisting of a forecast-driven offline stage and a game-theoretic online stage.

    \item We develop FUSION to couple proactive offline service preparation with adaptive online scheduling. {For the offline stage, Pro-LNN performs multi-step spatio-temporal demand forecasting, eACO-VRP constructs executable vehicle--UAV routes, and Off-AIC$^2$ establishes forecast-driven ES--AP contracts with individual rationality, near-truthfulness, and weak budget balance under the stated assumptions. For the online stage, we formulate real-time task scheduling as an SD-side potential game under human--machine coexistence and design PG-BRD to obtain congestion-aware SD--SP assignments with convergence to an $\varepsilon$-NE, which reduces to a pure-strategy NE when $\varepsilon=0$.}

    \item Evaluations on synthetic data and real-world load traces show that FUSION achieves higher social welfare while maintaining interaction delay and signaling energy overheads comparable to the considered baselines.
\end{itemize}

\section{Literature Investigation}

\subsection{Service Provision in UAV-Assisted Edge Networks} 
{There exists an extensive body of literature on improving the communication/computation performance and reducing the energy footprint of UAV-assisted edge computing systems~\cite{RW-UAV-Li,RW-UAV-Wang,RW-UAV-Xu,RW-UAV-Zeng,RW-TJCCT}.
For instance, \textit{Wang et al.}~\cite{RW-UAV-Wang} designed a multi-agent deep reinforcement learning (DRL)-based resource scheduling scheme for UAV-assisted wireless edge networks; 
\textit{Xu et al.}~\cite{RW-UAV-Xu} studied the minimization of average inference delay in UAV-assisted computing systems via DRL-based resource allocation;
\textit{Li et al.}~\cite{RW-UAV-Li} investigated distributed/federated edge intelligence for UAV-assisted networks while optimizing energy consumption and resource utilization; 
and \textit{Zeng et al.}~\cite{RW-UAV-Zeng} introduced an ES-assisted UAV navigation system that integrates DRL-based decision making with network-aware containerized resource allocation to enhance navigation quality.}
{More recently, \textit{Sun et al.}~\cite{RW-TJCCT} developed a two-timescale framework that incorporates price incentive-based service provisioning with task offloading and UAV trajectory/resource optimization.} {Most existing UAV-assisted edge frameworks begin with a pre-established UAV infrastructure and focus on optimizing its subsequent operation, such as online scheduling, trajectory control, and edge service provisioning. In contrast, the preceding resource provisioning stage, covering vehicle-carried UAV deployment and advance ES--AP contracting before demand realization, remains largely overlooked. Consequently, the four interrelated challenges identified in Sec.~\ref{sec:Intro} have not been jointly addressed.}

\begin{table}[t!]
	{\footnotesize
		\caption{\footnotesize{A summary of related studies\\(ST: Spatio-Temporal, Pred.: Predictive, Het.: Heterogeneous)}\label{Tab. RW}}
		\begin{center}
			\setlength{\tabcolsep}{0.1mm}{
				\begin{tabular}{|c|c|c|c|c|c|c|c|c|}
					\hline
					\multirow{2}{*}{\textbf{Reference}} 
					& \multicolumn{3}{c|}{\makecell[c]{\textbf{Network}\\ \textbf{Characteristics}}} 
					& \multicolumn{3}{c|}{\textbf{Service Delivery}} 
					& \multicolumn{2}{c|}{\textbf{Participant}} \\
					\cline{2-9}
					& Static & Dynamic & ST  
					& Offline & Online & Pred. 
					& Het. users & Het. ESs \\ 
					\hline
					
					\makecell[c]{\cite{RW-UAV-Wang}}
					&  & $\surd$ & $\surd$
					&  & $\surd$ &  
					&  &  \\
					\hline
					
					\makecell[c]{\cite{RW-UAV-Xu}}
					&  & $\surd$ & 
					&  & $\surd$ &  
					& $\surd$ &  \\
					\hline
					
					\makecell[c]{\cite{RW-UAV-Li}}
					&  & $\surd$ & 
					& $\surd$ &  & 
					&  & $\surd$ \\
					\hline
					
					\makecell[c]{\cite{RW-UAV-Zeng}}
					&  & $\surd$ & 
					&  & $\surd$ & 
					&  &  \\
					\hline
					{\makecell[c]{\cite{RW-TJCCT}}} & & {$\surd$} & {$\surd$} & & {$\surd$} & & {$\surd$} & {$\surd$} \\ \hline
					
					\makecell[c]{\cite{RW-PG-Hsieh,RW-PG-Ding,RW-PG-Wu}}
					& $\surd$ &  &
					& $\surd$ &  & 
					&  & $\surd$ \\
					\hline
					
					\cite{RW-PG-Zhou}
					& $\surd$ &  &
					& $\surd$ &  & 
					&  & $\surd$ \\
					\hline
					\cite{RW-Auction-Ren}
					&  & $\surd$ & $\surd$
					& $\surd$ &  & 
					& $\surd$ & $\surd$ \\
					\hline
					
					\cite{RW-Auction-Xu,RW-Auction-Li}
					&  &$\surd$  &
					& &  $\surd$ & 
					& $\surd$ & $\surd$ \\
					\hline

					\cite{RW-Auction-Lu}
					&  & $\surd$ &
					& $\surd$ &  & 
					& $\surd$ &  \\
					\hline

					\rowcolor[gray]{0.9}
					\textbf{This Work}
					&  & $\surd$ & $\surd$
					& $\surd$ & $\surd$ & $\surd$
					& $\surd$ & $\surd$ \\
					\hline
			\end{tabular}}
	\end{center}}
\end{table}

\subsection{Auction-Driven and Game-Theoretic  Service Provisioning}
Auctions have been widely used for resource provisioning from service providers to users.
For instance, \textit{Ren et al.}~\cite{RW-Auction-Ren} investigated truthful auctions for dependent task scheduling over dynamic edge networks; 
\textit{Xu et al.}~\cite{RW-Auction-Xu} designed combinatorial auction-based deep neural network (DNN) inference for heterogeneous edge-cloud networks to maximize social welfare;
\textit{Li et al.}~\cite{RW-Auction-Li} developed an online auction for generative content scheduling;
\textit{Lu et al.}~\cite{RW-Auction-Lu} proposed auction-based distributed/federated learning to improve the distributed model training speed and energy consumption in edge networks. While the above studies provide valuable insights, they are mostly \textit{reactive} (i.e., decisions are made after demand realization), and typically assume fixed resources and stable demand. In contrast, UAV-assisted edge computing involves dynamic resources and stringent delay/energy constraints, calling for spatio-temporally aware auctions with minimal online interactions between service providers and end users while accounting for uncertainties in both resource availability and demand. This gap is addressed by our Off-AIC$^2$ framework. Beyond auctions, potential games have also been explored as distributed optimization strategies for service provisioning\cite{RW-PG-Ding,RW-PG-Zhou}. For example, 
\textit{Ding et al.}~\cite{RW-PG-Ding} investigated task processing across user--edge--cloud networks  as potential games;
\textit{Zhou et al.}~\cite{RW-PG-Zhou} modeled computation offloading as a potential game with latency--energy trade-offs;
\textit{Hsieh et al.}~\cite{RW-PG-Hsieh} proposed a knapsack potential game for task scheduling;
\textit{Wu et al.}~\cite{RW-PG-Wu} studied task scheduling for multi-cell edge--cloud networks through potential games with NE guarantees.
Despite their notable contributions, the above works typically consider a single/homogeneous user type, overlooking the fundamentally different QoS requirements of HUs and MUs, as well as their interactions with heterogeneous ESs. Our PG-BRD design considers potential-game-based scheduling in UAV-assisted networks with human--machine coexistence, spatio-temporal task windows, and heterogeneous service resources. Table~\ref{Tab. RW} summarizes the differences between FUSION and representative existing studies.

\section{Overview and Modeling}
  \begin{figure}[t!]
	\centering
	\includegraphics[width=1\columnwidth]{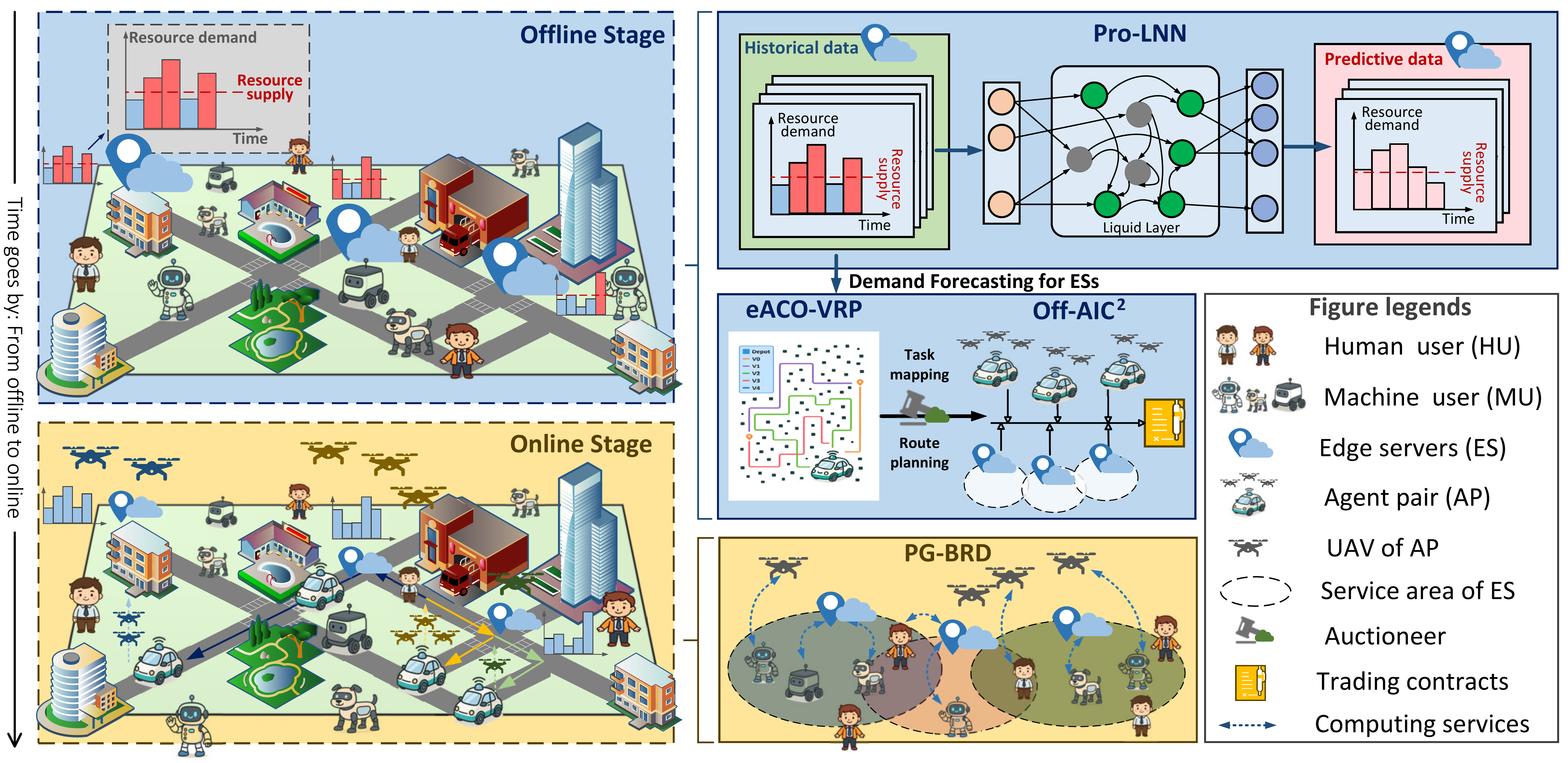}
	\caption{Schematic of FUSION over an air--ground integrated network.}
	\label{fig:system}
	
\end{figure}
\subsection{Overview of Our Methodology}\label{sec:system_model}
We consider an air--ground integrated network, wherein nodes are partitioned into \textit{service demanders} and \textit{service providers}.

\noindent \textit{(i) Service demanders (SDs):}
SDs are users that generate tasks such as model training/inference\footnote{Each SD generates one computational task per trading round and may execute it locally or offload it to an accessible SP.}. 
To capture different QoS needs under human--machine coexistence, we further classify SDs into two categories:
\begin{itemize}[leftmargin=4mm]
	\item \textbf{HUs}, denoted by $\bm{\mathcal{U}}^{\mathrm{H}}=\{u_1^{\mathrm{H}},\dots,u_{|\bm{\mathcal{U}}^{\mathrm{H}}|}^{\mathrm{H}}\}$, are sensitive to service latency because applications such as AR/VR, interactive gaming, real-time video conferencing, and intelligent driving are interactive and experience-driven. Even modest latency increases can reduce their perceived utility.
	\item \textbf{MUs}, denoted by $\bm{\mathcal{U}}^{\mathrm{M}}=\{u_1^{\mathrm{M}},\dots,u_{|\bm{\mathcal{U}}^{\mathrm{M}}|}^{\mathrm{M}}\}$, correspond to autonomous devices such as robots, mapping vehicles, and sensing units. 
	MUs primarily require their computation tasks to be completed before specified deadlines and are therefore less sensitive to instantaneous latency, provided that task feasibility is preserved (e.g., a robot's image-analysis task remains feasible if its inference result is obtained before the deadline).
\end{itemize}

\noindent \textit{(ii) Service providers (SPs):}
We consider two types of SPs:
\begin{itemize}[leftmargin=4mm]
	\item \textbf{ESs}, denoted by $\bm{\mathcal{S}}=\{s_1,\dots,s_{|\bm{\mathcal{S}}|}\}$, are geo-distributed across the map and provide low-latency services to nearby users. Each ES operates with a limited pool of orthogonal subcarriers and finite resource capacity, inherently constraining the number of users that can be concurrently served. During peak demand periods, hotspot ESs are particularly susceptible to resource saturation, which can trigger substantial queuing delays or even admission blocking.
	\item \textbf{APs}, denoted by $\bm{\mathcal{V}}=\{v_1,\dots,v_{|\bm{\mathcal{V}}|}\}$, are formed by vehicles carrying multiple on-board charging UAVs.
	Due to the limited power of UAVs, vehicles first traverse pre-planned routes to transport UAVs to designated service/ES locations. Upon arriving at their contracted ESs, UAVs are deployed to offer extra computing resources. This procedure incurs operational costs, including vehicle repositioning, UAV dispatching/launching procedures, flight energy consumption, and maintenance overhead.
\end{itemize}

Under this two-sided demand--supply structure, SDs compete for reliable computing services, while SPs collaboratively provide resources subject to their capacity and operational constraints. However, ESs face \textit{uncertain service demand:} request peaks vary stochastically across ESs in time and space, and the HUs-to-MUs composition also changes over time. To capture these dynamics, each service window is discretized into $T$ time slots indexed by $t\in\{0,1,\ldots,T-1\}$. We therefore design a two-stage framework that separates offline planning from online execution, as shown in Fig. \ref{fig:system}. From the ES perspective, the offline stage (Sec.~\ref{SEC:OFFLINE}) precedes each trading round and uses historical records and statistical demand patterns. Since future HUs and MUs are unknown, only aggregate ES-level demand can be predicted. Accordingly, Pro-LNN forecasts future demand, which then supports AP route planning, UAV deployment, and auction-based ES--AP contracting through eACO-VRP and Off-AIC$^{2}$. The online stage (Sec.~\ref{SEC:ONLINE}) begins when actual tasks arrive and SD identities and locations become known. Given the active participants and their instantaneous requirements, real-time task--resource scheduling is performed via an exact-potential game on the SD side. Each SD updates its scheduling decision based on a potential-aligned payoff, while the resulting assignment profile implicitly determines SP-side admissions, served sets, congestion levels, and WPS-based service assignment under the corresponding coverage and capacity constraints. The induced game guarantees finite best-response convergence to an $\varepsilon$-NE under a positive improvement threshold, and to a pure-strategy Nash equilibrium when the threshold is zero. {Thus, FUSION uses historical demand, AP mobility/resource states, and reported bids and asks for offline route planning, capacity reservation, and ES--AP contracting, while PG-BRD adapts SD--SP assignments online based on actual task arrivals, channel conditions, SD locations, and congestion states. Appx.~G further compares FUSION with small-scale full-information benchmarks to quantify the performance impact of this two-stage decomposition.}

\subsection{Modeling of SDs, SPs, and Contracts}
Hereafter, we provide detailed modeling of SPs and SDs.

\noindent $\bullet$ \textit{Modeling of SDs.}
We model HUs and MUs as follows:

\noindent\textbf{\textit{(i)} HUs:}
Each HU $u_m^{\mathrm{H}} \in \bm{\mathcal{U}}^{\mathrm{H}}$ is modeled by a 5-tuple
$u_m^{\mathrm{H}} = 
\{ e_m^{\mathrm{H,tx}},\, r_m^{\mathrm{H}},\, d_m^{\mathrm{H}},\, f_m^{\mathrm{H,loc}}, \, \bm{C}_m^{\mathrm{H}} \}$,
where $e_m^{\mathrm{H,tx}}$ is its uplink transmit power (Watts);
$r_m^{\mathrm{H}}$ is its required computing workload (CPU cycles);
$d_m^{\mathrm{H}}$ is its input data size (bits); 
$f_m^{\mathrm{H,loc}}$ is its local CPU frequency;
and $\bm{C}_m^{\mathrm{H}}$ represents the set of reachable ESs within the communication coverage of $u_m^{\mathrm{H}}$.

\noindent\textbf{\textit{(ii)} MUs:}
Each MU $u_n^{\mathrm{M}} \in \bm{\mathcal{U}}^{\mathrm{M}}$ is represented by a 6-tuple
$u_n^{\mathrm{M}} = 
\{ \tau_n^{\mathrm{ddl}},\, e_n^{\mathrm{M,tx}},\, r_n^{\mathrm{M}},\, d_n^{\mathrm{M}},\, f_n^{\mathrm{M,loc}},\, \bm{C}_n^{\mathrm{M}} \}$,
where $\tau_n^{\mathrm{ddl}}$ is its task deadline time;
$e_n^{\mathrm{M,tx}}$ is its transmit power; 
$r_n^{\mathrm{M}}$ and $d_n^{\mathrm{M}}$ denote its workload and data size, respectively. Also,
$f_n^{\mathrm{M,loc}}$ denotes its CPU capability,
and $\bm{C}_n^{\mathrm{M}}$ captures SPs that can serve $u_n^{\mathrm{M}}$.

\noindent $\bullet$ \textit{Modeling of SPs.}
We model ESs and APs as follows:

\noindent\textbf{\textit{(i)} ESs:}
{Each ES $s_i \in \bm{\mathcal{S}}$ is represented by a 5-tuple:
$s_i = 
\{t_i^{\mathrm{E,b}}, f_i^{\mathrm{E}},\, e_i^{\mathrm{E}},\, G_i,\, l_i^{\mathrm{E}} \}$,
where $t_i^{\mathrm{E,b}}$ denotes the time-slot at which $s_i$ begins to request resource supply to mitigate the shortage of edge resources within its own coverage area;
$f_i^{\mathrm{E}}$ is its computing capacity (e.g., CPU cycles/s);
$e_i^{\mathrm{E}}$ is its local power consumption (Watts); $l_i^{\mathrm{E}} = (x_i^{\mathrm{E}}, y_i^{\mathrm{E}}, h_i^{\mathrm{E}})$ is its 3D location.
In addition, each ES $s_i$ is presumed to have a set of orthogonal subcarriers that can serve at most $G_i$ SDs simultaneously.}

\noindent\textbf{\textit{(ii)} APs:}
{Each AP $v_k\in \bm{\mathcal{V}}$ is modeled by a 7-tuple:
$v_k =
\{ C_k^{\mathrm{A}},\, f_k^{\mathrm{A}},\, e_k^{\mathrm{comp}},\, e_k^{\mathrm{move}},\, e_k^{\mathrm{fly}},\, G_k^\mathrm{A},\, l_k^{\mathrm{A}} \langle t\rangle \}$,
where $C_k^{\mathrm{A}}$ denotes the maximum number of UAVs carried by the vehicle; $f_k^{\mathrm{A}}$ captures the computing capability of each mounted UAV; $e_k^{\mathrm{move}}$, $e_k^{\mathrm{comp}}$, and $e_k^{\mathrm{fly}}$ denote the power consumption of ground movement, UAV computation, and UAV flight, respectively; $G_k^\mathrm{A}$ denotes the maximum number of SDs that each UAV of $v_k$ can serve through its limited orthogonal subcarriers; $l_k^{\mathrm{A}}  \langle t\rangle{=}(x_k^{\mathrm{A}}\langle t\rangle, y_k^{\mathrm{A}}\langle t\rangle)$ is the ground location of the vehicle at time $t$.} We further define the total online service capacity of AP $v_k$ as $Q_k^{\mathrm{A}} = C_k^{\mathrm{A}}G_k^{\mathrm{A}}$, which denotes the maximum number of SD service units that can be supported by all UAVs carried by AP $v_k$ during the considered online service window.

\noindent $\bullet$ \textit{Modeling of Offline Contracts.} {In the offline stage, a trading contract is established between each ES $s_i$ and AP $v_k$, denoted as
$	\mathbb{C}^{\mathrm{E}\leftrightarrow\mathrm{A}}_{i,k}
	= \{N^\mathrm{Trad}_{i,k}, p_{i,k}^{\mathrm{ES}},
	r_{i,k}^{\mathrm{ES}}\}$, where $N^\mathrm{Trad}_{i,k}$ represents the quantity of traded resources, 
$p_{i,k}^{\mathrm{ES}}$ indicates unit payment (per task) from $s_i$ to contractual AP $v_k$, and 
$r_{i,k}^{\mathrm{ES}}$ refers to the unit reward (per task) of $v_k$ for offering services. }

\section{Offline Stage: Forecasting-Driven Contracting and AP Deployment} \label{SEC:OFFLINE}
{The offline stage converts predicted spatio-temporal demand into route-feasible and incentive-aware AP service plans.}
It comprises three components: \textit{(i)} Pro-LNN, which forecasts workloads across distributed ESs; \textit{(ii)} eACO-VRP, which constructs AP routes subject to service-window and capacity constraints; and \textit{(iii)} Off-AIC$^2$, which determines winning ES--AP contracts together with their service quantities, payments, and rewards. We let $x^\mathrm{Off}_{i,k}$ denote the assignment/matching between ESs and APs, where $x^\mathrm{Off}_{i,k}=1$ indicates that ES $s_i$ successfully secures the service of AP $v_k$ and thereby signs an offline contract, and $x^\mathrm{Off}_{i,k}=0$ otherwise. For notational simplicity, we define $\bm{X}^\mathrm{Off}=\big\{x^\mathrm{Off}_{i,k} \mid s_i \in \bm{\mathcal{S}}, v_k \in \bm{\mathcal{V}}\big\}$ as the assignment profile.

\subsection{Utilities of ESs, APs, and the Auctioneer}
\subsubsection{Utility of ESs and APs}
In the offline stage, as delivering edge services to users and leasing computing capacity can generate both costs and revenues for ESs, we model the utility of each ES $s_i$ as its net profit as follows:
\begin{equation}\label{equ. U_ES_1}
	U^\mathrm{ES,Off}_{i} = \sum_{v_k\in\bm{\mathcal{V}}}x^\mathrm{Off}_{i,k}N_{i,k}^\mathrm{Trad}\left(\overline{p}_{i}-p^\mathrm{ES}_{i,k}\right) ,
\end{equation}
where $\overline{p}_{i}$ denotes the historical average payment earned by $s_i$ for providing edge services to users.

The utility of each AP is modeled as the monetary reward obtained from serving contracted ESs minus its mobility, flight, computation, and maintenance costs. 
For AP $v_k$ serving ES $s_i$, the vehicle moving cost and UAV flight cost are
$c^\mathrm{move}_{i,k}=\tau^\mathrm{move}_{i,k}e^\mathrm{move}_k$ and
$c^\mathrm{fly}_{i,k}=\tau^\mathrm{fly}_{i,k}e^\mathrm{fly}_k$, respectively.
The unit computation cost of processing one task is $c^{\mathrm{comp}}_{i,k}=\frac{\overline r\, e_k^\mathrm{comp}}{f_k^\mathrm{A}} $,
where $\overline r$ is the historical average computational workload per task. 
Therefore, the offline utility of AP $v_k$ is
\begin{equation}\label{equ. SEC_U}
{\small	\begin{aligned}
		U^\mathrm{A,Off}_{k}
		&=
		\sum_{s_i\in\bm{\mathcal{S}}}x^\mathrm{Off}_{i,k}
		\Big[
		N_{i,k}^\mathrm{Trad} r^\mathrm{ES}_{i,k}
		-\omega_1\big(
		c^\mathrm{move}_{i,k}
		+c^\mathrm{fly}_{i,k}
		+N_{i,k}^\mathrm{Trad}c^{\mathrm{comp}}_{i,k}
		\big)
		\Big]  \\
		&-\omega_2 c_k^\mathrm{hard}
		\mathbbm{1}_{\{\sum_{s_i\in\bm{\mathcal{S}}}x^\mathrm{Off}_{i,k}>0\}},
	\end{aligned}}
\end{equation}
where $\mathbbm{1}_{\{\cdot\}}$ denotes the binary indicator function, $\omega_1$ converts energy consumption into monetary cost and $\omega_2 c_k^\mathrm{hard}$ captures the maintenance/depreciation cost incurred when AP $v_k$ participates in service provisioning.

\subsubsection{Utility of the Auctioneer}
{
	An auctioneer (e.g., a trusted third party) collects bids and asks, determines winning ES--AP associations, and establishes the corresponding trading contracts.
} We define the utility of the auctioneer as the difference between the payments collected from ESs and the rewards paid to APs:
\begin{equation}\label{key}
	\begin{aligned}
		U^\mathrm{Auc,Off}=
		\sum_{s_i\in \bm{\mathcal{S}}}\sum_{v_k\in \bm{\mathcal{V}}}x^\mathrm{Off}_{i,k}N^\mathrm{{Trad}}_{i,k} \left(p^\mathrm{{ES}}_{i,k}-r_{i,k}^\mathrm{ES}\right).
	\end{aligned}
\end{equation}

\subsection{Problem of Interest in Offline Stage of FUSION}
The objective of the offline stage is to maximize the social welfare (SW), defined as the sum of the utilities of participating ESs, APs, and the auctioneer:
\begin{equation}\label{equ. SW}
	\begin{aligned}
			U^\mathrm{SW}
			=
			\sum_{s_i\in \bm{\mathcal{S}}}U^\mathrm{ES,Off}_{i}
			+
			\sum_{v_k\in \bm{\mathcal{V}}}U^\mathrm{A,Off}_{k}
			+
			U^\mathrm{Auc,Off}.
		\end{aligned}
\end{equation}

	We formulate the joint selection of winning ES--AP pairs and their contracts as an offline optimization problem $\bm{\mathcal{F}}^{(\mathrm{I})}$:
\begin{subequations}
	\begin{align}
		\bm{\mathcal{F}}^{(\mathrm{I})}:~&{\max_{\bm{X}^\mathrm{Off},\,\{\mathbb{C}_{i,k}^{\mathrm{E}\leftrightarrow\mathrm{A}}\}}
			~~U^\mathrm{SW}}\tag{5}\label{PF_1}\\
		\text{s.t.}~~~
		&\sum_{v_k\in\bm{\mathcal{V}}} x^\mathrm{Off}_{i,k}\le1,~~~ \forall s_i\in\bm{\mathcal{S}}\tag{5a}\label{9b}\\
		&x^\mathrm{Off}_{i,k}\in\left\{0,1\right\},~~~\forall s_i\in\bm{\mathcal{S}},~\forall v_k\in\bm{\mathcal{V}}\tag{5b}\label{9c}\\
&\sum_{s_i\in\bm{\mathcal{S}}}x^\mathrm{Off}_{i,k}N^\mathrm{Trad}_{i,k}
\le Q_k^{\mathrm{A}}, 
~~~\forall v_k\in\bm{\mathcal{V}}.
\tag{5c}\label{9d}
	\end{align}
\end{subequations}
\setcounter{equation}{5}

\noindent where~(\ref{9b}) ensures that each ES is matched with at most one AP, (\ref{9c}) enforces the binary nature of $x^\mathrm{Off}_{i,k}$, and (\ref{9d}) guarantees that the promised resources in signed contracts do not exceed the supply of each AP $v_k$. {
	Problem $\bm{\mathcal{F}}^{(\mathrm{I})}$ represents a combinatorial mixed-integer optimization problem~\cite{AUCTION,NP} involving binary ES--AP assignment variables, continuous contract variables, and route-feasibility constraints. The problem becomes computationally challenging because ES--AP matching, service quantities, prices, and AP routes are tightly coupled.}
Because ES--AP contracts are established before actual task arrivals, the offline stage requires an estimate of future resource demand. We use Pro-LNN to forecast ES-side spatio-temporal demand, eACO-VRP to construct AP service routes from these forecasts, and Off-AIC$^2$ to determine the final ES--AP associations and contracts.

\subsection{Forecasting, Agent Deployment, and Contracting}
We next present the offline stage, which couples demand forecasting (Pro-LNN), AP deployment (eACO-VRP), and long-term contracting (Off-AIC$^2$).

\noindent $\bullet$ At the beginning of each offline trading round, each ES runs Pro-LNN to predict its service demand over future rounds and submits the resulting demand profile as bid information. Each AP reports its available service capacity and ask price based on its computing capability, energy consumption, and service duration.

\noindent $\bullet$ {Given the bids, asks, predicted ES demands, and AP resource states, the auction platform first invokes eACO-VRP to construct route-feasible AP service plans, specifying the ESs each AP can visit and their visiting order. Off-AIC$^2$ then determines the winning ES--AP contracts and applies its critical-price-based rules to compute the corresponding payments and rewards.}

\noindent $\bullet$ The offline outcome consists of the ES--AP matching relations, payment vectors, and pre-planned service routes for each AP, which are passed to the online stage over the subsequent time horizon.

This ordering ensures that incentive-aware contracting is performed only over physically executable ES--AP service plans, while resulting routes, contracts, and capacities support online scheduling via PG-BRD\footnote{A detailed discussion on the necessity and complementarity of the FUSION components is provided in Appx.~C.} We describe three modules: Pro-LNN-based demand forecasting, eACO-VRP-based AP deployment, and the Off-AIC$^2$ contracting mechanism.

\subsubsection{Design of Proactive Spatio-Temporal Demand Prediction via Liquid Neural Networks (Pro-LNN)} 
Pro-LNN uses a liquid neural network (LNN)~\cite{LNN} with neuron-wise time constants to model ES-side demand across trading rounds $\nu=1,2,\ldots$\footnote{Here, $\nu$ indexes the offline trading round, and each offline trading round is followed by an online service window consisting of $T$ online execution slots.}. The time constants provide different memory scales for capturing short- and long-term demand variations. Demand observations may also be missing or noisy. The implementation considered here uses a fixed interval $\Delta\nu$ between trading rounds, as detailed in Appx.~B.1.

Let $N_{i}^{\mathrm{Dem},(\nu)}$ denote the task demand of ES $s_i$ at the $\nu$-th trading round. We form an input vector
$\mathbf{z}_i^{(\nu)}=\big[N_{i}^{\mathrm{Dem},(\nu)},\mathbf{x}_i^{(\nu)}\big]$,
where $\mathbf{x}_i^{(\nu)}$ stacks optional exogenous features at round $\nu$, such as time-of-day and day-of-week indicators, ES-type and capacity descriptors, and local mobility/load statistics extracted from recent logs. Given a length-$L$ history of past observations
$\{\mathbf{z}_{i}^{(\nu-L+1)},\ldots,\mathbf{z}_{i}^{(\nu)}\}$, Pro-LNN encodes this sequence via a liquid cell with neuron-wise time constants, integrates the measurements into a hidden state $\mathbf{h}_i^{(\nu)}$, and then applies a shallow readout network to produce an $H$-step forecast
$\{\hat N_{i}^{\mathrm{Dem},(\nu+1)},\ldots,\hat N_{i}^{\mathrm{Dem},(\nu+H)}\}$.
Pro-LNN is trained with an $H$-step mean squared error (MSE) loss between the forecasts and ground-truth demands. A multi-step consistency term penalizes prediction errors during closed-loop rollouts, as detailed in Appx.~B.1.
These forecasts are computed offline before each trading round and then fed into the eACO-VRP module for ES recommendation and AP route pre-planning, as well as the auction module for ES--AP contract design. The update equations and training pseudocode are provided in Appx.~B.1.

\subsubsection{Enhanced Ant Colony Optimization for Vehicle Route Planning (eACO-VRP)}\label{sec:eACO-VRP}
After Pro-LNN predicts the upcoming service demand, FUSION uses eACO-VRP to determine whether APs can reach and serve the predicted overloaded ESs. Specifically, each AP selects and orders ESs under service-window, mobility, UAV-capacity, and travel/service-cost constraints, yielding a constrained vehicle-routing problem~\cite{ACO,ACO1}. Executed before each online service window, eACO-VRP constructs route-feasible ES--AP candidates from predicted demand and AP states, which are then passed to Off-AIC$^2$ for contract, payment, and reward determination. Online variations in HU/MU arrivals, channel states, SD locations, and congestion are handled by PG-BRD\footnote{Further discussions on the responsiveness, scalability, and robustness of eACO-VRP under demand drift and topology perturbations are in Appx.~B.3.}. Based on this design, for a given \(v_k\), we model \(v_k\) and all ESs as vertices of a directed complete graph \(\bm{\mathcal{G}}_k(\bm{\mathcal{N}}_k,\bm{\mathcal{E}}_k)\). Here, $\bm{\mathcal{N}}_k=\{\mathrm{n}_0,\mathrm{n}_1,\ldots,\mathrm{n}_{|\bm{\mathcal{S}}|}\}$ denotes the vertex set, where \(\mathrm{n}_0\) represents \(v_k\) and \(\mathrm{n}_1,\ldots,\mathrm{n}_{|\bm{\mathcal{S}}|}\) denote ESs; while the edge set is
$\bm{\mathcal{E}}_k=\{(\mathrm{n}_r,\mathrm{n}_s)\mid \mathrm{n}_r,\mathrm{n}_s\in\bm{\mathcal{N}}_k,\ r\neq s\}$.
Let the initial location of \(\mathrm{n}_0\) be \(l^{\mathrm{v}}_{0}\langle 0\rangle\), and the fixed locations of \(\mathrm{n}_1,\ldots,\mathrm{n}_{|\bm{\mathcal{S}}|}\) be \(l^{\mathrm{s}}_1,\ldots,l^{\mathrm{s}}_{|\bm{\mathcal{S}}|}\).
We use \(\mathbb{K}\) to denote the ant set: each ant \(\mathbbm{k}\in\mathbb{K}\) simulates the trajectory of \(v_k\), starting from \(l^{\mathrm{v}}_{0}\langle 0\rangle\). At each step, it evaluates current time-slot $t_{\mathrm{cur}}$, travel time to each candidate ES, ES service-window beginning time, remaining UAV/service capacity, the route utility upon arrival.

{For a candidate transition from $\mathrm{n}_r$ to $\mathrm{n}_s$, let $a_s=t_{\mathrm{cur}}+\tau^{\mathrm{move}}_{r,s} $ be the predicted arrival time and define the service-window slack as $\mathrm{slack}_s=t_s^{\mathrm{E,b}}-a_s $. A candidate ES is feasible only if $\mathrm{slack}_s\ge0$ and AP $v_k$ still has sufficient remaining UAV/service capacity. We define the distance-based heuristic desirability as $\eta_{r,s}=\frac{1}{d_{r,s}+\epsilon_d} $, where $d_{r,s}$ is the travel distance between $\mathrm{n}_r$ and $\mathrm{n}_s$ and $\epsilon_d>0$ avoids division by zero. The state-transition rule, i.e., the probability that ant \(\mathbbm{k}\) moves from \(\mathrm{n}_r\) to \(\mathrm{n}_s\) (\(r\neq s,\ s\neq 0\)), is
\begin{equation}\label{Pro}
	\resizebox{0.93\linewidth}{!}{$
		pr^{\mathbbm{k}}_{r,s}=
		\begin{cases}
			\displaystyle
			\frac{\varphi^{\varepsilon_1}_{r,s}\,\eta^{\varepsilon_2}_{r,s}\,
				(\mathrm{slack}_s+\epsilon_t)^{-\varepsilon_3}}
			{\sum\limits_{q\in J^{\mathbbm{k}}(\mathrm{n}_r)}
				\varphi^{\varepsilon_1}_{r,q}\,\eta^{\varepsilon_2}_{r,q}\,
				(\mathrm{slack}_q+\epsilon_t)^{-\varepsilon_3}}
			& \text{if } \mathrm{n}_s \in J^{\mathbbm{k}}(\mathrm{n}_r),\\[12pt]
			0 & \text{otherwise,}
		\end{cases}
		$}
\end{equation}
where \(\varepsilon_1,\varepsilon_2,\varepsilon_3>0\) are weights and $\epsilon_t>0$ avoids division by zero. Here, \(\varphi_{r,s}\) is the pheromone level on edge \((\mathrm{n}_r,\mathrm{n}_s)\), and \(\eta_{r,s}\) is the distance-based heuristic desirability. The factor $(\mathrm{slack}_s+\epsilon_t)^{-\varepsilon_3}$ prioritizes feasible ESs with tighter service windows, while infeasible ESs with negative slack are excluded from $J^{\mathbbm{k}}(\mathrm{n}_r)$.}
In~\eqref{Pro}, $J^{\mathbbm{k}}(\mathrm{n}_r)$ denotes the feasible ES-vertex set available to ant $\mathbbm{k}$ at node $\mathrm{n}_r$: an ES vertex $\mathrm{n}_s$ belongs to $J^{\mathbbm{k}}(\mathrm{n}_r)$ if \textit{(i)} $t_{\mathrm{cur}}+\tau^{\mathrm{move}}_{r,s}\le t_s^{\mathrm{E,b}}$, and \textit{(ii)} AP $v_k$ still has enough remaining UAV/service capacity to serve the ES demand.
Once a route is found by ants, pheromones on edges are updated as:
\begin{equation}\label{pheromone}
	\varphi_{r,s}\ \leftarrow\ (1-\theta)\,\varphi_{r,s} + \Delta\varphi_{r,s},
\end{equation}
where \(0<\theta<1\) is the evaporation coefficient. The pheromone deposit is defined as
\begin{equation}\label{pheromone1}
	\Delta\varphi_{r,s}=
	\begin{cases}
		\chi\,\max\{0,\Delta U^\mathrm{A}_{r,s}\},
		& \text{if }(\mathrm{n}_r,\mathrm{n}_s)\in\bm{M}^{\mathrm{max}},\\
		0,
		& \text{otherwise},
	\end{cases}
\end{equation}
where $\chi>0$ is the pheromone-deposit scaling factor and $\Delta U^\mathrm{A}_{r,s}$ denotes the utility increment for \(v_k\) when moving from \(\mathrm{n}_r\) to serve the ES at \(\mathrm{n}_s\). The best path among ants is selected as $\bm{M}^{\mathrm{max}}\in\arg\max_{\bm{M}^{\mathbbm{k}}:\mathbbm{k}\in\mathbb{K}}\sum_{(\mathrm{n}_a,\mathrm{n}_b)\in \bm{M}^{\mathbbm{k}}}\Delta U^\mathrm{A}_{a,b}. $

For brevity, the pseudo-code of eACO-VRP is provided in Alg.~\ref{Alg2}, while its four main steps are summarized below:

\begin{algorithm}[t]
	{\footnotesize
		\caption{{Proposed eACO-VRP}}\label{Alg2}
		\LinesNumbered
		\textbf{Input:} AP $v_k$ (source node $\mathrm{n}_0$), ES set $\bm{\mathcal{S}}$, directed complete graph $\bm{\mathcal{G}}_k(\bm{\mathcal{N}}_k,\bm{\mathcal{E}}_k)$, ant set $\mathbb{K}$, max iterations $iter_{\mathrm{max}}$ \\
		
		\textbf{Initialization:} $\varphi_{r,s}\!\leftarrow\!\varphi_0$ for all $(\mathrm{n}_r,\mathrm{n}_s)\!\in\!\bm{\mathcal{E}}_k$; $iter\!\leftarrow\!1$ \\
		
		\While{$iter \le iter_{\mathrm{max}}$}{
			\For{each ant $\mathbbm{k} \in \mathbb{K}$}{
				$\bm{M}^{\mathbbm{k}}\!\leftarrow\!\varnothing$;\quad $\mathrm{n}_{r}\!\leftarrow\!\mathrm{n}_0$ \\
				\While{end condition not met}{
					\textbf{Compute:} feasible set $J^{\mathbbm{k}}(\mathrm{n}_{r})$ and transition probabilities $pr^{\mathbbm{k}}_{r,s}$ (via (\ref{Pro})) \\
					\If{$J^{\mathbbm{k}}(\mathrm{n}_{r})\neq\varnothing$}{
						{sample next node $\mathrm{n}_{s} $ according to $ pr^{\mathbbm{k}}_{r,s}$ from $J^{\mathbbm{k}}(\mathrm{n}_{r})$} \\
						$\bm{M}^{\mathbbm{k}}\!\leftarrow\!\bm{M}^{\mathbbm{k}}\cup\{\text{edge}(\mathrm{n}_{r},\mathrm{n}_{s})\}$ \\
						$\mathrm{n}_{r}\!\leftarrow\!\mathrm{n}_{s}$
					}
					\Else{break}
				}
			}
			
			{$\bm{M}^{\mathrm{max}}\!\leftarrow\!\arg\max_{\bm{M}^{\mathbbm{k}}:\mathbbm{k}\in\mathbb{K}}\;\sum_{(\mathrm{n}_r,\mathrm{n}_s)\in\bm{M}^{\mathbbm{k}}}\Delta U^\mathrm{A}_{r,s}$} \\
			\For{each $(\mathrm{n}_r,\mathrm{n}_s)\in\bm{\mathcal{E}}_k$}{Pheromone evaporation and deposit (\ref{pheromone}) and (\ref{pheromone1})
			}
			
			$iter\!\leftarrow\!iter+1$
		}
		
		\textbf{Output:} $\bm{M}^{\mathrm{max}}$ (pre-planned route for $v_k$), and the recommended ES tasks along $\bm{M}^{\mathrm{max}}$
	}
\end{algorithm}

\noindent\textbf{Step 1. Initialization} (lines 1-2, Alg. \ref{Alg2}):
Using the AP/ES information (e.g., geographic locations, the  time-slot when ES $s_i$ begins to request extra resources, AP $v_k$ payments and service costs, and the travel/service time and energy), we construct, for each AP $v_k$ (with source node $\mathrm{n}_0$), a  graph $\bm{\mathcal{G}}_k(\bm{\mathcal{N}}_k,\bm{\mathcal{E}}_k)$ as discussed above.

\noindent\textbf{Step 2. Path exploration} (lines 4-13, Alg. \ref{Alg2}):
At each iteration, an ant $\mathbbm{k}$ located at its current vertex $\mathrm{n}_r$ jointly evaluates the current time-slot, the travel time to each vertex, the service time at that vertex, and the  time-slot when ES $s_i$ begins to request extra resources. This screening step yields the feasible candidate set $J^{\mathbbm{k}}(\mathrm{n}_r)$ and the associated transition probabilities $pr^{\mathbbm{k}}_{r,s}$ as discussed above. The ant then probabilistically selects its next vertex using \eqref{Pro}. 

\noindent\textbf{Step 3. Pheromone update} (lines 14-16, Alg. \ref{Alg2}):
For each constructed path $\bm{M}^{\mathbbm{k}}$, we compute the cumulative utility $\sum_{(\mathrm{n}_r,\mathrm{n}_s)\in\bm{M}^{\mathbbm{k}}}\Delta U^\mathrm{A}_{r,s}$. Among all ants, the path with the largest utility is selected as $\bm{M}^{\mathrm{max}}$, and the pheromone trails on its edges are reinforced according to  \eqref{pheromone} and \eqref{pheromone1}, while pheromones on other edges evaporate.

\noindent\textbf{Step 4. Output} (line 18, Alg. \ref{Alg2}):
After the final iteration, eACO-VRP returns $\bm{M}^{\mathrm{max}}$ as the pre-planned route for AP $v_k$, together with the recommended ES service targets along this route.

\subsubsection{Offline Auction-based Incentive-Compatible Contracting (Off-AIC$^2$)}
Building on the pre-planned AP routes obtained above, we now introduce Off-AIC$^2$, an offline double-auction mechanism that establishes incentive-aware service contracts between ESs and APs. The goal is to decide ES--AP associations, how much service capacity should be reserved along each AP route, and how payments/rewards should be set so that participants (i.e., ESs and APs) have limited incentive to misreport their private values. 

In Off-AIC$^2$, on the \emph{buyer} side, each ES $s_i$ forms its service reservation request by aggregating the $H$-step demand forecast as
$D_i = \sum_{h=1}^{H}\hat N_i^{\mathrm{Dem},(\nu+h)}$. Here, $H$ denotes the forecasting horizon, and $D_i$ is the requested quantity of AP-assisted service units submitted in the current offline auction. ES $s_i$ has a private per-unit valuation $v_i$ for receiving edge service and participates in Off-AIC$^2$ by submitting a per-unit bid $\overline{\mathrm{bid}}_i$ together with its reservation request $D_i$. Here, $\overline{\mathrm{bid}}_i$ represents the maximum payment that $s_i$ is willing to make per unit of service. On the \emph{seller} side, each AP $v_k$, based on its total service capacity $Q_k^\mathrm{A}$ and its service cost, reports a per-unit ask price $\mathrm{ask}_k$ and its capacity $Q_k^\mathrm{A}$, where $\mathrm{ask}_k$ is the minimum acceptable payment per unit of service\footnote{Here, $\mathrm{ask}_k$ denotes AP $v_k$'s route-aware amortized per-unit service cost, covering its expected computation, vehicle movement, UAV flight, and maintenance/depreciation costs under the route-feasible service plan in the current offline trading round.}. Subsequently, for a potential ES--AP pair $(s_i,v_k)$, we define the pairwise surplus as
$	U_{i,k}^{\mathrm{Off}}=\overline{\mathrm{bid}}_i-\mathrm{ask}_k $, where $U_{i,k}^{\mathrm{Off}}>0$ indicates that the ES's per-unit willingness to pay exceeds the AP's per-unit ask, and hence the pair is economically admissible before route-feasibility screening.
Off-AIC$^2$ then uses the reported bids, asks, demands, and capacities to compute pairwise utilities and select a set of winning ESs and APs, determine their matching/association pattern, and compute payments and rewards. The complete pseudo-code 
of Off-AIC$^2$ is provided in Alg.~\ref{Alg3}\footnote{For each active contract $(s_i,v_k)$ with $x_{i,k}^{\mathrm{Off}}=1$, the contract variables are set to
	$p_{i,k}^{\mathrm{ES}} = p_i^{\mathrm{b}\ast}$,
	$r_{i,k}^{\mathrm{ES}} = r_k^{\mathrm{s}\ast}$, and
	$N_{i,k}^{\mathrm{Trad}} = D_i$,
	unless a partial-service contract is explicitly specified. Thus, the settlement vectors from Alg.~\ref{Alg3} directly determine the payment and reward terms in $\mathbb{C}_{i,k}^{\mathrm{E}\leftrightarrow\mathrm{A}}$.}, while its main three steps are summarized below:

\begin{algorithm}[t]
	{\footnotesize
		\caption{Proposed Off-AIC$^2$}\label{Alg3}
		\LinesNumbered
		\textbf{Input:}
		ES set $\bm{\mathcal{S}}$ with reservation requests $\{D_i\}$ and bids $\{\overline{\mathrm{bid}}_i\}$;
		AP set $\bm{\mathcal{V}}$ with asks $\{\mathrm{ask}_k\}$ and total service capacities $\{Q_k^\mathrm{A}\}$
		
		\textbf{Initialization:}
		Set $\bm{X}^\ast\!\leftarrow\!\bm{0}$, $\bm{p}^{\mathrm{b}\ast}\!\leftarrow\!\bm{0}$, $\bm{r}^{\mathrm{s}\ast}\!\leftarrow\!\bm{0}$.
		Set residual capacities $R_k^\mathrm{A}\leftarrow Q_k^\mathrm{A}$ for all APs $v_k$.

		Sort ES indices into $\mathcal{L}^{\mathrm{ES}}$ by $\overline{\mathrm{bid}}_i$ (descending), and AP indices into $\mathcal{L}^{\mathrm{AP}}$ by $\mathrm{ask}_k$ (ascending); relabel $(i,k)$ accordingly.\\
		Find a pair of pivotal ranks $(I^\ast,K^\ast)$ such that
		$	\overline{\mathrm{bid}}_{I^\ast} \ge \mathrm{ask}_{K^\ast}, \quad
		\sum_{k=1}^{K^\ast} Q_k^\mathrm{A} \;\ge\; \sum_{i=1}^{I^\ast} D_i$,
		and $I^\ast$ is maximized (breaking ties by larger capacity slack
		$\sum_{k=1}^{K^\ast} Q_k^\mathrm{A} - \sum_{i=1}^{I^\ast} D_i$).
		
		Initialize the set of unmatched ESs:
		$\mathcal{U}^{\mathrm{free}} \leftarrow \{ s_i \mid 1 \le i \le I^\ast \}$.\\
		\For{$k$ in top-$K^\ast$ APs (according to $\mathcal{L}^{\mathrm{AP}}$)}{
			\If{$R_k^\mathrm{A} \le 0$ \text{or} $\mathcal{U}^{\mathrm{free}}=\emptyset$}{
				\textbf{continue}
			}
			Build candidate ES set $
			\mathcal{C}_k \leftarrow \big\{ s_i \in \mathcal{U}^{\mathrm{free}} \,\big|\,
			U_{i,k}^\mathrm{Off} > 0 \ \wedge\ \overline{\mathrm{bid}}_i \ge \mathrm{ask}_k \big\}$;
			
			\If{$\mathcal{C}_k = \emptyset$}{
				\textbf{continue}
			}
			$\mathcal{K}_k^\ast \leftarrow$ eACO-VRP (Alg.~\ref{Alg2});\\
			\For{$s_i \in \mathcal{K}_k^\ast$}{
				Set $x^{\mathrm{Off}}_{i,k}\leftarrow 1$;\\
				Remove $s_i$ from $\mathcal{U}^{\mathrm{free}}$:
				$\mathcal{U}^{\mathrm{free}} \leftarrow \mathcal{U}^{\mathrm{free}} \setminus \{s_i\}$;
			}
			Update residual service capacity:
			$R_k^\mathrm{A} \leftarrow R_k^\mathrm{A} - \sum_{s_i\in\mathcal{K}_k^\ast} D_i$.
		}
		Denote the resulting matching by $\bm{X}^\ast=\{x^{\mathrm{Off}}_{i,k}\}$.
		
		\For{each matched ES $s_i$ with $\sum_k x^{\mathrm{Off}}_{i,k} = 1$}{
			Set $high \leftarrow \overline{\mathrm{bid}}_i$,
			$low \leftarrow \overline{\mathrm{bid}}_{I^\ast+1}$
			(or a pre-specified reserve lower bid if $I^\ast=|\bm{\mathcal{S}}|$);\\
			\While{$high-low>\epsilon_{\mathrm{bin}}$}{
				Set $b_{\mathrm{eff}} \leftarrow (low+high)/2$;\\
				Rerun Step~2 on the top-$I^\ast$ ESs and top-$K^\ast$ APs,
				using $b_{\mathrm{eff}}$ as ES $s_i$'s bid, keeping all other bids/asks unchanged, and reinitializing temporary residual capacities as $R_\ell^\mathrm{A}\leftarrow Q_\ell^\mathrm{A}$ for all APs $v_\ell$,
				to obtain a temporary matching $\bm{X}^{\mathrm{Temp}}$;\\
				\lIf{$\sum_k x^{\mathrm{Temp}}_{i,k} = 1$}{$high \leftarrow b_{\mathrm{eff}}$}
				\lElse{$low \leftarrow b_{\mathrm{eff}}$}
			}
			Set $p_i^{\mathrm{b}\ast}\leftarrow high$.
		}
		
		\For{each matched AP $v_k$ with $\sum_i x^{\mathrm{Off}}_{i,k} \ge 1$}{
			Set $low \leftarrow \mathrm{ask}_k$,
			$high \leftarrow \mathrm{ask}_{K^\ast+1}$ 
			(or a pre-specified reserve upper ask if $K^\ast=|\bm{\mathcal{V}}|$);\\
			\While{$high-low>\epsilon_{\mathrm{bin}}$}{
				Set $\mathrm{ask}_{\mathrm{eff}} \leftarrow (low+high)/2$;\\
				Rerun Step~2 on the top-$I^\ast$ ESs and top-$K^\ast$ APs,
				using $\mathrm{ask}_{\mathrm{eff}}$ as AP $v_k$'s ask, keeping all other bids/asks unchanged, and reinitializing temporary residual capacities as $R_\ell^\mathrm{A}\leftarrow Q_\ell^\mathrm{A}$ for all APs $v_\ell$,
				to obtain a temporary matching $\bm{X}^{\mathrm{Temp}}$;\\
				\lIf{$\sum_i x^{\mathrm{Temp}}_{i,k} \ge 1$}{$low \leftarrow \mathrm{ask}_{\mathrm{eff}}$}
				\lElse{$high \leftarrow \mathrm{ask}_{\mathrm{eff}}$}
			}
			Set $r_k^{\mathrm{s}\ast}\leftarrow low$.
		}
		
		\textbf{Output:}
		Matching $\bm{X}^\ast$, ES payments $\bm{p}^{\mathrm{b}\ast}$, AP rewards $\bm{r}^{\mathrm{s}\ast}$.
	}
\end{algorithm}

\noindent
\textbf{Step 1. List generation and pivotal-index search} (lines 3-5, Alg. \ref{Alg3}):
Off-AIC$^2$ operates as a one-shot sealed-bid double auction: all ES buyers $s_i$ and AP sellers $v_k$ simultaneously submit their bids, asks, and resource quantity information to an auctioneer without observing each other's submissions.
We first sort ESs into a list $\mathcal{L}^{\mathrm{ES}}$ in descending order of $\overline{\mathrm{bid}}_i$, and APs into a list $\mathcal{L}^{\mathrm{AP}}$ in ascending order of $\mathrm{ask}_k$. Without loss of generality, we relabel indices so that these sorted orders are directly reflected by $i$ and $k$.
The algorithm then searches for a pair of \emph{pivotal ranks} $(I^\ast,K^\ast)$ such that
$\overline{\mathrm{bid}}_{I^\ast} \ge \mathrm{ask}_{K^\ast}, 
\sum_{k=1}^{K^\ast} Q_k^\mathrm{A} \ge
\sum_{i=1}^{I^\ast} D_i$.

In words, $(I^\ast,K^\ast)$ is the largest prefix of high-bidding ESs and low-asking APs whose aggregate demand and supply can be satisfied at a non-loss-making price level.

\noindent
\textbf{Step 2. ES--AP matching and route planning via eACO-VRP} (lines 6-16, Alg. \ref{Alg3}):
Given the pivotal ranks $(I^\ast,K^\ast)$, the top-$I^\ast$ ESs and top-$K^\ast$ APs enter the matching stage. We initialize the unmatched ES set as $\mathcal{U}^{\mathrm{free}}=\{s_i\mid 1\le i\le I^\ast\}$ and process APs in the order of $\mathcal{L}^{\mathrm{AP}}$. For each AP $v_k$ with residual capacity $R_k^\mathrm{A}>0$ and $\mathcal{U}^{\mathrm{free}}\neq\emptyset$, we construct
$\mathcal{C}_k=\{s_i\in\mathcal{U}^{\mathrm{free}}\mid U_{i,k}^{\mathrm{Off}}>0,\ \overline{\mathrm{bid}}_i\ge\mathrm{ask}_k\}$.
The condition $U_{i,k}^{\mathrm{Off}}>0$ identifies pairs with positive reported per-unit surplus; route and capacity feasibility are evaluated by eACO-VRP.

For AP $v_k$, eACO-VRP searches candidate routes using ES demand, travel time, energy cost, and expected service payments. It selects an ES set $\mathcal{K}_k^\ast\subseteq\mathcal{C}_k$ and a visiting order subject to the residual capacity $R_k^\mathrm{A}$ and service-window constraints. For each selected ES $s_i$, we set $x_{i,k}^{\mathrm{Off}}=1$ and remove $s_i$ from $\mathcal{U}^{\mathrm{free}}$, so that each ES is matched with at most one AP. Combining the per-AP decisions yields $\bm{X}^\ast$. Payments and rewards are determined in Step~3, and individual rationality is analyzed under the settlement assumptions in Appx.~B.2.

\noindent
\textbf{Step 3. Pricing and contract finalization via binary search} (lines 18-33, Alg. \ref{Alg3}):
Once $\bm{X}^\ast$ is determined, Off-AIC$^2$ uses binary search to approximate critical-value payments and rewards. For each participant, the search varies its report while holding the other participants' reports fixed and re-executing the matching step.
In particular, for each matched ES buyer $s_i$ (with $\sum_k x^{\mathrm{Off}}_{i,k}=1$), we treat $\overline{\mathrm{bid}}_i$ as an \emph{effective bid parameter} and search over the interval $[\overline{\mathrm{bid}}_{I^\ast+1},\,\overline{\mathrm{bid}}_i]$ to find the smallest bid $p_i^{\mathrm{b}\ast}$ such that, if $s_i$ had declared this lower bid (while all other bids/asks remained unchanged) and Step~2 was re-executed, $s_i$ would still be selected as a winner. If $I^\ast=|\bm{\mathcal{S}}|$, $\overline{\mathrm{bid}}_{I^\ast+1}$ is replaced by a pre-specified reserve lower bid. The resulting $p_i^{\mathrm{b}\ast}$ approximates the ES's critical payment within the specified search interval.
Similarly, for each matched AP seller $v_k$ (with $\sum_i x^{\mathrm{Off}}_{i,k}\ge 1$), we vary its effective ask within $[\mathrm{ask}_k,\,\mathrm{ask}_{K^\ast+1}]$ (or a pre-specified reserve upper ask if $K^\ast=|\bm{\mathcal{V}}|$) and re-run Step~2 to determine the largest reward level $r_k^{\mathrm{s}\ast}$ such that $v_k$ remains in the winning set. The resulting $r_k^{\mathrm{s}\ast}$ approximates the AP's critical reward within the specified search interval.

{
	The settlement vectors $(\bm{p}^{\mathrm{b}\ast},\bm{r}^{\mathrm{s}\ast})$ determine the payment and reward terms of the offline ES--AP contracts, subject to the non-deficit settlement band specified in Appx.~B.2. Appx.~B.2 analyzes individual rationality, near-truthfulness, and weak budget balance under the stated settlement and bounded-approximation-error assumptions.\footnote{The route-aware allocation is not guaranteed to be monotone. The near-truthfulness analysis therefore relies on bounded allocation and pricing errors rather than strict dominant-strategy truthfulness.} Sec.~\ref{Evaluation_IR} evaluates participant utilities under truthful and deviating reports.
}

\section{Online Stage: Potential Games for Service Scheduling under Human-Machine Coexistence}\label{SEC:ONLINE}
Unlike the offline stage, the online stage handles \textit{real-time} service requests from HUs and MUs, whose number and composition may deviate from offline forecasts. Thus, the offline outcome can only reserve feasible UAVs and service capacities for heavily loaded ESs, but cannot determine the real-time SD--SP assignment, service timing, or service rate. Therefore, we build on the offline contracts, AP routes, and capacity budgets generated before the online service window starts, and design an online SD-side assignment mechanism. This dynamically refines SD--SP associations, while the corresponding SP-side served sets, congestion levels, and resource shares are induced by the SD assignment profile, capacity constraints, and weighted processor sharing (WPS)-based sharing rule.
	In particular, we formulate the online PG-BRD as an SD-side assignment game. The strategic decision-makers are the online SDs, including HUs and MUs, collected as
	$\mathcal{U}^{\mathrm{On}}=\{u^{\mathrm{On}}_1,\ldots,u^{\mathrm{On}}_{\mathbbm{j}},\ldots,u^{\mathrm{On}}_{|\mathcal{U}^{\mathrm{On}}|}\}$. 
	The SPs, collected as
	$\mathcal{S}^{\mathrm{On}}=\{s^{\mathrm{On}}_1,\ldots,s^{\mathrm{On}}_{\mathbbm{i}},\ldots,s^{\mathrm{On}}_{|\mathcal{S}^{\mathrm{On}}|}\}$\footnote{The online SP set $\mathcal{S}^{\mathrm{On}}$ contains both native ES service units and UAV-assisted service units made available by the offline ES--AP contracts. Specifically, when an AP is contracted to support an ES during the current service window, its deployed UAV resources are treated as additional online SP resources with computing capacity, subcarrier capacity, and coverage determined by the corresponding offline route and contract.},
	act as service-side entities that provide computing resources subject to coverage, capacity, and WPS-based sharing rules, but they are not independent best-response players in PG-BRD.

	Each SD $u^{\mathrm{On}}_{\mathbbm{j}}$ chooses either local execution or one accessible SP for service, i.e., $\pi_{\mathbbm{j}}\in\mathcal{A}^{\mathrm{On}}_{\mathbbm{j}}=
		\{0\}\cup\bm{C}_{\mathbbm{j}}$,
	where $\pi_{\mathbbm{j}}=0$ denotes local execution, and $\bm{C}_{\mathbbm{j}}$ denotes the set of SPs accessible to $u^{\mathrm{On}}_{\mathbbm{j}}$ in the current online slot. Therefore, the online decision variable optimized by PG-BRD is the SD assignment profile $\bm{\pi}=(\pi_{\mathbbm{j}})_{u^{\mathrm{On}}_{\mathbbm{j}}\in\mathcal{U}^{\mathrm{On}}}$. For implementation, the scalar assignment $\pi_{\mathbbm{j}}$ can equivalently be represented by a binary indicator vector
\begin{equation}
		\bm{X}^\mathrm{On}_{\mathbbm{j}}
		= \big(x^{\mathrm{On}}_{\mathbbm{j},0},\,\ldots,\,x^{\mathrm{On}}_{\mathbbm{j},\mathbbm{i}},\,\ldots,\,x^{\mathrm{On}}_{\mathbbm{j},|\mathcal{S}^{\mathrm{On}}|}\big)
		\in \{0,1\}^{\,1+|\mathcal{S}^{\mathrm{On}}|},
\end{equation}
where $x^{\mathrm{On}}_{\mathbbm{j},0}=1$ indicates local execution and $x^{\mathrm{On}}_{\mathbbm{j},\mathbbm{i}}=1$ indicates offloading to SP $s^{\mathrm{On}}_{\mathbbm{i}}$. Once $\bm{\pi}$ is given, each SP's served SD set, load level, capacity feasibility, and WPS-based resource share are determined by the model. Thus, SP-side service acceptance and resource sharing are induced quantities rather than additional online strategic actions. We next formalize this joint effect through a delay and congestion model, which will serve as the basis for defining the utilities of SDs and SPs.

\subsection{Utility of SDs and SPs}
\subsubsection{Delay and Congestion Model}
We consider the SP set $\mathcal{S}^{\mathrm{On}}$ and the SD set $\mathcal{U}^{\mathrm{On}}=\bm{\mathcal{U}}^\mathrm{H}\cup\bm{\mathcal{U}}^{\mathrm{M}}$. Each SP $s^{\mathrm{On}}_{\mathbbm{i}}$ has computing capacity $f^{\mathrm{On}}_{\mathbbm{i}}>0$ (CPU cycles/s). The SDs assigned to this SP form the set
$\mathcal{U}^{\mathrm{On}}_{\mathbbm{i}}=\{u^{\mathrm{On}}_{\mathbbm{v}}\in\mathcal{U}^{\mathrm{On}}:x^{\mathrm{On}}_{\mathbbm{v},\mathbbm{i}}=1\}$.
We adopt weighted processor sharing (WPS)~\cite{RW-PG-Zhou} to allocate computing capacity among these SDs. Each associated SD receives a positive weight $\gamma_{\mathbbm{j},\mathbbm{i}}$ determined by its service class, task urgency, fairness quota, or contractual priority. These weights are fixed during each PG-BRD run. The effective computing rate allocated to $u^{\mathrm{On}}_{\mathbbm{j}}$ is
$\hat f_{\mathbbm{j},\mathbbm{i}}
= f^{\mathrm{On}}_{\mathbbm{i}}
\frac{\gamma_{\mathbbm{j},\mathbbm{i}}}{\sum_{u^{\mathrm{On}}_{\mathbbm{v}}\in\mathcal{U}^{\mathrm{On}}_{\mathbbm{i}}} \gamma_{\mathbbm{v},\mathbbm{i}}}$,
and the corresponding edge-side computing delay can be obtained:
\begin{equation}
	t^{\mathrm{comp}}_{\mathbbm{j},\mathbbm{i}}
	= \frac{r_{\mathbbm{j}}}{\hat f_{\mathbbm{j},\mathbbm{i}}}
	= \frac{r_{\mathbbm{j}}}{\gamma_{\mathbbm{j},\mathbbm{i}}}
	\times \frac{\sum_{u^{\mathrm{On}}_{\mathbbm{v}}\in\mathcal{U}^{\mathrm{On}}_{\mathbbm{i}}}\gamma_{\mathbbm{v},\mathbbm{i}}}{f^{\mathrm{On}}_{\mathbbm{i}}}, 
	\label{eq:comp_delay_origin}
\end{equation}
where $r_{\mathbbm{j}}$ denotes the computing workload (e.g., CPU cycles needed) of $u^{\mathrm{On}}_{\mathbbm{j}}$. This WPS model compactly captures how the delay of each SD grows with the aggregate load at the SP. We further summarize the impact of  associated SDs on $s^{\mathrm{On}}_{\mathbbm{i}}$ through an \emph{effective congestion load} defined as
$y_{\mathbbm{i}} = \sum_{u^{\mathrm{On}}_{\mathbbm{v}}\in\mathcal{U}^{\mathrm{On}}_{\mathbbm{i}}} \gamma_{\mathbbm{v},\mathbbm{i}}$.
With this notation, \eqref{eq:comp_delay_origin} can be written in a factorized form as:
\begin{equation}
	t^{\mathrm{comp}}_{\mathbbm{j},\mathbbm{i}}
	= \theta_{\mathbbm{j},\mathbbm{i}}\;\tau_{\mathbbm{i}}(y_{\mathbbm{i}}),\quad
	\theta_{\mathbbm{j},\mathbbm{i}} = \frac{r_{\mathbbm{j}}}{\gamma_{\mathbbm{j},\mathbbm{i}}},\quad
	\tau_{\mathbbm{i}}(y) = \frac{y}{f^{\mathrm{On}}_{\mathbbm{i}}},
	\label{eq:comp_delay_factorized}
\end{equation}
where $\theta_{\mathbbm{j},\mathbbm{i}}$ is a user-specific load factor, whereas $\tau_{\mathbbm{i}}(y)$ captures the server-side congestion effect as a function of the aggregate load. This separable structure will be instrumental for constructing the potential-aligned payoff and the exact game potential in the subsequent analysis.

While \eqref{eq:comp_delay_factorized} characterizes the congestion-induced computing delay, an offloading SD also experiences wireless transmission latency in the uplink: when $u^{\mathrm{On}}_{\mathbbm{j}}$ offloads to $s^{\mathrm{On}}_{\mathbbm{i}}$, the uplink delay is 
$t^{\mathrm{tx}}_{\mathbbm{j},\mathbbm{i}}
= \frac{d_{\mathbbm{j}}}{W \log_2 \!\big(1 + e^{\mathrm{tx}}_{\mathbbm{j}}\,\Gamma_{\mathbbm{j},\mathbbm{i}} \big)}$,
where $d_{\mathbbm{j}}$ is the input data size, $W$ is the bandwidth, $e^{\mathrm{tx}}_{\mathbbm{j}}\Gamma_{\mathbbm{j},\mathbbm{i}}$ is the signal-to-noise ratio (SNR) between $u^{\mathrm{On}}_{\mathbbm{j}}$ and $s^{\mathrm{On}}_{\mathbbm{i}}$. Consequently, the total edge-side latency experienced by $u^{\mathrm{On}}_{\mathbbm{j}}$ when uploading its workload to $s^{\mathrm{On}}_{\mathbbm{i}}$ is:
\begin{equation}
	T^{\mathrm{edge}}_{\mathbbm{j},\mathbbm{i}}
	= t^{\mathrm{comp}}_{\mathbbm{j},\mathbbm{i}} + t^{\mathrm{tx}}_{\mathbbm{j},\mathbbm{i}}.
	\label{eq:tot_edge_latency}
\end{equation}

\subsubsection{Utility of HUs}
Building on the above latency and energy models, we now characterize the utility of an HU when deciding whether to execute a task locally or offload it to an SP. Consider an HU $u^{\mathrm{On}}_{\mathbbm{j}}\in\bm{\mathcal{U}}^{\mathrm{H}}$ with local CPU frequency $f^{\mathrm{loc}}_{\mathbbm{j}}$ and local computational power $e^{\mathrm{loc}}_{\mathbbm{j}}$. The local execution time and energy consumption of its task are given by
$t^{\mathrm{loc}}_{\mathbbm{j}}={r_{\mathbbm{j}}}/{f^{\mathrm{loc}}_{\mathbbm{j}}}$ and
$E^{\mathrm{loc}}_{\mathbbm{j}}={r_{\mathbbm{j}}\,e^{\mathrm{loc}}_{\mathbbm{j}}}/{f^{\mathrm{loc}}_{\mathbbm{j}}}$, respectively.
When the HU $u^{\mathrm{On}}_{\mathbbm{j}}$ offloads its task to $s^{\mathrm{On}}_{\mathbbm{i}}$, it experiences  latency $t^{\mathrm{comp}}_{\mathbbm{j},\mathbbm{i}} + t^{\mathrm{tx}}_{\mathbbm{j},\mathbbm{i}}$ and uplink energy $E^{\mathrm{tx}}_{\mathbbm{j},\mathbbm{i}}$. Compared with local task execution, this leads to the latency saving of
$t^{\mathrm{save}}_{\mathbbm{j},\mathbbm{i}}
{=} t^{\mathrm{loc}}_{\mathbbm{j}} - \big(t^{\mathrm{comp}}_{\mathbbm{j},\mathbbm{i}} + t^{\mathrm{tx}}_{\mathbbm{j},\mathbbm{i}}\big)$
and energy saving of
$c^{\mathrm{save}}_{\mathbbm{j},\mathbbm{i}}
{=} E^{\mathrm{loc}}_{\mathbbm{j}} - E^{\mathrm{tx}}_{\mathbbm{j},\mathbbm{i}}$ ($E^{\mathrm{tx}}_{\mathbbm{j},\mathbbm{i}}$ denotes the uplink transmission energy).

To map these latency and energy savings into a scalar benefit, we introduce non-negative valuation coefficients $\mathbb{V}^{\mathrm{H}}_t$ and $\mathbb{V}^{\mathrm{H}}_e$ to capture HU sensitivities to latency and energy, respectively. The HU's perceived valuation of being served by $s^{\mathrm{On}}_{\mathbbm{i}}$ is then modeled via
$v^{\mathrm{H}}_{\mathbbm{j},\mathbbm{i}}
= \mathbb{V}^{\mathrm{H}}_t\,t^{\mathrm{save}}_{\mathbbm{j},\mathbbm{i}}
+ \mathbb{V}^{\mathrm{H}}_e\,c^{\mathrm{save}}_{\mathbbm{j},\mathbbm{i}}$.
Given this valuation, the HU's net utility must also account for the service payment: if $s^{\mathrm{On}}_{\mathbbm{i}}$ charges a price $p_{\mathbbm{j},\mathbbm{i}}$ for serving $u^{\mathrm{On}}_{\mathbbm{j}}$, its net utility is
\begin{equation}
	U^{\mathrm{H}}_{\mathbbm{j},\mathbbm{i}}
	= v^{\mathrm{H}}_{\mathbbm{j},\mathbbm{i}} - p_{\mathbbm{j},\mathbbm{i}},\qquad
	U^{\mathrm{H}}_{\mathbbm{j},0}=0,
	\label{eq:hu_utility}
\end{equation}
where the value of $U^{\mathrm{H}}_{\mathbbm{j},0}$ corresponds to local task execution.

\subsubsection{Utility of MUs}
While HUs mainly care about the joint latency–energy tradeoff, MUs (e.g., robots or sensing devices) are typically driven by task deadlines and reliability requirements. We therefore adopt a slightly different utility structure for MUs. 
Specifically, an MU $u^{\mathrm{On}}_{\mathbbm{j}}\in\bm{\mathcal{U}}^{\mathrm{M}}$ is associated with a latency deadline $\tau^\mathrm{ddl}_{\mathbbm{j}}$, while its total edge-side latency at $s^{\mathrm{On}}_{\mathbbm{i}}$ is $T^{\mathrm{edge}}_{\mathbbm{j},\mathbbm{i}}$ in \eqref{eq:tot_edge_latency}. To capture the fact that the MU primarily values meeting this deadline, we adopt a \textit{soft-deadline} valuation
\begin{equation}
	G^{\mathrm{M}}_{\mathbbm{j},\mathbbm{i}}
	= v_{\mathbbm{j}}\,
	{\big[\tau^\mathrm{ddl}_{\mathbbm{j}}-T^{\mathrm{edge}}_{\mathbbm{j},\mathbbm{i}}\big]_+}\times ({\tau^\mathrm{ddl}_{\mathbbm{j}}})^{-1},
	\label{eq:mu_soft_deadline}
\end{equation}
where $[x]_+=\max\{x,0\}$ and $v_{\mathbbm{j}}>0$ denotes the MU's task reward scale. This valuation is maximal when the edge-side latency is zero, decreases linearly as the latency consumes the available deadline slack, remains positive before the deadline, and becomes zero once the deadline is violated. In particular, ``soft-deadline'' means that the reward is continuous and piecewise linear with respect to latency, rather than being a discontinuous hard-threshold reward.

In addition to timeliness, an MU also benefits from saving local computing energy by task offloading. Similar to HUs, the MU's energy saving when utilizing edge services is
$c^{\mathrm{M},\mathrm{save}}_{\mathbbm{j},\mathbbm{i}}
= E^{\mathrm{loc}}_{\mathbbm{j}} - E^{\mathrm{tx}}_{\mathbbm{j},\mathbbm{i}}$,
where $E^{\mathrm{loc}}_{\mathbbm{j}}$ is the local execution energy and $E^{\mathrm{tx}}_{\mathbbm{j},\mathbbm{i}}$ is the uplink transmission energy. Let $\mathbb{V}^{\mathrm{M}}_e$ denote the MU's valuation coefficient on energy saving. Combining deadline satisfaction, energy saving, and the service payment $p_{\mathbbm{j},\mathbbm{i}}$, the MU's utility when offloading its task to $s^{\mathrm{On}}_{\mathbbm{i}}$ is given by
\begin{equation}
	U^{\mathrm{M}}_{\mathbbm{j},\mathbbm{i}}
	= G^{\mathrm{M}}_{\mathbbm{j},\mathbbm{i}}
	+ \mathbb{V}^{\mathrm{M}}_e\,c^{\mathrm{M},\mathrm{save}}_{\mathbbm{j},\mathbbm{i}}
	- p_{\mathbbm{j},\mathbbm{i}},\qquad
	U^{\mathrm{M}}_{\mathbbm{j},0}=0,
	\label{eq:au_utility}
\end{equation}
where $U^{\mathrm{M}}_{\mathbbm{j},0}$ corresponds to local task execution. 

Note that both HU and MU utilities in \eqref{eq:hu_utility} and \eqref{eq:au_utility} inherently share the same congestion component $t^{\mathrm{comp}}_{\mathbbm{j},\mathbbm{i}}=\theta_{\mathbbm{j},\mathbbm{i}}\tau_{\mathbbm{i}}(y_{\mathbbm{i}})$ in \eqref{eq:comp_delay_factorized}. This shared structure will be crucial for aligning SD-side updates with the exact game potential in our subsequent analysis (Sec.~\ref{subsec:potential}).

\subsubsection{Cost of SPs}
Having specified how HUs and MUs (i.e., SDs) evaluate edge services, we now focus on the SPs' perspective. As more SDs offload their tasks to an SP, the SP consumes more  energy and experiences higher operating costs, which must be reflected in its payoff and, ultimately, in the system-level SW. Subsequently, for each SP $s^{\mathrm{On}}_{\mathbbm{i}}$, given its aggregated computing load 
$\mathbbm{x}_{\mathbbm{i}} = \sum_{u^{\mathrm{On}}_{\mathbbm{v}}\in\mathcal{U}^{\mathrm{On}}_{\mathbbm{i}}} r_{\mathbbm{v}}$, we model its  computation/operation cost as follows:
\begin{equation}
	K_{\mathbbm{i}}(\mathbbm{x}_{\mathbbm{i}})
	= \omega_1\,e^{\mathrm{comp}}_{\mathbbm{i}}
	\frac{\mathbbm{x}_{\mathbbm{i}}}{f^{\mathrm{On}}_{\mathbbm{i}}}
	+ \omega_2 c_{\mathbbm{i}}^\mathrm{hard}
	\mathbbm{1}_{\{\mathbbm{x}_{\mathbbm{i}}>0\}},
	\label{eq:sp_cost}
\end{equation}
where $e^{\mathrm{comp}}_{\mathbbm{i}}$ is the average computing power consumption of $s^{\mathrm{On}}_{\mathbbm{i}}$, and $c_{\mathbbm{i}}^\mathrm{hard}$ captures hardware- and maintenance-related costs ($\omega_1,\omega_2\geq 0$ are weighting coefficients). In \eqref{eq:sp_cost}, the term proportional to $\mathbbm{x}_{\mathbbm{i}}/f^{\mathrm{On}}_{\mathbbm{i}}$ reflects that processing a larger workload leads to higher energy and operational expenditure, while the indicator term charges maintenance/depreciation only when the SP actively serves at least one SD.
It is worth noting that service payments have different roles in welfare evaluation and game dynamics. When evaluating transfer-free system welfare, payments between SDs and SPs cancel out as internal transfers. However, in the SD-side best-response dynamics, posted prices affect individual action choices and must be included in the game potential. We therefore distinguish the transfer-free welfare $\mathcal{W}(\bm{\pi})$ from the game potential $\Phi^{\mathrm{G}}(\bm{\pi})$ constructed in Sec.~\ref{subsec:potential}.

SP resources are limited, and SD decisions are coupled through congestion and capacity constraints. We therefore formulate online scheduling as a game in which SDs update their assignments using a potential-aligned payoff. This formulation supports distributed assignment updates while accounting for their effects on shared resources.

\subsection{Exact Potential Function}\label{subsec:potential}
In order to analyze our online resource scheduling game and to enable simple distributed best-response dynamics with provable convergence, we construct an exact game potential $\Phi^{\mathrm{G}}(\bm{\pi})$ aligned with the potential-guided payoff used by PG-BRD. This construction leverages the unified congestion structure in~\eqref{eq:comp_delay_factorized}, the active-regime form of the MU soft-deadline utility, and fixed posted prices that are determined before the PG-BRD procedure for a given online scheduling instance.

For an MU $u^{\mathrm{On}}_{\mathbbm{j}}$ served by $s^{\mathrm{On}}_{\mathbbm{i}}$, the soft-deadline valuation in~\eqref{eq:mu_soft_deadline} is non-increasing in the congestion load $y_{\mathbbm{i}}$. In the active regime, it decreases linearly with congestion, yielding the decomposition in~\eqref{eq:au_decomp_soft_new}; outside this regime, it remains zero.
For HUs, the utility in \eqref{eq:hu_utility} can be decomposed into a congestion-independent valuation part and a congestion-dependent penalty:
\begin{equation}{\small\begin{aligned}
			U^{\mathrm{H}}_{\mathbbm{j},\mathbbm{i}}
			{=} \underbrace{\mathbb{V}^{\mathrm{H}}_t\!\big(t^{\mathrm{loc}}_{\mathbbm{j}} - t^{\mathrm{tx}}_{\mathbbm{j},\mathbbm{i}}\big)
				+ \mathbb{V}^{\mathrm{H}}_e\!\big(E^{\mathrm{loc}}_{\mathbbm{j}} - E^{\mathrm{tx}}_{\mathbbm{j},\mathbbm{i}}\big)}_{S_{\mathbbm{j}}(\mathbbm{i})\;\text{(valuation independent of }y_{\mathbbm{i}}\text{)}}
			{-}\;\underbrace{\mathbb{V}^{\mathrm{H}}_t\,\theta_{\mathbbm{j},\mathbbm{i}}\,\tau_{\mathbbm{i}}(y_{\mathbbm{i}})}_{\text{congestion penalty}}
			{-}\;p_{\mathbbm{j},\mathbbm{i}},
	\end{aligned}}
	\label{eq:hu_decomp_new}
\end{equation}
where $S_{\mathbbm{j}}(\mathbbm{i})$ collects all valuation terms that do not depend on the congestion load $y_{\mathbbm{i}}$, while the second term captures how queueing-induced delay degrades the HU's utility. 
The fixed payment $p_{\mathbbm{j},\mathbbm{i}}$ is written separately to distinguish the posted-price term from the service valuation.

MU utilities also admit the same structural form in the active soft-deadline regime (see Appx.~E.4). In this regime, the MU utility can be rewritten as:
\begin{equation}
	U^{\mathrm{M}}_{\mathbbm{j},\mathbbm{i}}
	= \bar S_{\mathbbm{j}}(\mathbbm{i})
	\;-\;\kappa_{\mathbbm{j}}\,\theta_{\mathbbm{j},\mathbbm{i}}\,\tau_{\mathbbm{i}}(y_{\mathbbm{i}})
	\;-\;p_{\mathbbm{j},\mathbbm{i}},
	\label{eq:au_decomp_soft_new}
\end{equation}
where $\bar S_{\mathbbm{j}}(\mathbbm{i})$ gathers the non-congestion valuation terms (soft-deadline valuation base reward and energy saving), and $\kappa_{\mathbbm{j}}>0$ reflects the MU's effective sensitivity to congestion. 
As in~\eqref{eq:hu_decomp_new}, the posted price is kept separate from the non-congestion valuation term. 
 
{ Considering both HUs and MUs, we collect their non-congestion service valuation into
 $S_{\mathbbm{j}}(\mathbbm{i})$, where $\mathbbm{i}=0$ denotes local execution and
 $S_{\mathbbm{j}}(0)=0$. For HUs, $S_{\mathbbm{j}}(\mathbbm{i})$ is obtained from the latency--energy valuation in~\eqref{eq:hu_decomp_new}; for MUs, it is obtained from the soft-deadline valuation in~\eqref{eq:au_decomp_soft_new} under the active soft-deadline regime.
 To align selfish SD updates with a global potential, PG-BRD uses a potential-aligned scheduling payoff rather than the raw transfer-inclusive utility alone. Given the actions of all other SDs, define $y_{\mathbbm{i}}^{-\mathbbm{j}} = \sum_{\mathbbm{v}\neq\mathbbm{j}:\pi_{\mathbbm{v}}=\mathbbm{i}} \gamma_{\mathbbm{v},\mathbbm{i}}, ~\mathbbm{x}_{\mathbbm{i}}^{-\mathbbm{j}} = \sum_{\mathbbm{v}\neq\mathbbm{j}:\pi_{\mathbbm{v}}=\mathbbm{i}} r_{\mathbbm{v}} . $ For an offloading action $\mathbbm{i}\in\mathcal{S}^{\mathrm{On}}$, define
 \begin{equation}
 	\begin{aligned}
 		\widetilde U_{\mathbbm{j},\mathbbm{i}}(\bm{\pi}_{-\mathbbm{j}})
 		=&~
 		S_{\mathbbm{j}}(\mathbbm{i})
 		-
 		\int_{y_{\mathbbm{i}}^{-\mathbbm{j}}}^{\,y_{\mathbbm{i}}^{-\mathbbm{j}}+\gamma_{\mathbbm{j},\mathbbm{i}}}
 		\tau_{\mathbbm{i}}(z)\,dz\\
 		&-
 		\Big[
 		K_{\mathbbm{i}}\big(\mathbbm{x}_{\mathbbm{i}}^{-\mathbbm{j}}+r_{\mathbbm{j}}\big)
 		-
 		K_{\mathbbm{i}}\big(\mathbbm{x}_{\mathbbm{i}}^{-\mathbbm{j}}\big)
 		\Big]
 		-
 		p_{\mathbbm{j},\mathbbm{i}},
 	\end{aligned}
 	\label{eq:unified_utility_new}
 \end{equation}
 and set $\widetilde U_{\mathbbm{j},0}=0$ for local execution. Here, $p_{\mathbbm{j},\mathbbm{i}}$ is fixed during the PG-BRD procedure for the considered online state. The integral term prices the marginal congestion contribution of SD $u_{\mathbbm{j}}^{\mathrm{On}}$, while the bracketed term prices its marginal SP-side operating cost. Therefore, a unilateral increase in $\widetilde U_{\mathbbm{j},\mathbbm{i}}$ is exactly aligned with the change in the game potential defined below.}

To describe how individual decisions interact at the system-level, we collect all SD assignments into the \emph{joint assignment profile}\footnote{Given $\bm{\pi}$, the binary indicators $x^{\mathrm{On}}_{\mathbbm{j},\mathbbm{i}}$ are induced by
	$x^{\mathrm{On}}_{\mathbbm{j},\mathbbm{i}} = \mathbbm{1}_{\{\pi_{\mathbbm{j}}=\mathbbm{i}\}},~
	\mathbbm{i}\in\{0\}\cup\mathcal{S}^{\mathrm{On}}$ ,
	so that the vector $\bm{X}^{\mathrm{On}}_{\mathbbm{j}}$ and the scalar action $\pi_{\mathbbm{j}}$ are equivalent representations of $u^{\mathrm{On}}_{\mathbbm{j}}$'s decision.
} $\bm{\pi}=(\pi_{\mathbbm{j}})_{\mathbbm{j}\in\mathcal{U}^{\mathrm{On}}}$, where $\pi_{\mathbbm{j}}\in\{0\}\cup\mathcal{S}^{\mathrm{On}}$ denotes the choice of $u^{\mathrm{On}}_{\mathbbm{j}}$: $\pi_{\mathbbm{j}}=0$ means that the task is processed locally at $u^{\mathrm{On}}_{\mathbbm{j}}$, while $\pi_{\mathbbm{j}}=\mathbbm{i}$ indicates task offloading to SP $s^{\mathrm{On}}_{\mathbbm{i}}$. 
Given a profile $\bm{\pi}$, we can then compute, for each SP $s^{\mathrm{On}}_{\mathbbm{i}}$, the effective congestion load and the physical computing load as follows:
\begin{equation}
	y_{\mathbbm{i}}(\bm{\pi}) = \sum_{\mathbbm{j}:\,\pi_{\mathbbm{j}}=\mathbbm{i}} \gamma_{\mathbbm{j},\mathbbm{i}},
	\qquad
	\mathbbm{x}_{\mathbbm{i}}(\bm{\pi}) = \sum_{\mathbbm{j}:\,\pi_{\mathbbm{j}}=\mathbbm{i}} r_{\mathbbm{j}},
	\label{eq:cong_vars_pi_final}
\end{equation}
where $y_{\mathbbm{i}}(\bm{\pi})$ feeds into the congestion-induced delay 
via $t^{\mathrm{comp}}_{\mathbbm{j},\mathbbm{i}}
= \theta_{\mathbbm{j},\mathbbm{i}}\,\tau_{\mathbbm{i}}\big(y_{\mathbbm{i}}(\bm{\pi})\big)$
in~\eqref{eq:comp_delay_factorized}, and $\mathbbm{x}_{\mathbbm{i}}(\bm{\pi})$ 
determines the SP-side computation/operation cost through 
$K_{\mathbbm{i}}\big(\mathbbm{x}_{\mathbbm{i}}(\bm{\pi})\big)$ in~\eqref{eq:sp_cost}. Thus, any unilateral change of $\pi_{\mathbbm{j}}$ only modifies the loads of the involved SPs in~\eqref{eq:cong_vars_pi_final}, and these updated loads propagate to SD delays and SP costs through \eqref{eq:comp_delay_factorized} and \eqref{eq:sp_cost}, respectively.

In order to analyze the resulting game and to show that simple myopic best-response updates can converge, we construct the following game potential:
\begin{equation}\label{eq:potential_def_pi_final}
	\begin{aligned}
		\Phi^{\mathrm{G}}(\bm{\pi})
		=&
		\sum_{\mathbbm{j}:\,\pi_{\mathbbm{j}}\neq 0}
		S_{\mathbbm{j}}(\pi_{\mathbbm{j}})
		-
		\sum_{\mathbbm{i}\in\mathcal{S}^{\mathrm{On}}}
		\int_{0}^{\,y_{\mathbbm{i}}(\bm{\pi})}
		\tau_{\mathbbm{i}}(z)\,dz\\
		&-
		\sum_{\mathbbm{i}\in\mathcal{S}^{\mathrm{On}}}
		K_{\mathbbm{i}}\!\big(\mathbbm{x}_{\mathbbm{i}}(\bm{\pi})\big)
		-
		\sum_{\mathbbm{j}:\,\pi_{\mathbbm{j}}\neq 0}
		p_{\mathbbm{j},\pi_{\mathbbm{j}}}.
	\end{aligned}
\end{equation}
The first term aggregates SD-side non-congestion service valuations, the second term captures aggregate congestion, the third term captures SP-side operating costs, and the last term accounts for fixed posted prices that directly affect SD-side best responses. For system-efficiency evaluation, we also use the transfer-free welfare $\mathcal{W}(\bm{\pi})
=
\sum_{\mathbbm{j}:\,\pi_{\mathbbm{j}}\neq 0}
S_{\mathbbm{j}}(\pi_{\mathbbm{j}})
-
\sum_{\mathbbm{i}\in\mathcal{S}^{\mathrm{On}}}
\int_{0}^{\,y_{\mathbbm{i}}(\bm{\pi})}
\tau_{\mathbbm{i}}(z)\,dz-
\sum_{\mathbbm{i}\in\mathcal{S}^{\mathrm{On}}}
K_{\mathbbm{i}}\!\big(\mathbbm{x}_{\mathbbm{i}}(\bm{\pi})\big) $,
which differs from $\Phi^{\mathrm{G}}(\bm{\pi})$ only by excluding monetary transfers.

\begin{Prop}[Exact potential of the online scheduling game]\label{thm:exact_pi}
	Consider the online scheduling game in which each SD updates its action according to the potential-aligned payoff in~\eqref{eq:unified_utility_new}. For any SD $u_{\mathbbm{j}}^{\mathrm{On}}$ and any unilateral deviation from $\mathbbm{i}$ to $\mathbbm{i}'$, we have
	\begin{equation}
		\widetilde U_{\mathbbm{j},\mathbbm{i}'}(\bm{\pi}_{-\mathbbm{j}})
		-
		\widetilde U_{\mathbbm{j},\mathbbm{i}}(\bm{\pi}_{-\mathbbm{j}})
		=
		\Phi^{\mathrm{G}}(\bm{\pi}_{-\mathbbm{j}},\mathbbm{i}')
		-
		\Phi^{\mathrm{G}}(\bm{\pi}_{-\mathbbm{j}},\mathbbm{i}).
	\end{equation}
	Hence, $\Phi^{\mathrm{G}}(\bm{\pi})$ is an exact potential function. Any unilateral payoff-improving update strictly increases $\Phi^{\mathrm{G}}(\bm{\pi})$.
\end{Prop}

Since each SD has a finite feasible action set, any sequence of strict better-response updates under~\eqref{eq:unified_utility_new} terminates in finitely many steps at a pure-strategy Nash equilibrium when $\varepsilon=0$, and at an $\varepsilon$-NE when PG-BRD uses a positive acceptance threshold $\varepsilon$. The proof is given in Appx.~E.

\subsection{Problem Formulation}

In the online stage, we optimize the SD assignment profile $\bm{\pi}$ subject to coverage and subcarrier-capacity constraints. SP-side served sets, congestion levels, and WPS resource shares are induced by this profile. Since any maximizer of $\Phi^{\mathrm{G}}(\bm{\pi})$ admits no profitable unilateral feasible deviation, we formulate the online scheduling problem as $\bm{\mathcal{F}}^{(\mathrm{II})}$:
\begin{subequations}
	\begin{align}
		\bm{\mathcal{F}}^{(\mathrm{II})}:\quad
		&\max_{\ \bm{\pi}}~\Phi^{\mathrm{G}}(\bm{\pi})
		\tag{23}\label{PF_2}\\[2pt]
		\mathrm{s.t.}\ \
		&\pi_{\mathbbm{j}}\in\{0\}\cup\mathcal{S}^{\mathrm{On}},\quad \forall\,u^\mathrm{On}_\mathbbm{j}\in\mathcal{U}^{\mathrm{On}},
		\tag{23a}\label{23a}\\
		&x_{\mathbbm{j},\mathbbm{i}}^\mathrm{On}=0,\ \ \text{if } s^\mathrm{On}_\mathbbm{i}\notin \bm{C}_\mathbbm{j},\quad \forall\,u^\mathrm{On}_\mathbbm{j}\in\mathcal{U}^{\mathrm{On}},
		\tag{23b}\label{23b}\\
		&\sum_{u^\mathrm{On}_\mathbbm{j}\in {\mathcal{U}^{\mathrm{On}}}} x_{\mathbbm{j},\mathbbm{i}}^\mathrm{On}
		\le G_{\mathbbm{i}},\quad \forall\,s^\mathrm{On}_\mathbbm{i}\in\mathcal{S}^{\mathrm{On}},
		\tag{23c}\label{23c}
	\end{align}
\end{subequations}
\setcounter{equation}{23}

\noindent
where $\bm{\pi}=(\pi_{\mathbbm{j}})_{\mathbbm{j}\in\mathcal{U}^{\mathrm{On}}}$ denotes the joint assignment profile, with $\pi_{\mathbbm{j}}=0$ representing local task execution at $u^\mathrm{On}_\mathbbm{j}$ and $\pi_{\mathbbm{j}}=\mathbbm{i}$ indicating task offloading to SP $s^\mathrm{On}_\mathbbm{i}$. The binary variables $x_{\mathbbm{j},\mathbbm{i}}^\mathrm{On}$ are induced by $\bm{\pi}$ as defined in Sec.~\ref{subsec:potential} (i.e., $x_{\mathbbm{j},\mathbbm{i}}^\mathrm{On}=1$ if $\pi_{\mathbbm{j}}=\mathbbm{i}$).

Constraint~\eqref{23a} enforces each SD to choose only one execution option for its current task: either local processing ($\pi_{\mathbbm{j}}=0$) or offloading to a single SP ($\pi_{\mathbbm{j}}\in\mathcal{S}^{\mathrm{On}}$). Constraint~\eqref{23b} captures the service-coverage feasibility: $\bm{C}_\mathbbm{j}$ denotes the set of SPs that can physically serve $u^\mathrm{On}_\mathbbm{j}$ (e.g., within communication range), and thus the association variable for
an SD outside the coverage of $s^\mathrm{On}_\mathbbm{i}$ is forced to zero. Constraint~\eqref{23c} limits the admission load at each SP, where $G_{\mathbbm{i}}$ is the maximum number of SDs that SP $s^\mathrm{On}_\mathbbm{i}$ can concurrently serve.

{
	Problem $\bm{\mathcal{F}}^{(\mathrm{II})}$ is an integer programming (IP) problem over a discrete action space and is, in general, NP-hard.
} Since this online scheduling has to be repeatedly solved at each trading round under stringent latency and computational constraints at the edge, running a centralized IP solver or heavyweight meta-heuristics to seek a global optimum is impractical. Instead, by recasting $\bm{\mathcal{F}}^{(\mathrm{II})}$ as an exact-potential game with potential $\Phi^{\mathrm{G}}(\bm{\pi})$, we obtain a formulation that admits a fully distributed implementation. In particular, PG-BRD (Sec.~\ref{sec:PG-BRD}) lets each SD myopically maximize the potential-aligned payoff in~\eqref{eq:unified_utility_new} subject to feasibility. Every accepted update monotonically increases $\Phi^{\mathrm{G}}(\bm{\pi})$; if PG-BRD terminates with $\mathrm{converged}=\textbf{true}$, the resulting profile is an $\varepsilon$-NE, and it becomes a pure-strategy NE when $\varepsilon=0$ (also see Appx.~D and Appx.~E). Thus, PG-BRD provides a feasible distributed solution to~\eqref{PF_2}, with the stated equilibrium guarantee when the convergence condition is satisfied.

\subsection{Potential-Guided Best-Response Dynamics (PG-BRD)}\label{sec:PG-BRD}

 For brevity, we provide the complete pseudo-code of PG-BRD in Alg. \ref{alg:pg_brd} and summarize its four main steps below:
 
 \begin{algorithm}[t!]
 	{\footnotesize
 		\caption{Proposed PG-BRD}
 		\label{alg:pg_brd}
 		\LinesNumbered
 		{
 			\textbf{Input:} SD set $\mathcal{U}^{\mathrm{On}}$, SP set $\mathcal{S}^{\mathrm{On}}$, 
 			coverage sets $\{\bm{C}_{\mathbbm{j}}\}$, subcarrier-capacity $\{G_{\mathbbm{i}}\}$, prices $\{p_{\mathbbm{j},\mathbbm{i}}\}$, WPS weights $\{\gamma_{\mathbbm{j},\mathbbm{i}}\}$, 
 			max iterations $T_{\max}$, threshold $\varepsilon$.\\[2pt]
 		}
 		
 		\textbf{Initialization:}
 		Construct a random feasible profile 
 		$\bm{\pi} = (\pi_{\mathbbm{j}})_{\mathbbm{j}\in\mathcal{U}^{\mathrm{On}}}$ 
 		that satisfies constraints~\eqref{23a}--\eqref{23c};
 		set $\mathbbm{t} \leftarrow 0$ and $\mathrm{converged} \leftarrow \textbf{false}$.\\[2pt]
 		
 		\While{not $\mathrm{converged}$ \textbf{and} $\mathbbm{t} < T_{\max}$}{
 			
 			$\mathrm{converged} \leftarrow \textbf{true}$;\\
 			Generate a random permutation $\mathcal{P}$ of indices in $\mathcal{U}^{\mathrm{On}}$;\\
 			
 			\For{each SD $u^{\mathrm{On}}_{\mathbbm{j}}$ in $\mathcal{P}$}{
 				let $a^{\mathrm{cur}}_{\mathbbm{j}} \leftarrow \pi_{\mathbbm{j}}$;\\
 				for all $\mathbbm{i}$, compute the loads excluding $u^{\mathrm{On}}_{\mathbbm{j}}$:
 				$y^{-{\mathbbm{j}}}_{\mathbbm{i}}$, $\mathbbm{x}^{-{\mathbbm{j}}}_{\mathbbm{i}}$, 
 				and $n^{-{\mathbbm{j}}}_{\mathbbm{i}}$;\\
 				
 				build the feasible action set $\mathcal{A}^{\mathrm{feas}}_{\mathbbm{j}}$ according to~\eqref{feas_new};\\
 				
 				compute current payoff 
 				$\widetilde U_{\mathbbm{j}}(a^{\mathrm{cur}}_{\mathbbm{j}}\mid\bm{\pi}_{-\mathbbm{j}})$ 
 				via~\eqref{eq:unified_utility_new};\\
 				
 				\For{each $a\in\mathcal{A}^{\mathrm{feas}}_{\mathbbm{j}}$}{
 					temporarily set $\pi_{\mathbbm{j}}\leftarrow a$ and evaluate
 					$\widetilde U_{\mathbbm{j}}(a\mid\bm{\pi}_{-\mathbbm{j}})$;\\
 				}
 				
 				let 
 				$a^{\mathrm{best}}_{\mathbbm{j}}
 				\in\arg\max_{a\in\mathcal{A}^{\mathrm{feas}}_{\mathbbm{j}}} 
 				\widetilde U_{\mathbbm{j}}(a\mid\bm{\pi}_{-\mathbbm{j}})$;\\
 				
 				\If{$\widetilde U_{\mathbbm{j}}(a^{\mathrm{best}}_{\mathbbm{j}}\mid\bm{\pi}_{-\mathbbm{j}}) 
 					- \widetilde U_{\mathbbm{j}}(a^{\mathrm{cur}}_{\mathbbm{j}}\mid\bm{\pi}_{-\mathbbm{j}}) > \varepsilon$}{		$\pi_{\mathbbm{j}}\leftarrow a^{\mathrm{best}}_{\mathbbm{j}}$;\\
 					$\mathrm{converged} \leftarrow \textbf{false}$;
 				}
 				\Else{
 					$\pi_{\mathbbm{j}}\leftarrow a^{\mathrm{cur}}_{\mathbbm{j}}$;
 				}
 			}
 			
 			$\mathbbm{t} \leftarrow \mathbbm{t} + 1$;
 		}
 		
 	\textbf{Output:} Final feasible profile $\bm{\pi}^{\star}=\bm{\pi}$ and convergence flag $\mathrm{converged}$.
 	}
 \end{algorithm}

\noindent\textbf{Step 1. Random-feasible initialization} (lines 1-2, Alg~\ref{alg:pg_brd}):
We first construct a feasible initial association profile
$\bm{\pi}^{(0)}\in\{0,1,\ldots,|\mathcal{S}^{\mathrm{On}}|\}^{|\mathcal{U}^{\mathrm{On}}|}$,
where $\pi_{\mathbbm{j}}=0$ corresponds to local task execution at $u^{\mathrm{On}}_{\mathbbm{j}}$ and $\pi_{\mathbbm{j}}=\mathbbm{i}>0$ means that $u^{\mathrm{On}}_{\mathbbm{j}}$ offloads its task to SP $s^{\mathrm{On}}_{\mathbbm{i}}$.
For each SD $u^{\mathrm{On}}_{\mathbbm{j}}$, we randomly select an action from the set $\{0\}\cup\bm{C}_{\mathbbm{j}}$ and only accept it if the resulting profile satisfies the coverage constraint~\eqref{23b} and the subcarrier-capacity constraint~\eqref{23c}. 
This procedure ensures that the algorithm starts from a valid resource scheduling configuration with no constraint violations.

\noindent\textbf{Step 2. Feasible-action enumeration under coverage and subcarrier-capacity constraints} (lines 7-9, Alg.~\ref{alg:pg_brd}):
Given a current profile $\bm{\pi}$, we compute for each SP $s^{\mathrm{On}}_{\mathbbm{i}}$ the effective congestion load $y_{\mathbbm{i}}(\bm{\pi})$ and the physical computing load $\mathbbm{x}_{\mathbbm{i}}(\bm{\pi})$ as in~\eqref{eq:cong_vars_pi_final}, as well as the number of currently served SDs
$n_{\mathbbm{i}}(\bm{\pi}) = \big|\{\,\mathbbm{j}\mid \pi_{\mathbbm{j}}=\mathbbm{i}\,\}\big|$.
When it is SD $u^{\mathrm{On}}_{\mathbbm{j}}$'s turn to update its decision, we temporarily remove its own contribution and obtain the loads generated by all other SDs, denoted by $y^{-{\mathbbm{j}}}_{\mathbbm{i}}$, $\mathbbm{x}^{-{\mathbbm{j}}}_{\mathbbm{i}}$, and $n^{-{\mathbbm{j}}}_{\mathbbm{i}}$. These quantities represent the congestion level and computing load at each SP if $u^{\mathrm{On}}_{\mathbbm{j}}$ were not associated with any SP.
Given these congestion and load levels generated by all other SDs, we  construct the feasible action set for $u^{\mathrm{On}}_{\mathbbm{j}}$ as:
	\begin{equation}\label{feas_new}
    \resizebox{0.935\linewidth}{!}{$
		\mathcal{A}_{\mathbbm{j}}^{\mathrm{feas}}
		{=}
		\Big\{ a\in\{0\}\cup\mathcal{S}^{\mathrm{On}}
		  \Big|  
		a=0\ \text{or}\ \big(a=\mathbbm{i}>0,\ s^{\mathrm{On}}_{\mathbbm{i}}\in\bm{C}_{\mathbbm{j}},\ n^{-{\mathbbm{j}}}_{\mathbbm{i}} + 1 \le G_{\mathbbm{i}}\big)
		\Big\},$}
	\end{equation}
where $a=0$ corresponds to local task execution, and an offloading action $a=\mathbbm{i}>0$ is feasible only if SP $s^{\mathrm{On}}_{\mathbbm{i}}$ covers $u^{\mathrm{On}}_{\mathbbm{j}}$ and still has enough residual subcarrier-capacity to admit one more SD without violating its capacity $G_{\mathbbm{i}}$. This step explicitly removes any candidate action that would break either the coverage or the subcarrier-capacity constraints.

\noindent\textbf{Step 3. Asynchronous best response under the exact potential} (lines 10-18, Alg.~\ref{alg:pg_brd}):
In each iteration, SDs are visited one-by-one in a random order, and at most one SD is allowed to change its decision at a time. For a given SD $u^{\mathrm{On}}_{\mathbbm{j}}$ and current profile $\bm{\pi}$, we first evaluate its current potential-aligned payoff
$\widetilde U_{\mathbbm{j}}(\pi_{\mathbbm{j}}\mid\bm{\pi}_{-\mathbbm{j}})$ using~\eqref{eq:unified_utility_new}. Then, for each feasible action $a\in\mathcal{A}_{\mathbbm{j}}^{\mathrm{feas}}$, we temporarily set $\pi_{\mathbbm{j}}=a$, recompute the congestion and delay quantities for this SD, and obtain the resulting utility $\widetilde U_{\mathbbm{j}}(a\mid\bm{\pi}_{-\mathbbm{j}})$.
The SD then chooses a best response
$a^{\mathrm{best}}_{\mathbbm{j}}\in\arg\max_{a\in\mathcal{A}_{\mathbbm{j}}^{\mathrm{feas}}} \widetilde U_{\mathbbm{j}}(a\mid\bm{\pi}_{-\mathbbm{j}})$.
To avoid frequent switches caused by marginal gains, we introduce an acceptance threshold $\varepsilon>0$ and let $u^{\mathrm{On}}_{\mathbbm{j}}$ update its decision only if
$\widetilde U_{\mathbbm{j}}(a^{\mathrm{best}}_{\mathbbm{j}}\mid\bm{\pi}_{-\mathbbm{j}}) -
\widetilde U_{\mathbbm{j}}(\pi_{\mathbbm{j}}\mid\bm{\pi}_{-\mathbbm{j}}) > \varepsilon$.
Otherwise, it keeps its current action $\pi_{\mathbbm{j}}$ unchanged. Because the game admits the exact potential $\Phi^{\mathrm{G}}(\bm{\pi})$, every accepted update in this step strictly increases $\Phi^{\mathrm{G}}(\bm{\pi})$.

\noindent\textbf{Step 4. Convergence and resulting equilibrium} (lines 3-20, Alg.~\ref{alg:pg_brd}):
The above two steps (feasible-action enumeration and asynchronous best responses) are repeatedly applied by going over all SDs. When a complete pass over $\mathcal{U}^{\mathrm{On}}$ produces no accepted update, no SD can improve its potential-aligned payoff by more than $\varepsilon$ through a unilateral action change. Therefore, the algorithm reaches an $\varepsilon$-NE with respect to~\eqref{eq:unified_utility_new}; when $\varepsilon=0$, this reduces to a pure-strategy NE. {If $T_{\max}$ is reached before the convergence condition is satisfied, PG-BRD returns the final feasible profile obtained within the iteration budget; the equilibrium guarantee holds only upon convergence.}

\section{Evaluations}
We next conduct experiments to evaluate FUSION.
We consider a service provisioning region of size $1000\ \mathrm{m}\times 1000\ \mathrm{m}$, where ESs are randomly placed (HUs and MUs are in turn randomly placed within the coverage areas of ESs) and APs can be dispatched to assist overloaded ESs. Availability windows of ESs are randomly generated over a normalized time horizon to capture the temporal variability of computation demand. Other key parameters are as follows\cite{AUCTION,ACO,ECI1}: 
$f_i^{\mathrm{E}}\in[1,3]\times10^{12}$ CPU cycles/s (this is the aggregated clock speed of CPUs co-located at the ES),
$f_k^{\mathrm{A}}\in[1,3]\times10^{10}$ CPU cycles/s,
$f_m^{\mathrm{H,loc}}=f_n^{\mathrm{M,loc}}\in[1,2]\times10^{9}$ CPU cycles/s,
$d_m^{\mathrm{H}}\in[10,55]\ \mathrm{Mbit}$,
$d_n^{\mathrm{M}}\in[100,550]\ \mathrm{Mbit}$,
$G_i\in[6,8]$,
$G_k^\mathrm{A}\in[1,2]$,
$C_k^\mathrm{A}\in[8,10]$,
$W=50\ \mathrm{MHz}$,
$e_m^{\mathrm{H,tx}}=e_n^{\mathrm{M,tx}}\in[0.2,0.4]\ \mathrm{W}$,
$r_m^{\mathrm{H}} = 600 $ cycles/bit $ \times d_m^{\mathrm{H}}$, and $r_n^{\mathrm{M}} = 600 $ cycles/bit $ \times d_n^{\mathrm{M}}$.
Channel gains and data rates are generated using the distance-dependent path-loss models following~\cite{UAV-SNR}. 
Under these models, most links fall into an SNR range of $10$--$25$ dB, and UAV links achieve slightly higher average SNR due to their higher line-of-sight (LoS) probability.
Unless otherwise specified, the plotted simulations are averaged over 100 independent Monte Carlo trials.

\subsection{Benchmark Methods and Evaluation Metrics}
We evaluate FUSION from two complementary angles:
\textit{(i)} prediction performance on service demand forecasting module (Sec. \ref{Evaluation_LNN}); and
\textit{(ii)} system-level performance (Sec. \ref{Evaluation_FUSION}).

\noindent
\textit{Benchmarks and Metrics for Demand Forecasting:} We compare Pro-LNN against two widely used prediction models.

\noindent$\bullet$ \textbf{LSTM-driven demand prediction\cite{LSTM}:}
A standard long short-term memory (LSTM) network is applied to the  input sequences $\{\mathbf{z}_i\}$ and trained to predict future ES demands. This baseline uses a discrete-time recurrent neural network (RNN).

\noindent$\bullet$ \textbf{Transformer-driven demand prediction\cite{transformer}:}
A temporal Transformer-based predictor with self-attention is used over the historical context window. This serves as a powerful baseline, but typically requires more computing resources.

Performance evaluations consider the following metrics:

\noindent$\bullet$ \textbf{Accuracy:}
We report the root mean square error (RMSE), mean absolute error (MAE), and symmetric mean absolute percentage error (sMAPE) of the $H$-step forecasts $\{\hat N_i^{\mathrm{Dem},(\nu+1:\nu+H)}\}$ against the ground truth $\{N_i^{\mathrm{Dem},(\nu+1:\nu+H)}\}$, averaged over all ESs and all test instances.

\noindent$\bullet$ \textbf{Model footprint:}
We report the number of trainable parameters and the model size in megabytes using 32-bit floating-point (FP32) representation. These metrics characterize the storage and memory footprint of each predictor.

\noindent$\bullet$ \textbf{Robustness:}
To assess robustness under realistic noisy conditions, we compare the mean squared error (MSE) on clean test inputs with that on corrupted inputs containing randomly missing observations, additive perturbations, and occasional demand bursts, and report the  performance drop $\Delta\text{MSE}$.

\noindent
\textit{System-Level Benchmarks and Metrics:} We then compare FUSION with several benchmark methods enumerated below:

\noindent$\bullet$ \textbf{Pure online scheduling (PurOnline)\cite{RW-PG-Zhou}:}
This baseline removes the offline stage of FUSION: ESs and APs do not establish offline contracts, and UAV service routes are not pre-planned. HU/MU tasks are scheduled solely by PG-BRD using the currently available ES/AP resources.

\noindent$\bullet$ \textbf{Two-stage resource provisioning without PG-based online scheduling (FUSION\_NoPG):}
This variant keeps the same offline stage as FUSION, i.e., Off-AIC$^2$ with eACO-VRP is still used to form offline AP--ES contracts and pre-planned routes. However, the online stage abandons PG-BRD: HU/MU tasks are randomly assigned to feasible ESs/APs within coverage subject to resource capacity and deadline constraints.

\noindent$\bullet$ \textbf{Two-stage resource provisioning with random offline contracts (FUSION\_Random):}
This baseline uses the same online stage as FUSION and still relies on PG-BRD for congestion-aware online resource scheduling. However, its offline stage no longer uses Off-AIC$^2$ or eACO-VRP; instead, each AP is randomly associated with a subset of ESs and moves along a random feasible route until its capacity is exhausted. 

\noindent$\bullet$ \textbf{Two-stage resource provisioning without ES spatio-temporal features (FUSION\_NoST):}
This baseline preserves the two-stage structure of FUSION, but its offline auction ignores ES-specific time windows and travel-time constraints. Specifically, AP--ES pairs are determined only by simple price/volume trading without using eACO-VRP, while the online stage still uses PG-BRD for task scheduling.

To assess performance, we use the following metrics:

\noindent$\bullet$ \textbf{Social welfare (SW):}
SW combines the offline welfare $U^{\mathrm{SW}}$ in~\eqref{equ. SW} and the online welfare $\mathcal{W}(\bm{\pi})$ defined in Sec.~\ref{subsec:potential}.

\noindent$\bullet$ \textbf{Delay incurred by interactions between SPs and SDs (DoI):}
DoI captures the total decision-making latency (in milliseconds) from control-message exchanges between ESs/APs (SPs) and HUs/MUs (SDs) across both offline and online stages. To compute DoI, following \cite{ECI1,ECI2,ECI3}, we assume that the per-interaction delay lies in the range of $[1,15]$ ms.

\noindent$\bullet$ \textbf{Energy consumption incurred by interactions between SPs and SDs (ECoI):}
ECoI measures signaling-related energy consumption during decision making.
Following~\cite{ECI1,ECI2}, the ES transmit power is set within $[6,20]$ W.
Specifically, ECoI is computed by summing the energy consumed by all control-message transmissions, with interaction delays used as the transmission durations.

\noindent$\bullet$ \textbf{Near-truthfulness and individual rationality of Off-AIC$^2$:}
We examine participant utilities under unilateral bid/ask deviations, with utilities evaluated using true valuations and costs. The deviation tests assess the incentive to misreport, while comparisons between true costs/values and final rewards/payments assess individual rationality for winning participants.

\subsection{Performance Evaluation on Prediction}\label{Evaluation_LNN}
\subsubsection{Robustness vs. Cost}\label{subsec:B-Sim}
Table~\ref{tab:lnn_robust} compares prediction error and model size under clean and corrupted inputs. The corrupted setting includes 30\% randomly missing values, Gaussian noise with standard deviation $\sigma=0.05$, and burst shocks of length 1--3. On clean data, LSTM achieves the lowest MSE of 0.119, compared with 0.467 for Transformer and 0.459 for Pro-LNN. Pro-LNN uses 11.1k parameters, whereas LSTM uses 52.8k. Under corrupted inputs, the MSE increases to 0.983 for LSTM, 0.781 for Transformer, and 0.762 for Pro-LNN. Thus, Pro-LNN has a smaller model footprint and lower prediction error under the tested corruptions, although LSTM is more accurate on clean inputs.
\begin{table}[t]
	\centering
	\small
	\caption{Robustness--complexity comparison of different forecasters on synthetic non-stationary series. The ``corrupt'' case injects 30\% random missing values, Gaussian noise with standard deviation $\sigma=0.05$, and random burst shocks.}
	\label{tab:lnn_robust}
	\setlength{\tabcolsep}{0.6mm}
	\begin{tabular}{lccccc}
		\toprule
		\textbf{Model} 
		& {\#Params} 
		& $\text{MSE}_{\text{clean}}$ 
		& $\text{MSE}_{\text{corrupt}}$ 
		& $\Delta\text{MSE}$ 
		& Rel.\ increase \\
		\midrule
		LSTM        
		& 52800 
		& 0.119 
		& 0.983 
		& 0.864 
		& +726.3\% \\
		Transformer 
		& 69500
		& 0.467 
		& 0.781 
		& 0.314 
		& +67.3\%  \\
		\rowcolor[gray]{0.9}	Pro-LNN     
		& 11100 
		& 0.459 
		& 0.762 
		& 0.303 
		& +66.1\%  \\
		\bottomrule
	\end{tabular}
\end{table}

\subsubsection{Real-World Dataset Test} 
We evaluate Pro-LNN on the UCI Electricity Load Diagrams 2011--2014 dataset~\cite{UCI} using two load series, MT\_005 and MT\_145. In this univariate forecasting task, a history of $T_{\mathrm{his}}=48$ measurements is used to predict the next $H=10$ time steps. The electricity-load traces serve as temporal proxies for ES-side demand, allowing us to evaluate forecasting performance on real-world load variations.
Quantitative results are summarized in Table~\ref{tab:uci-accuracy} and Table~\ref{tab:uci-robust}. As shown in Table~\ref{tab:uci-accuracy}, Pro-LNN achieves comparable prediction accuracy to LSTM and Transformer, while being much more compact. In particular, Pro-LNN attains an RMSE of 0.1234, close to that of the LSTM (0.1121) and the Transformer (0.1158), but only uses 11{,}114 parameters and has a 0.046\,MB storage/memory footprint. The parameter counts and model sizes indicate a smaller storage footprint for Pro-LNN, which is relevant to deployment on resource-constrained edge platforms. To evaluate robustness, we inject  corrupted inputs, following the same configuration as in the synthetic experiments detailed in Sec.~\ref{subsec:B-Sim}.
As reported in Table~\ref{tab:uci-robust}, Pro-LNN exhibits the smallest performance degradation among all models: its MSE increases from 0.0152 (clean) to 0.0467 (corrupt), yielding a $\Delta$MSE of 0.0315, which is lower than that of both LSTM (0.0359) and Transformer (0.0458). These observations indicate that Pro-LNN is highly parameter-efficient and demonstrates strong robustness to missing data and measurement noise,
 which in turn supports the suitability of Pro-LNN for forecasting future ES resource demands in the offline stage of FUSION.

\begin{table}[!t]
	\centering
	\small
	\caption{Prediction accuracy and model footprint on the UCI dataset (load series MT\_005 and MT\_145, input length $T_{\mathrm{his}}=48$, prediction horizon $H=10$).}
	\setlength{\tabcolsep}{1.2mm}
	\begin{tabular}{lccccc}
		\toprule
		\textbf{Model}      & \#Params & Size (MB) & RMSE   & MAE    & sMAPE  \\
		\midrule
		LSTM         & 52842    & 0.204     & 0.1121 & 0.0816 & 0.1831  \\
		Transformer  & 69482    & 0.400     & 0.1158 & 0.0824 & 0.1864 \\
		\rowcolor[gray]{0.9}Pro-LNN      & 11114    & 0.046     & 0.1234 & 0.0906 & 0.1983\\
		\bottomrule
	\end{tabular}\label{tab:uci-accuracy}
\end{table}

\begin{table}[!t]
	\centering
	\small
	\caption{Robustness to missing and noisy inputs on UCI Electricity (30\% random missing values, added Gaussian noise with standard deviation of $\sigma=0.05$, and random burst shocks).}
	\begin{tabular}{lccc}
		\toprule
		\textbf{Model}        & MSE (clean) & MSE (corrupt) & $\Delta$MSE \\
		\midrule
		LSTM         & 0.0126      & 0.0485        & 0.0359      \\[-1pt]
		Transformer  & 0.0134      & 0.0592        & 0.0458      \\[-1pt]
		\rowcolor[gray]{0.9}	Pro-LNN      & 0.0152      & 0.0467        & 0.0315      \\
		\bottomrule
	\end{tabular}
	\label{tab:uci-robust}
\end{table}

\begin{figure*}[htb]
	\centering
	
	\subfigure[]{
		\includegraphics[width=0.22\linewidth]{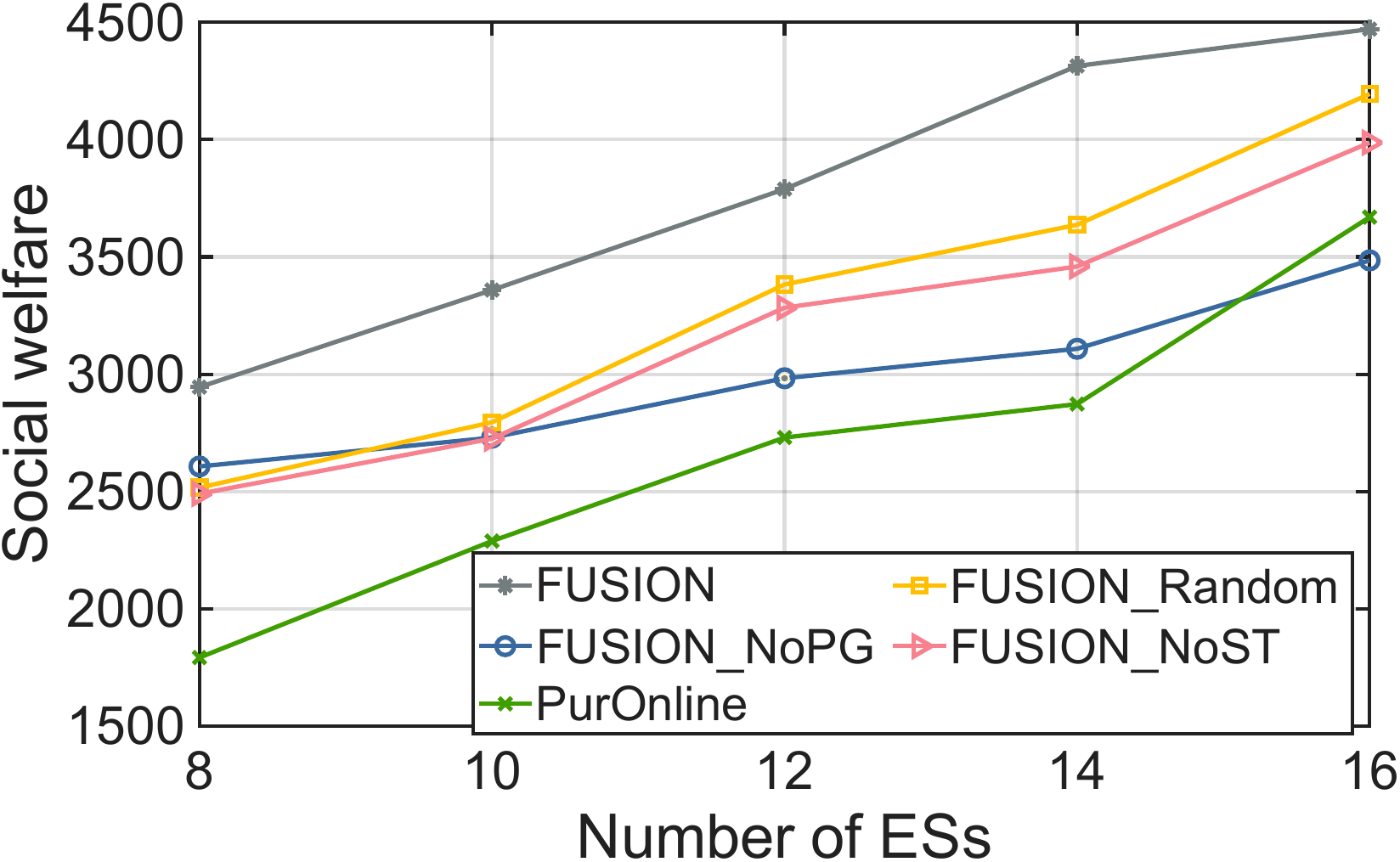}
		\label{fig_SWa}
	}
	\hfill
	\subfigure[]{
		\includegraphics[width=0.22\linewidth]{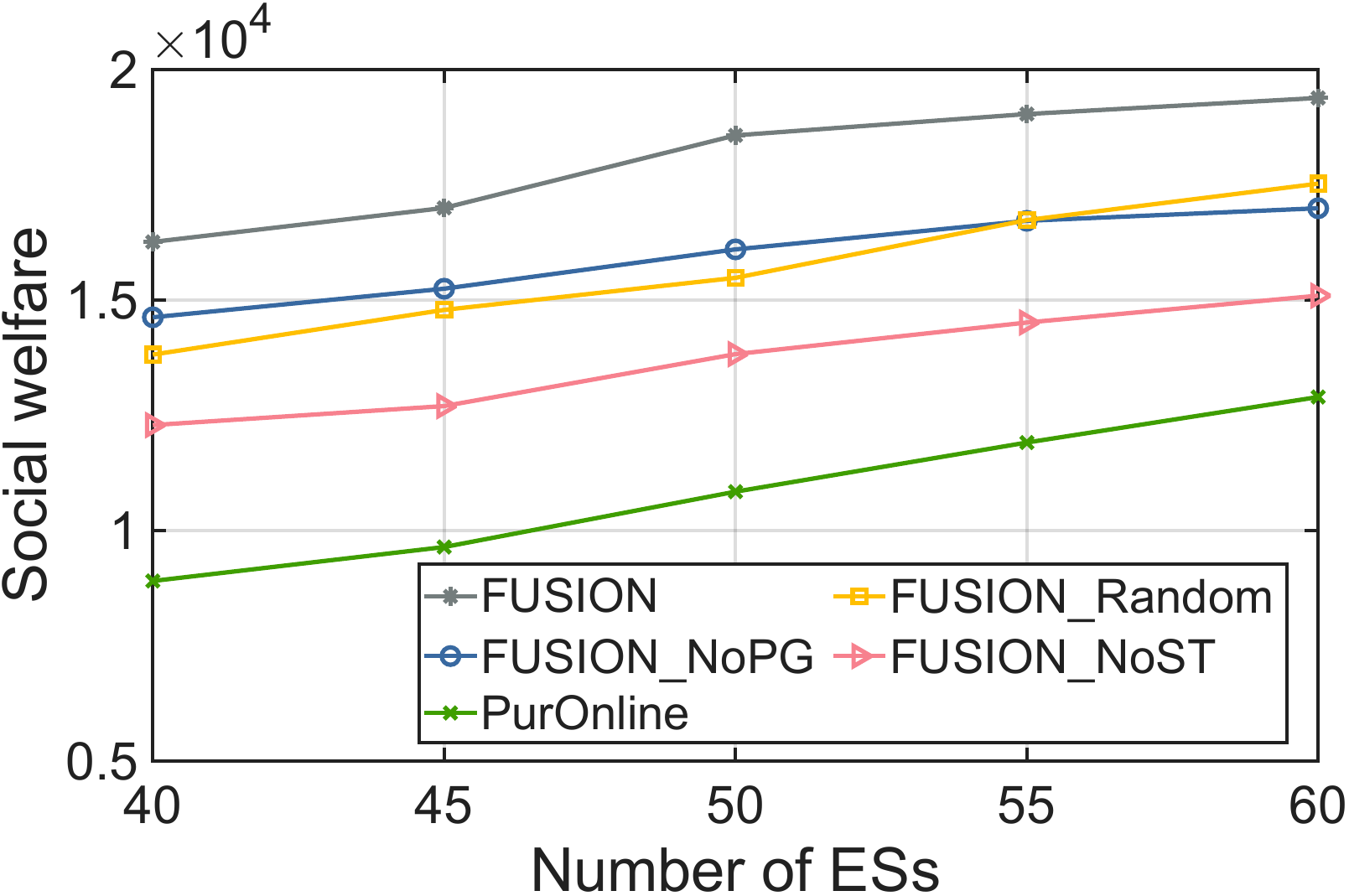}
		\label{fig_SWb}
	}
	\hfill
	\subfigure[]{
		\includegraphics[width=0.22\linewidth]{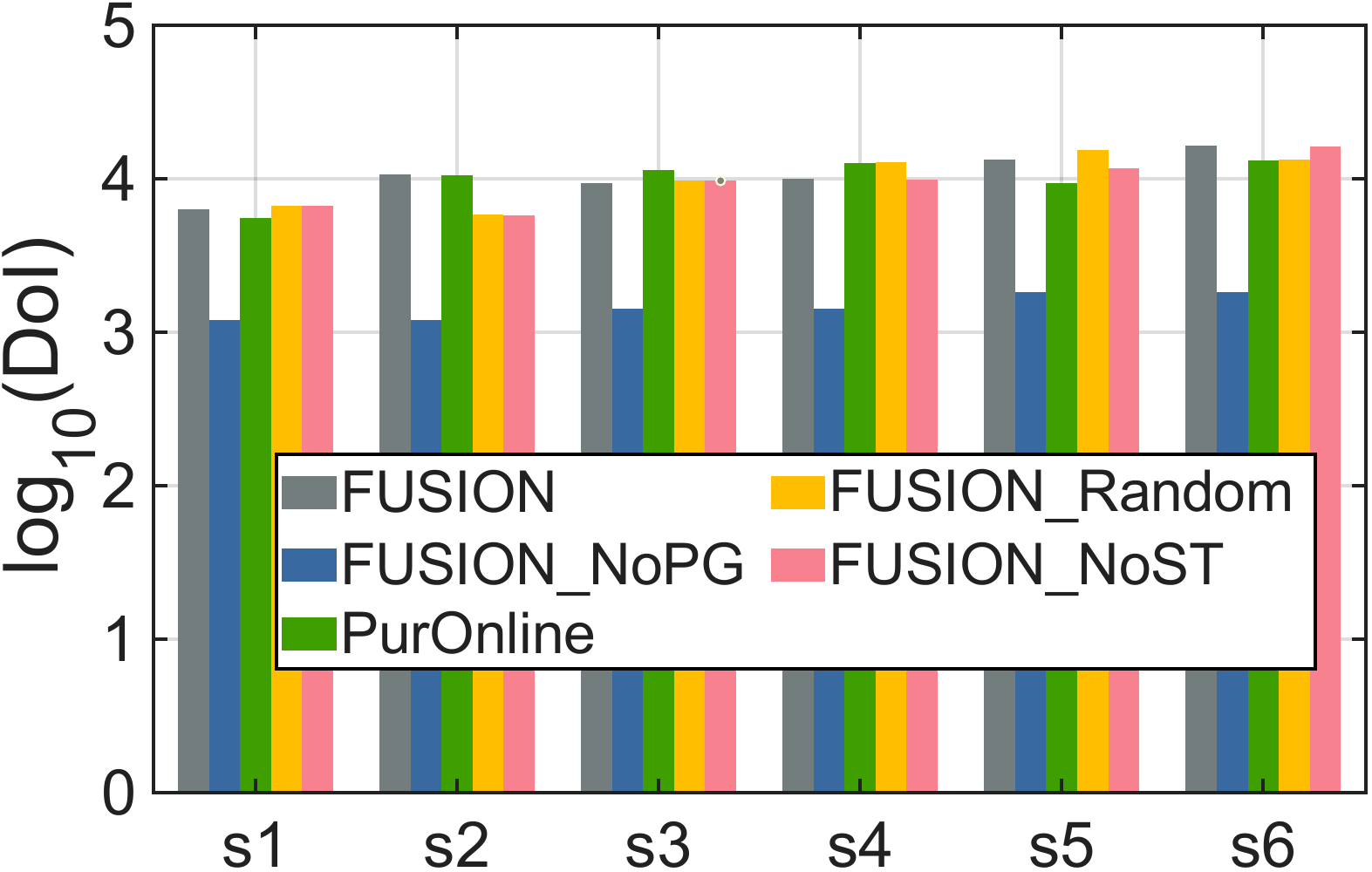}
		\label{fig:doi}
	}
	\hfill
	\subfigure[]{
		\includegraphics[width=0.22\linewidth]{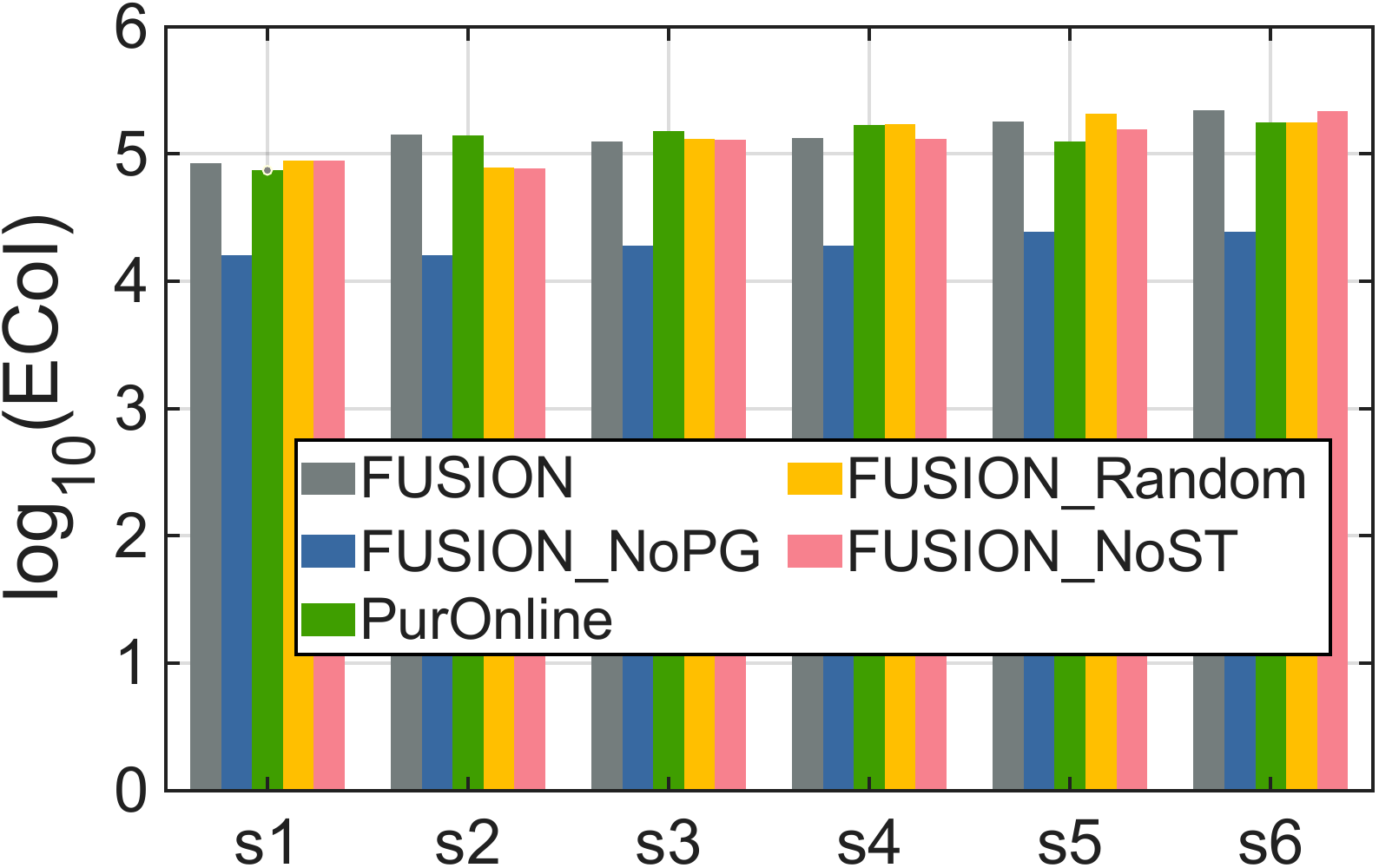}
		\label{fig:ecoi}
	}
	\hfill
	\subfigure[]{
		\includegraphics[width=0.215\linewidth]{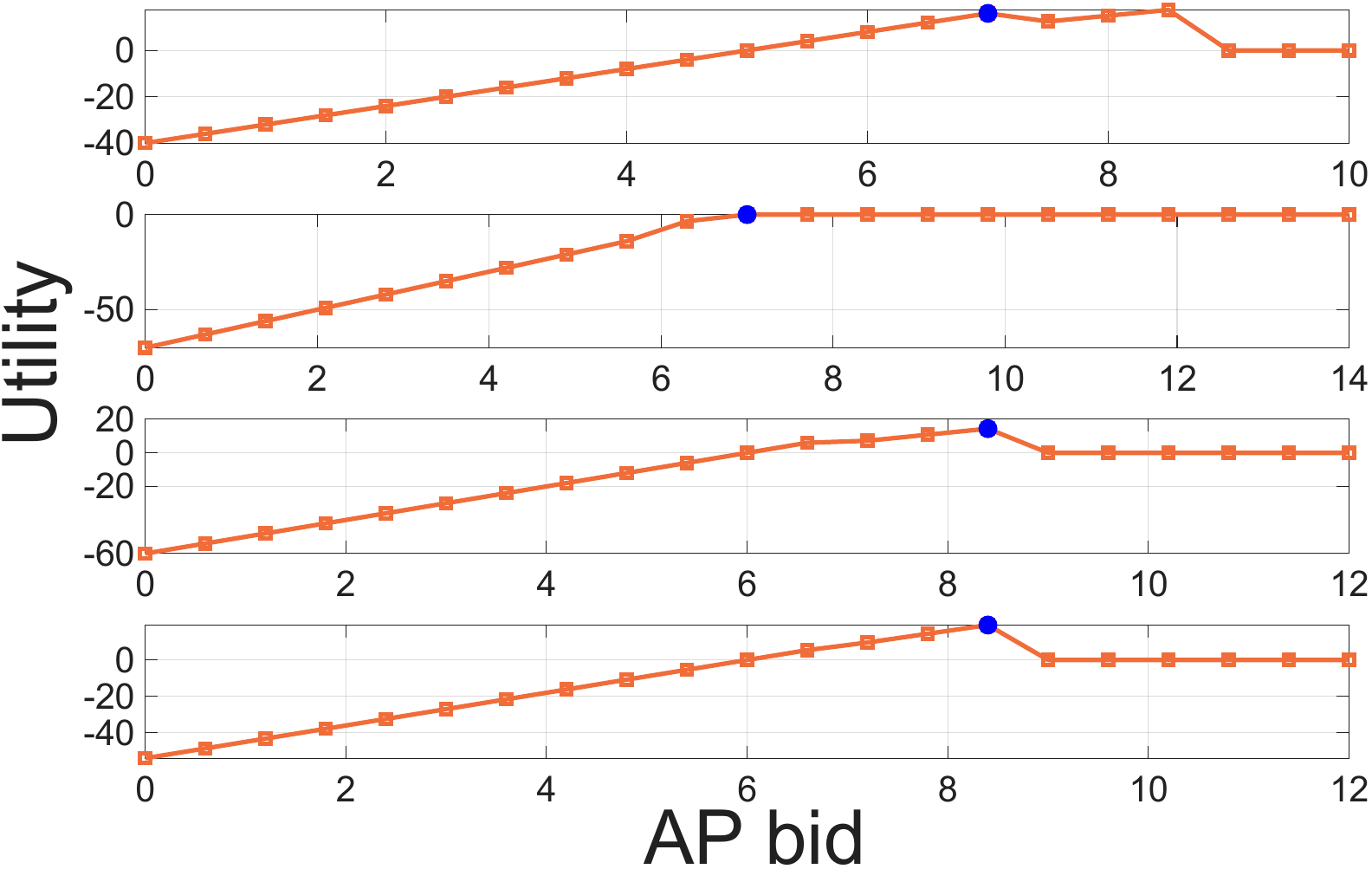}
		\label{fig:truth_agp}
	}
	\hfill
	\subfigure[]{
		\includegraphics[width=0.215\linewidth]{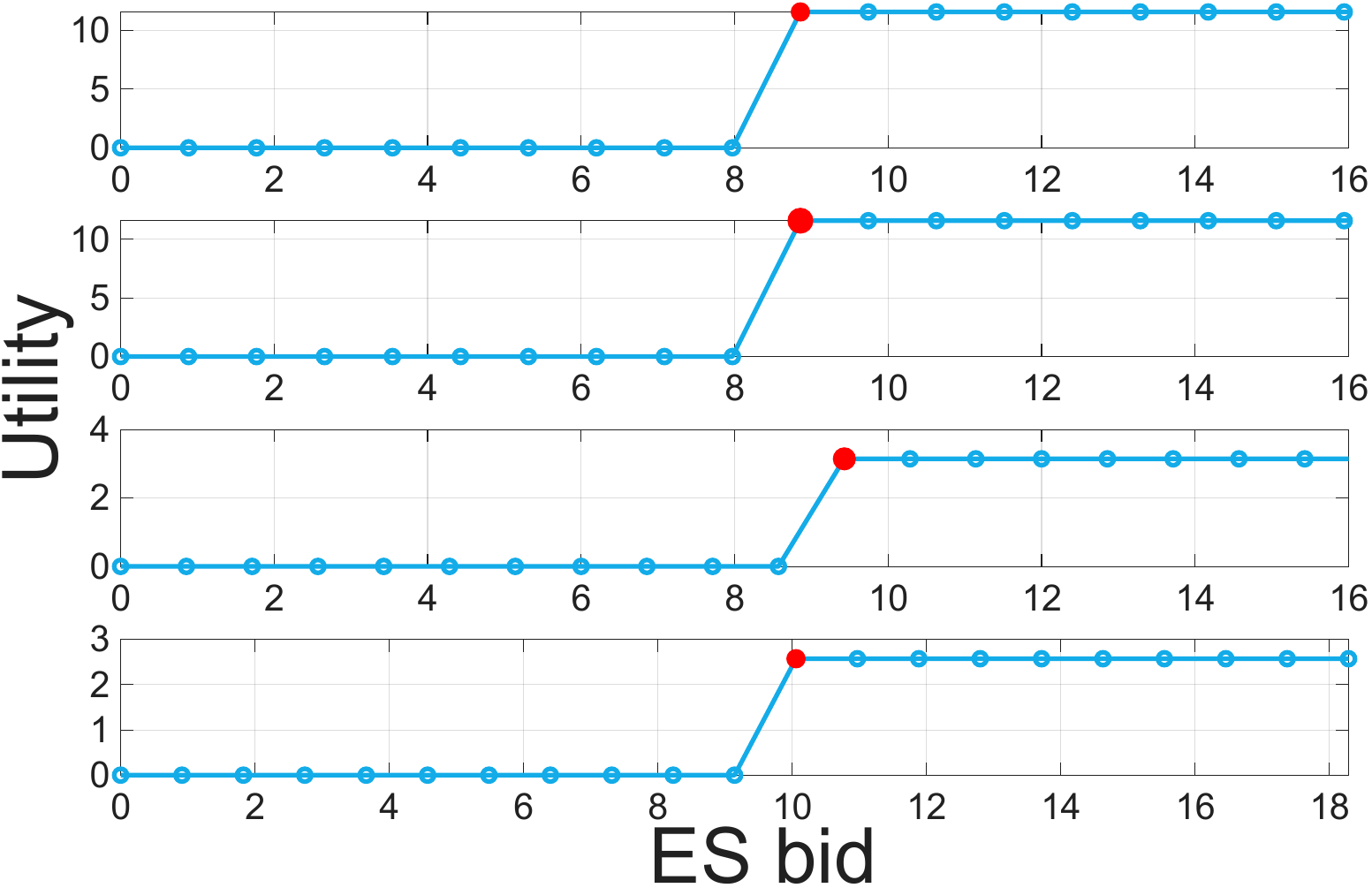}
		\label{fig:truth_es}
	}
	\hfill
	\subfigure[]{
		\includegraphics[width=0.245\linewidth]{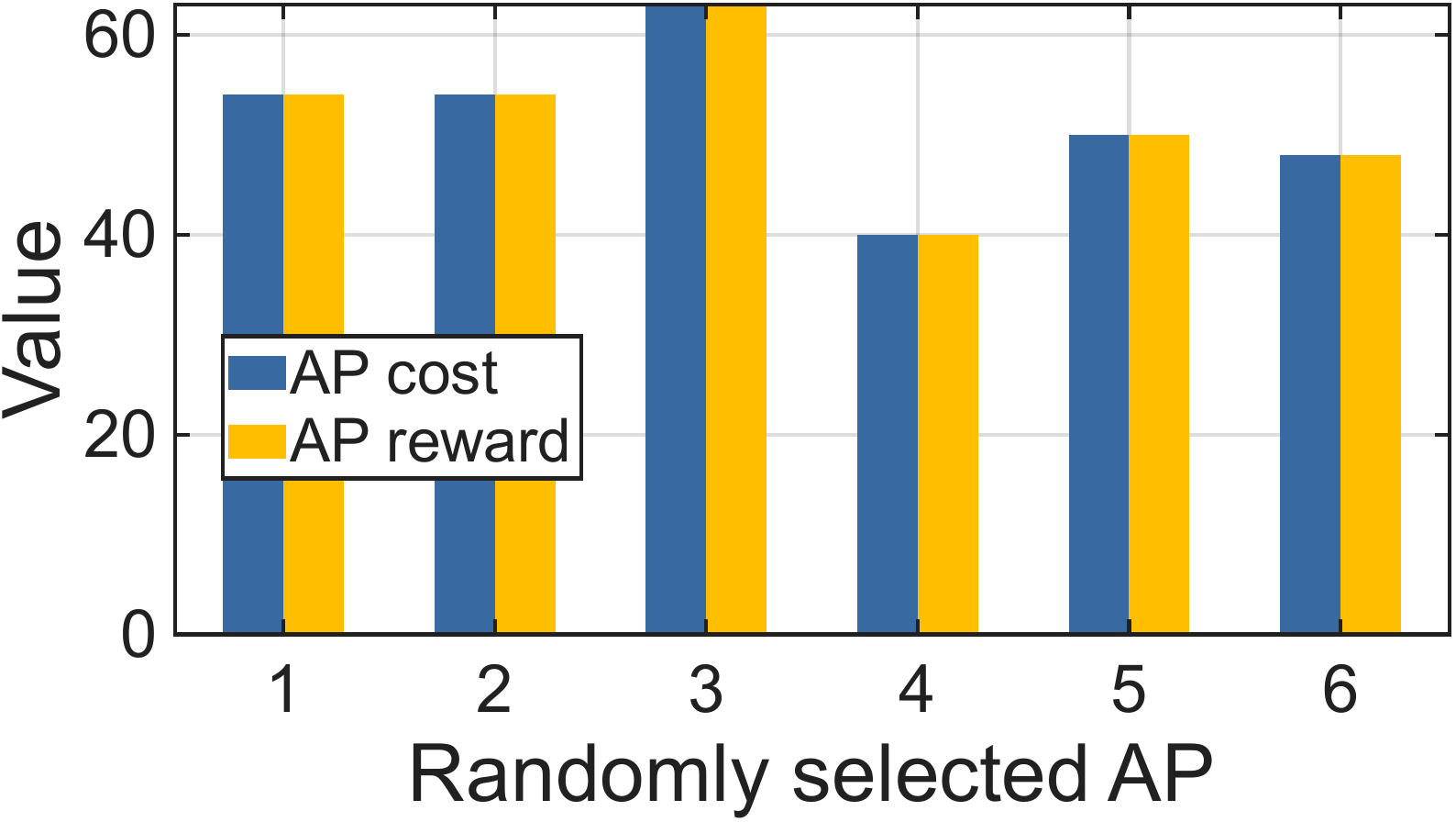}
		\label{fig:ir_agp}
	}
	\hfill
	\subfigure[]{
		\includegraphics[width=0.245\linewidth]{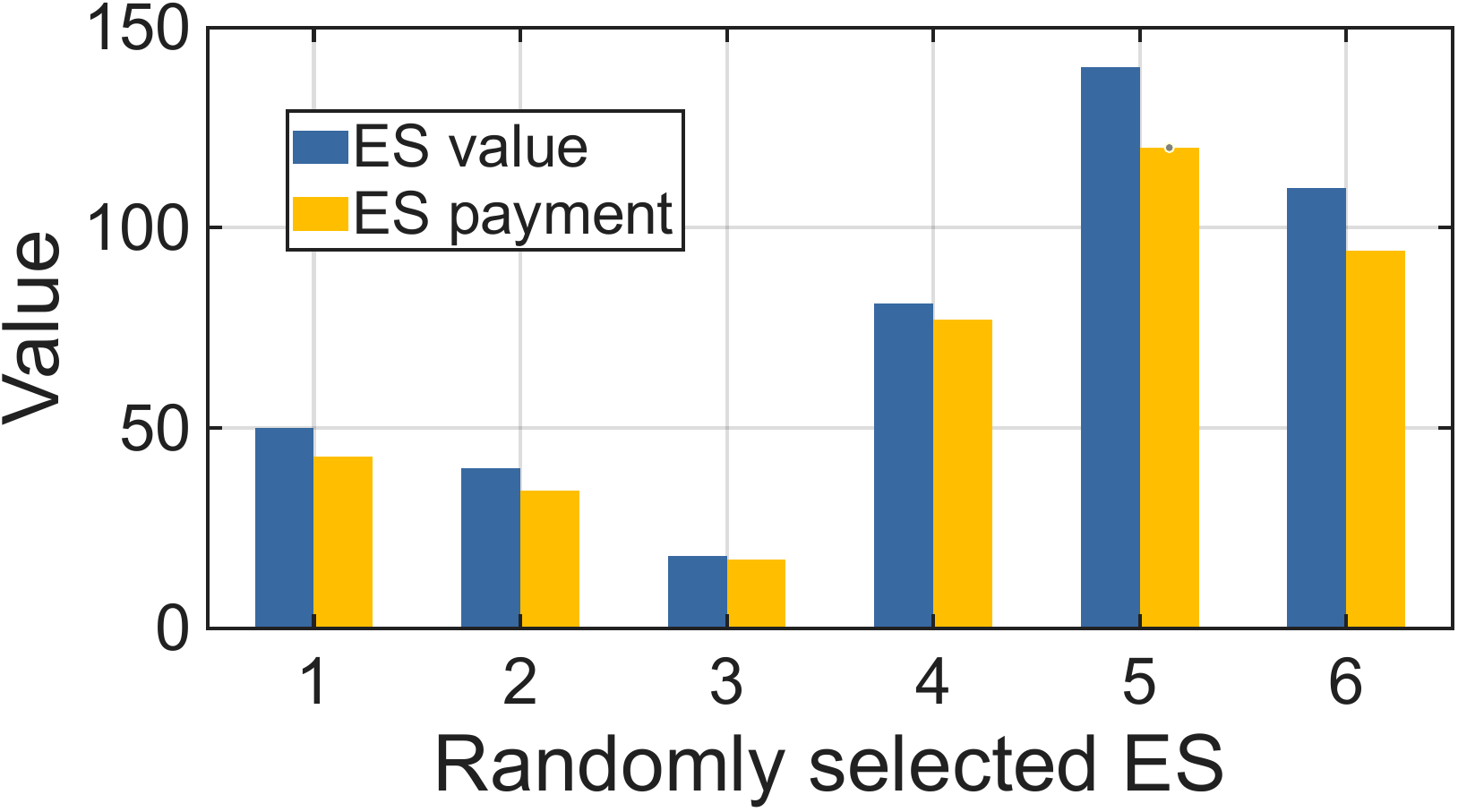}
		\label{fig:ir_es}
	}
	
	\caption{Performance comparisons and economic property analyses:
		(a)--(b) SW under different problem sizes
		(80 HUs, 80 MUs, and 5 APs in (a), and 400 HUs, 400 MUs, and 35 APs in (b));
		(c)--(d) DoI and ECoI under different problem sizes
		(s1--s6 are set as $\{75/75/15/8\}$, $\{75/75/20/8\}$,
		$\{90/90/20/10\}$, $\{90/90/25/10\}$,
		$\{115/115/25/12\}$, and $\{115/115/30/12\}$);
		(e)--(f) near-truthfulness; and
		(g)--(h) individual rationality.}
	\label{fig:all_performance}
\end{figure*}

\subsection{Performance Evaluation on FUSION}\label{Evaluation_FUSION}
To emulate heterogeneous, time-varying online demands across different ESs, we adopt the real-world UCI Electricity Load Diagrams dataset~\cite{UCI}, select several distinct load profiles, and map their normalized power-consumption traces to the per-time-slot resource demand of each ES. We then train Pro-LNN on these profiles to forecast the per-ES demand\footnote{We employ the UCI dataset as a trace to capture realistic temporal patterns of aggregate demand at each ES. Pro-LNN is trained on these time series and used to forecast the demand $N_i^{\mathrm{Dem},(\nu)}$ of ES $s_i$ over future trading rounds; the forecasted demand profiles drive the time-varying SD task-arrival intensities at each ES, while the total population sizes of SDs and SPs are kept fixed to control the overall network scale.}, and then evaluate FUSION as follows (i.e., in Sec. \ref{Evaluation_SW}--Sec. \ref{Evaluation_IR}).

\subsubsection{Social Welfare (SW)}\label{Evaluation_SW}
Figs.~\ref{fig_SWa}--\ref{fig_SWb} report SW as the number of ESs increases from 8 to 16. They correspond to two representative network configurations: 80 HUs, 80 MUs, and 5 APs in Fig.~\ref{fig_SWa}, and 400 HUs, 400 MUs, and 35 APs in Fig.~\ref{fig_SWb} to capture different network scales. In Fig.~\ref{fig_SWa}, SW increases for all methods with more ESs, since they provide more computing resources and service opportunities. However, the growth rates differ across schemes. FUSION consistently achieves the highest SW because its offline stage employs Off-AIC$^2$ to match AP resources with ESs of high revenue potential under spatio-temporal constraints, while its online stage adopts a PG-BRD scheduler to perform congestion-aware task allocation among HUs and MUs, supporting more effective use of the available resources. PurOnline, which lacks an offline stage and relies solely on an online potential game over real-time ES resources, achieves the lowest SW due to the absence of proactive UAV repositioning. FUSION\_NoPG shares the same offline design as FUSION but replaces PG-BRD with random online assignment, resulting in unresolved congestion and a notable SW degradation, highlighting the importance of the online PG. FUSION\_Random adopts random AP--ES association in the offline stage while retaining PG-BRD online, yielding intermediate performance and indicating that intelligent online scheduling alone cannot compensate for the lack of revenue- and spatio-temporal-aware offline planning. FUSION\_NoST ignores ES spatio-temporal characteristics in the offline stage and relies on a price-driven auction with PG-BRD online; it outperforms PurOnline and FUSION\_Random but still exhibits a consistent gap to FUSION. Fig.~\ref{fig_SWb} shows qualitatively similar trends under the larger-scale configuration: when ESs increase, SW rises for all methods, with FUSION consistently outperforming benchmarks. Results show that offline route-aware contracting and online congestion-aware scheduling both contribute to SW gains across tested network scales.

\subsubsection{Evaluation on Interaction Overhead}
To quantitatively assess interaction efficiency, we analyze DoI (Fig. \ref{fig:doi}) and ECoI (Fig. \ref{fig:ecoi}). The x-axis (s1–s6) denotes different system configurations (various HU/MU/ES/AP combinations detailed in the caption of the figure), and the y-axis uses a logarithmic scale to capture results across magnitudes.
Across all settings, DoI and ECoI remain within the same order of magnitude: FUSION is not always the best in DoI/ECoI, but the gaps with baselines (FUSION\_NoPG, FUSION\_Random, FUSION\_NoST, PurOnline) are rather small. This arises from FUSION's offline contracting stage, which exposes more candidate SPs and resource options, prompting the PG-BRD online scheduler to explore a larger strategy space, slightly increasing interaction rounds and per-transaction overhead. Auction and route planning are completed before online service execution. Their amortized overhead is included in the DoI and ECoI reported in Figs.~\ref{fig:doi}--\ref{fig:ecoi}. Combined with SW results (Figs. \ref{fig_SWa}-\ref{fig_SWb}), we conclude that FUSION maintains comparable DoI/ECoI to baselines while substantially improving SW through more efficient air--ground resource utilization.

\subsubsection{Near-Truthfulness and Individual Rationality}\label{Evaluation_IR}

We evaluate the near-truthfulness and individual rationality of Off-AIC$^2$ with 75 HUs, 75 MUs, 12 ESs, and 7 APs. Figs.~\ref{fig:truth_agp} and~\ref{fig:truth_es} plot seller (AP) and buyer (ES) utilities against reported asks/bids, respectively, while Figs.~\ref{fig:ir_agp} and~\ref{fig:ir_es} verify individual rationality by comparing true costs/values with payments. Fig.~\ref{fig:truth_agp} shows four independent market instances, each with one randomly selected AP. The horizontal axis denotes the reported ask and the vertical axis the resulting utility (revenue minus true cost), with the blue dot marking the true ask. As the ask increases, the AP's utility rises while it remains selected, then drops sharply beyond its critical region when the contract is lost. Across the instances, utility is maximized near the true ask, indicating near-truthful behavior. Fig.~\ref{fig:truth_es} presents the corresponding ES results. Each subplot selects one ES, with the horizontal axis representing the reported mean bid and the red dot indicating the true valuation. Below a threshold close to the true value, the ES obtains zero utility; after crossing it, the ES wins, while further bid increases yield negligible gains. Thus, the best-response region consistently concentrates around the true valuation. Step~3 of Off-AIC$^2$ approximates critical-value payments and rewards through binary search, and the utility curves reflect the resulting pricing procedure and route-aware allocation produced by eACO-VRP. Appx.~B.2 states the bounded-error assumptions used in the near-truthfulness analysis. The negative utilities in parts of Figs.~\ref{fig:truth_agp} and~\ref{fig:truth_es} arise from hypothetical reports evaluated using true valuations and costs. Figs.~\ref{fig:ir_agp} and~\ref{fig:ir_es} further confirm individual rationality: for randomly selected winning APs, the total reward determined by Off-AIC$^2$ is no smaller than the corresponding true total cost, while for selected ESs, the true valuation always exceeds the final payment.

\subsubsection{Additional DRL Comparison and Decomposition Gap}
Appx.~F further compares PG-BRD with representative DRL-based online schedulers, including DQN, PPO, and MADDPG, under identical offline service preparation. The results show that PG-BRD achieves higher social welfare and timely-service performance, whereas the DRL-based variants exhibit advantages in served ratio, energy cost, and online runtime.
Appx.~G evaluates the two-stage FUSION framework against small-scale full-information benchmarks. The results show that FUSION maintains a limited performance gap relative to the exact solutions or solver upper bounds while requiring substantially lower runtime, demonstrating a favorable performance--complexity tradeoff in the tested instances.

\section{Conclusion and Future Work}
{This work proposed FUSION, a forecasting-driven two-stage framework for service provisioning in air--ground integrated networks with human--machine coexistence. The offline stage combines Pro-LNN demand forecasting, eACO-VRP route planning, and Off-AIC$^2$ contracting to prepare AP-assisted service resources. The online stage uses PG-BRD to update SD--SP assignments under congestion and capacity constraints. When the convergence condition is satisfied, PG-BRD returns an $\varepsilon$-NE, which becomes a pure-strategy NE when $\varepsilon=0$. Experiments show higher social welfare with interaction delay and signaling energy overheads comparable to the considered baselines. Future work will investigate event-triggered AP replanning and adaptive service coordination under abrupt demand, mobility, and topology changes.}

\newpage
\clearpage
\appendices
\setcounter{equation}{24}
\section{Key Notations}
Key notations in this paper are summarized in Table \ref{Tab-notations}.

\begin{table*}[b!]
	{\footnotesize
		\caption{\footnotesize{Key notations}}
		\begin{center}\label{Tab-notations}
			\begin{tabular}{|l|l|}
				\hline
				\textbf{Notation} & \textbf{Explanation} \\ \hline
				$\bm{\mathcal{S}}$, $\bm{\mathcal{U}}^{\mathrm{H}}$, $\bm{\mathcal{U}}^{\mathrm{M}}$, $\bm{\mathcal{V}}$ 
				& \makecell[l]{Sets of edge servers (ESs), human users (HUs), machine users (MUs), and vehicle--UAV agent pairs (APs)} \\ \hline
				
				$s_i$, $u_m^{\mathrm{H}}$, $u_n^{\mathrm{M}}$, $v_k$ 
				& \makecell[l]{The $i$-th ES, the $m$-th HU, the $n$-th MU, and the $k$-th AP} \\ \hline
				
				$N_i^{\mathrm{Dem},(\nu)}$, $\hat N_i^{\mathrm{Dem},(\nu+h)}$ 
				& \makecell[l]{Actual and LNN-predicted task demand of ES $s_i$ at trading rounds $\nu$ and $\nu\!+\!h$} \\ \hline
				
				{
					$l_i^{\mathrm{E}}$, $l_k^{\mathrm{A}} \langle t\rangle$}
					& \makecell[l]{{Location of ES $s_i$ and ground location of AP $v_k$ in time $t$}} \\ \hline

				$r_m^{\mathrm{H}}$, $r_n^{\mathrm{M}}$, $r_{\mathbbm{j}}$ 
				& \makecell[l]{Computation workload (CPU cycles) of HU $u_m^{\mathrm{H}}$,  MU $u_n^{\mathrm{M}}$, and generic SD $u^{\mathrm{On}}_{\mathbbm{j}}$} \\ \hline
				
				$d_m^{\mathrm{H}}$, $d_n^{\mathrm{M}}$, $d_{\mathbbm{j}}$ 
				& \makecell[l]{Input data size (bits) of HU tasks, MU tasks, and generic SD tasks} \\ \hline
				
				$f_i^{\mathrm{E}}$, $f_{\mathbbm{i}}^{\mathrm{On}}$, $f_m^{\mathrm{H,loc}}$, $f_n^{\mathrm{M,loc}}$, $f_k^{\mathrm{A}}$ 
				& \makecell[l]{Computing capabilities of ES $s_i$, online SP $s_{\mathbbm{i}}^{\mathrm{On}}$, local HU/MU devices, and UAVs dispatched by AP $v_k$} \\ \hline
				
				$e_m^{\mathrm{H,tx}}$, $e_n^{\mathrm{M,tx}}$ 
				& \makecell[l]{Uplink transmit power of HU $u_m^{\mathrm{H}}$ and MU $u_n^{\mathrm{M}}$} \\ \hline
				
				$e_i^{\mathrm{E}}$
& \makecell[l]{Local computing power consumption of ES $s_i$} \\ \hline
				
				$e_k^{\mathrm{comp}}$, $e_k^{\mathrm{move}}$, $e_k^{\mathrm{fly}}$
& \makecell[l]{Power consumption associated with UAV computation, ground movement, and UAV flight for AP $v_k$} \\ \hline
				
				$\tau_n^{\mathrm{ddl}}$, $\tau^\mathrm{ddl}_{\mathbbm{j}}$
& \makecell[l]{Task deadline of an MU, indexed by $n$ in the MU model and by $\mathbbm{j}$ in the online formulation} \\ \hline
				
				$C_k^{\mathrm{A}}$, $G_k^{\mathrm{A}}$, $Q_k^{\mathrm{A}}$
				& \makecell[l]{Number of UAVs carried by AP $v_k$, maximum number of SDs served by each UAV,\\ and total AP service capacity} \\ \hline
				
				$N_{i,k}^{\mathrm{Trad}}$ 
				& \makecell[l]{Quantity of traded computing resources between ES $s_i$ and AP $v_k$ in an offline contract} \\ \hline
				
				$x_{i,k}^{\mathrm{Off}}$, $\bm{X}^{\mathrm{Off}}$
& \makecell[l]{Binary offline assignment between ES $s_i$ and AP $v_k$; $\bm{X}^{\mathrm{Off}}$ collects all $x_{i,k}^{\mathrm{Off}}$} \\ \hline
				
				$p_{i,k}^{\mathrm{ES}}$, $r_{i,k}^{\mathrm{ES}}$, $\overline{p}_i$
				& \makecell[l]{Unit payment collected from ES $s_i$, unit reward paid to AP $v_k$, and historical average user service price \\at ES $s_i$} \\ \hline
				
				$U_i^{\mathrm{ES,Off}}$, $U_k^{\mathrm{A,Off}}$, $U^{\mathrm{Auc,Off}}$ 
				& \makecell[l]{Utility of ES $s_i$, utility of AP $v_k$, and utility of the auctioneer in the offline stage} \\ \hline
				
				$U^{\mathrm{SW}}$ 
				& \makecell[l]{Total SW of the offline auction market (sum of utilities of ESs, APs, and auctioneer)} \\ \hline
				
				$\bm{\mathcal{G}}_k(\bm{\mathcal{N}}_k,\bm{\mathcal{E}}_k)$, $\mathbb{K}$ 
				& \makecell[l]{Directed complete graph and ant set used by eACO-VRP for routing AP $v_k$} \\ \hline

					$\varphi_{r,s}$, $\eta_{r,s}$, $pr^{\mathbbm{k}}_{r,s}$, $\Delta U^\mathrm{A}_{r,s}$
					& \makecell[l]{Pheromone level, distance-based heuristic desirability, transition probability of ant $\mathbbm{k}$, and AP-side utility\\ increment on edge $(\mathrm{n}_r,\mathrm{n}_s)$ in eACO-VRP} \\ \hline

				$\mathcal{S}^{\mathrm{On}}$, $\mathcal{U}^{\mathrm{On}}$, $\bm{\mathcal{U}}^{\mathrm{H}}$, $\bm{\mathcal{U}}^{\mathrm{M}}$ 
				& \makecell[l]{Sets of online service providers (SPs), online service demanders (SDs), HUs, and MUs} \\ \hline
				
				$x_{\mathbbm{j},0}^{\mathrm{On}}$, $x_{\mathbbm{j},\mathbbm{i}}^{\mathrm{On}}$, $\bm{X}^{\mathrm{On}}_{\mathbbm{j}}$ 
				& \makecell[l]{Indicator of local execution, indicator of scheduling SD $u^{\mathrm{On}}_{\mathbbm{j}}$ to SP $s^{\mathrm{On}}_{\mathbbm{i}}$, and the one-hot scheduling vector\\ of SD $u^{\mathrm{On}}_{\mathbbm{j}}$} \\ \hline
				
				$y_{\mathbbm{i}}$, $\mathbbm{x}_{\mathbbm{i}}$, $n_{\mathbbm{i}}(\bm{\pi})$ 
				& \makecell[l]{Effective congestion load, aggregate computing load, and number of associated SDs at SP $s^{\mathrm{On}}_{\mathbbm{i}}$ under profile $\bm{\pi}$} \\ \hline
				
				$t^{\mathrm{comp}}_{\mathbbm{j},\mathbbm{i}}$, $t^{\mathrm{tx}}_{\mathbbm{j},\mathbbm{i}}$, $T^{\mathrm{edge}}_{\mathbbm{j},\mathbbm{i}}$ 
				& \makecell[l]{Computation delay, uplink transmission delay, and total edge-side latency of SD $u^{\mathrm{On}}_{\mathbbm{j}}$ at SP $s^{\mathrm{On}}_{\mathbbm{i}}$} \\ \hline
				
				$S_{\mathbbm{j}}(\mathbbm{i})$, $\kappa_{\mathbbm{j}}$, $\theta_{\mathbbm{j},\mathbbm{i}}$, $\tau_{\mathbbm{i}}(y)$ 
				& \makecell[l]{Non-congestion benefit term, congestion sensitivity, load factor, and congestion function at SP $s^{\mathrm{On}}_{\mathbbm{i}}$} \\ \hline
				
				$K_{\mathbbm{i}}(\mathbbm{x}_{\mathbbm{i}})$, $G_{\mathbbm{i}}$ 
				& \makecell[l]{Computation/operation cost function and maximum number of SDs that online SP $s^{\mathrm{On}}_{\mathbbm{i}}$ can concurrently serve} \\ \hline
				
				$\bm{\pi}$, $\Phi^{\mathrm{G}}(\bm{\pi})$, $\bm{\mathcal{F}}^{(\mathrm{I})}$, $\bm{\mathcal{F}}^{(\mathrm{II})}$ 
				& \makecell[l]{Joint scheduling profile, game potential function, and optimization problems in the offline and online stages} \\ \hline 
				
			\end{tabular}
		\end{center}
	}
	
\end{table*}

\section{Details of Offline Stage}
\subsection{Detailed Pro-LNN Architecture and Training Procedure}
\label{app:prolnn}
In this section, we provide the detailed update equations and training procedure of the proposed Pro-LNN predictor used in Off-AIC$^{2}$.

\noindent\textbf{Step 1. Stable liquid cell }(line 7, Alg.~\ref{Alg1}):
At each transition $\nu\to\nu+1$, the liquid cell maps the new input $\mathbf{z}_i^{(\nu+1)}$ and the previous hidden state $\mathbf{h}_i^{(\nu)}\in\mathbb{R}^{d}$ to $\mathbf{h}_i^{(\nu+1)}$. Concretely, we first compute neuron-wise time constants and a pre-activation:
$\boldsymbol{\tau}^{(\nu)} = \mathrm{softplus} \big(W_\tau[\mathbf{z}_i^{(\nu+1)},\mathbf{h}_i^{(\nu)}]\big) + \tau_{\min}$ and
$\tilde{\mathbf{u}}^{(\nu)} = W_h \mathbf{h}_i^{(\nu)} + W_x \mathbf{z}_i^{(\nu+1)} + \mathbf{b}$,
where $W_\tau,W_h,W_x$ and $\mathbf{b}$ are trainable, and $\tau_{\min}>0$ enforces strictly positive time scales (we also clamp $\boldsymbol{\tau}^{(\nu)}$ to $[\tau_{\min},\tau_{\max}]$ for numerical stability). The effective update rate is then
$\boldsymbol{\alpha}^{(\nu)} 
= \mathrm{clip}\Big(\frac{\Delta \nu}{\boldsymbol{\tau}^{(\nu)}}, \alpha_{\min}, \alpha_{\max}\Big)$,
with a fixed $\Delta \nu$ (one trading interval in our implementation) and small constants $0<\alpha_{\min}\le\alpha_{\max}<1$. We then apply a lightweight layer normalization and non-linearity,
$\tilde{\mathbf{h}}_i^{(\nu)} = \tanh\!\big(\mathrm{LN}(\tilde{\mathbf{u}}^{(\nu)})\big)$,
and update the hidden state as
\begin{equation}
	\mathbf{h}_i^{(\nu+1)} 
	= (1-\boldsymbol{\alpha}^{(\nu)})\odot \mathbf{h}_i^{(\nu)}
	+ \boldsymbol{\alpha}^{(\nu)}\odot \tilde{\mathbf{h}}_i^{(\nu)}
	+ \gamma\,\tanh\!\big(\mathrm{LN}(\mathbf{h}_i^{(\nu)})\big),
\end{equation}
where $\gamma$ is a small residual scaling factor and $\odot$ denotes element-wise multiplication. This leaky-integrator update blends exponential memory decay with input-driven state injection, while the neuron-wise $\boldsymbol{\tau}^{(\nu)}$ allow different neurons to operate at different effective time scales.

\noindent\textit{Notation.}
For clarity, all nonlinear operators in the above update are applied element-wise unless otherwise stated. Specifically,
$\mathrm{softplus}(\cdot)$ denotes the standard softplus activation
$\mathrm{softplus}(x)=\log(1+e^{x})$,
$\mathrm{clip}(x,a,b)$ clips each element of $x$ into the interval $[a,b]$,
$\mathrm{LN}(\cdot)$ denotes layer normalization~\cite{ba2016layer},
and $\mathrm{LiquidCell}(\mathbf{z}_i^{(\nu+1)},\mathbf{h}_i^{(\nu)})$
is a shorthand for one step of the stable liquid update in Step~1
(using the time-constant gating and residual integration described above).

\noindent\textbf{Step 2. Historical encoding }(lines 5-7, Alg.~\ref{Alg1}):
Given the latest $L$ observations of ES $s_i$ up to trading round $\nu$, Pro-LNN encodes the history by iteratively applying the liquid cell:
$\mathbf{h}_i^{(\nu-L+1)} \leftarrow \mathrm{LiquidCell}\big(\mathbf{z}_i^{(\nu-L+1)},\mathbf{0}\big), 
\ldots, 
\mathbf{h}_i^{(\nu)} \leftarrow \mathrm{LiquidCell}\big(\mathbf{z}_i^{(\nu)},\mathbf{h}_i^{(\nu-1)}\big)$.

The neuron-wise time constants control memory retention and state updates for non-stationary or noisy demand trajectories.

\noindent\textbf{Step 3. Multi-horizon readout and consistency }(lines 8-17, Alg.~\ref{Alg1}):
After encoding the history, a shallow MLP readout maps the final hidden state to an $H$-step forecast:
\begin{equation}
	\hat{\mathbf{n}}_i^{(1:H)} = \mathrm{Readout}\big(\mathbf{h}_i^{(\nu)}\big)
	= [\hat N_i^{\mathrm{Dem},(\nu+1)},\ldots,\hat N_i^{\mathrm{Dem},(\nu+H)}],
\end{equation}
where $\mathrm{Readout}(\cdot)$ denotes a lightweight fully connected network that maps 
the $d$-dimensional hidden state to an $H$-dimensional prediction vector.
The base training loss is the mean squared error (MSE) over the $H$ horizons. To further improve stability for long horizons, we optionally add a \emph{multi-step consistency} objective: starting from a history window, we repeatedly roll the model forward in closed loop by feeding its own one-step predictions back into the window, and penalize the mismatch at each rolled-out step. 

\begin{algorithm}[t]
	\footnotesize
	\caption{Pro-LNN}
	\label{Alg1}
	\KwIn{Training sequences $\{\mathbf{z}_i^{(1:T)}\}$; history length $L$; horizon $H$; learning rate $\eta$; consistency depth $K$}
	\KwOut{Predicted sequences $\{\hat N_i^{\mathrm{Dem},(\nu+1:\nu+H)}\}$}
	
	Initialize liquid-cell and readout parameters $\theta$; \\
	\While{not converged}{
		Sample a mini-batch $\mathcal{B}$ of training windows from $\mathcal{D}$; \\
		\For{each window $(\mathbf{z}_i^{(\nu-L+1:\nu)}, N_i^{\mathrm{Dem},(\nu+1:\nu+H)}) \in \mathcal{B}$}{
			Set hidden state $\mathbf{h}_i \leftarrow \mathbf{0}$; \\
			\For{$k = \nu-L+1$ \KwTo $\nu$}{
				$\mathbf{h}_i \leftarrow \mathrm{LiquidCell}\big(\mathbf{z}_i^{(k)}, \mathbf{h}_i\big)$;
			}
			\textbf{(Base $H$-step forecast)} \\
			$\hat{\mathbf{n}}_i^{(1:H)} \leftarrow \mathrm{Readout}(\mathbf{h}_i)$; \\
			$\mathcal{L}_i \leftarrow \sum_{h=1}^{H}\big(\hat n_i^{(h)} - N_i^{\mathrm{Dem},(\nu+h)}\big)^{2}$; \\
			\If{$K > 0$}{
				Construct a history window $\mathbf{y}_i$ from $\{\mathbf{z}_i^{(\nu-L+1)},\ldots,\mathbf{z}_i^{(\nu)}\}$; \\
				\For{$k = 1$ \KwTo $K$}{
					$\hat{\mathbf{n}}_i^{(1:H)} \leftarrow f_\theta(\mathbf{y}_i)$; \\
					$\mathcal{L}_i \leftarrow \mathcal{L}_i + \big(\hat n_i^{(1)} - N_i^{\mathrm{Dem},(\nu+k)}\big)^{2}$; \\
					Update $\mathbf{y}_i$ by discarding the oldest step and appending $\hat n_i^{(1)}$ as the newest demand;
				}
			}
			$\mathcal{L}_i \leftarrow \mathcal{L}_i /(H+K)$;
		}
		Aggregate loss $\mathcal{L} \leftarrow \frac{1}{|\mathcal{B}|}\sum_{i\in\mathcal{B}} \mathcal{L}_i + \lambda\,\Omega(\theta)$; \\
		Update parameters $\theta \leftarrow \theta - \eta\nabla_\theta \mathcal{L}$;
	}
\end{algorithm}

\noindent\textbf{Step 4. Inference and integration }:
At inference, Pro-LNN takes, for each ES $s_i$, the latest $L$ observations $\{\mathbf{z}_i^{(\nu-L+1)},\ldots,\mathbf{z}_i^{(\nu)}\}$, encodes them via the liquid cell into $\mathbf{h}_i^{(\nu)}$, and applies the readout to obtain an $H$-step forecast $\{\hat N_i^{\mathrm{Dem},(\nu+1)},\ldots,\hat N_i^{\mathrm{Dem},(\nu+H)}\}$. These predictions are computed in the offline stage before each trading round and then fed into: \textit{(i)} the eACO-VRP module for time-window-aware ES recommendation and AP route pre-planning, and \textit{(ii)} the auction module for ES--AP contract design. Pro-LNN inference is completed during offline service preparation, before the corresponding online service window begins.

{\subsection{Economic Properties of Off-AIC$^2$}\label{app:auction}

We analyze the economic properties of Off-AIC$^2$ under standard mechanism-design assumptions. Specifically, participants have quasi-linear utilities; each ES has an independent private valuation for receiving service, and each AP has an independent private cost for providing service; participants do not collude; deterministic tie-breaking is used whenever multiple allocations have the same score; and service-capacity/window information is verifiable by the auction platform. In addition, the final settlement of each active ES--AP contract is selected within a non-deficit settlement band, as follows:
\begin{equation}
	\mathrm{ask}_k
	\le
	r_{i,k}^{\mathrm{ES}}
	\le
	p_{i,k}^{\mathrm{ES}}
	\le
	\overline{\mathrm{bid}}_i ,
	\qquad
	\forall (i,k)\ \text{with }x_{i,k}^{\mathrm{Off}}=1,
	\label{eq:settlement_band}
\end{equation}
which ensures that the ES payment is sufficient to cover the AP reward while remaining no larger than the ES-side declared valuation. These assumptions define the applicability scope of the individual-rationality, $\epsilon$-near-truthfulness, and weak-budget-balance results established below.

\begin{Defn}[Individual Rationality]
	Off-AIC$^2$ is individually rational if every winning ES obtains nonnegative utility, i.e., it pays no more than its true valuation for the contracted service, and every winning AP obtains nonnegative utility, i.e., it receives a reward no lower than its true service cost.
\end{Defn}

\begin{Defn}[$\epsilon$-Near-Truthfulness]
	Let $q$ denote any participant in the offline ES--AP contracting mechanism, with true type $\theta_q$, reported type $\hat{\theta}_q$, and other participants' reports $\mathbf{b}_{-q}$. Let $U_q(\theta_q;\hat{\theta}_q,\mathbf{b}_{-q})$ denote the utility of participant $q$ when its true type is $\theta_q$ but it reports $\hat{\theta}_q$. Off-AIC$^2$ is $\epsilon$-near-truthful if, for any participant $q$ and any possible misreport $\hat{\theta}_q\in\Theta_q$,
	\begin{equation}
		U_q(\theta_q;\theta_q,\mathbf{b}_{-q})
		\ge
		\sup_{\hat{\theta}_q\in\Theta_q}
		U_q(\theta_q;\hat{\theta}_q,\mathbf{b}_{-q})
		-\epsilon .
	\end{equation}
	This means that the maximum utility improvement obtainable through unilateral misreporting is bounded by $\epsilon$.
\end{Defn}

\begin{Defn}[Weak Budget Balance]
	Off-AIC$^2$ is weakly budget balanced if the total payment collected from all winning ESs is no smaller than the total reward paid to all winning APs. Equivalently, the auctioneer does not incur a deficit.
\end{Defn}

\begin{Defn}[Bounded Approximation Error]
	In an exactly monotone single-parameter allocation rule, raising an ES's bid should not reduce its chance of winning, and lowering an AP's ask should not reduce its chance of winning. In Off-AIC$^2$, the route-aware eACO-VRP subroutine may introduce small non-monotone perturbations because ES--AP matching is coupled with route feasibility, service-window compatibility, AP capacity, and travel costs. For each participant $q$, let $\Delta_q^{\mathrm{alloc}}$ denote the maximum utility gain that can be caused by such route-induced non-monotone allocation perturbations, and let $\eta_q^{\mathrm{price}}$ denote the maximum utility gain caused by binary-search approximation in the critical-value-like pricing step. We assume that
	\begin{equation}
		\Delta_q^{\mathrm{alloc}}+\eta_q^{\mathrm{price}}\le \epsilon_q,
		\qquad
		\epsilon=\max_q\epsilon_q .
	\end{equation}
\end{Defn}

We next prove that Off-AIC$^2$ satisfies individual rationality, $\epsilon$-near-truthfulness, and weak budget balance under the above assumptions.

\begin{Prop}[Individual Rationality of Off-AIC$^2$]
	All winning ESs and APs obtained from Off-AIC$^2$ are individually rational in the offline stage.
\end{Prop}

\begin{proof}
	The individual rationality of Off-AIC$^2$ follows from the candidate-pair screening rule in Step~2 and the settlement rule in Step~3 of Alg.~\ref{Alg3}. We prove it separately for the ES side and the AP side.
	
	\noindent
	\textbf{\textit{(i) ES side.}}
	Consider a winning ES $s_i$ matched with AP $v_k$. In Step~2 of Alg.~\ref{Alg3}, for each AP $v_k$, the candidate ES set is constructed as
	\begin{equation}
		\mathcal{C}_k =
		\big\{
		s_i\in\mathcal{U}^{\mathrm{free}}
		\mid
		U_{i,k}^{\mathrm{Off}}>0,\ 
		\overline{\mathrm{bid}}_i\ge \mathrm{ask}_k
		\big\},
	\end{equation}
	where $U_{i,k}^{\mathrm{Off}}=\overline{\mathrm{bid}}_i-\mathrm{ask}_k$ denotes the pairwise surplus. Hence, if $s_i$ is finally matched with $v_k$, the pair must satisfy $\overline{\mathrm{bid}}_i\ge \mathrm{ask}_k$ and has a nonempty settlement band.
	
	Let $v_i$ denote the true per-unit valuation of ES $s_i$. Under truthful reporting, $v_i=\overline{\mathrm{bid}}_i$. By the settlement rule in~\eqref{eq:settlement_band}, the final ES-side payment satisfies
	\begin{equation}
		p_{i,k}^{\mathrm{ES}}\le \overline{\mathrm{bid}}_i=v_i .
	\end{equation}
	Therefore, the utility of winning ES $s_i$ is
	\begin{equation}
		U_i^{\mathrm{ES,Off}}
		=
		N_{i,k}^{\mathrm{Trad}}
		\big(v_i-p_{i,k}^{\mathrm{ES}}\big)
		\ge 0 .
	\end{equation}
	Thus, every winning ES obtains nonnegative utility.
	
	\medskip
	\noindent
	\textbf{\textit{(ii) AP side.}}
	Consider a winning AP $v_k$. In Step~2, AP $v_k$ can be matched only with ESs that satisfy the surplus and feasibility conditions, and its total assigned demand must obey the capacity constraint
	\begin{equation}
		\sum_i x_{i,k}^{\mathrm{Off}}D_i\le Q_k^{\mathrm{A}} .
	\end{equation}
	Let $c_k$ denote the true per-unit service cost of AP $v_k$. Under truthful reporting, $c_k=\mathrm{ask}_k$. By the settlement rule in~\eqref{eq:settlement_band}, for every active pair $(i,k)$, the AP-side reward satisfies
	\begin{equation}
		r_{i,k}^{\mathrm{ES}}\ge \mathrm{ask}_k=c_k .
	\end{equation}
	Hence, the utility of AP $v_k$ is
	\begin{equation}
		U_k^{\mathrm{A,Off}}
		=
		\sum_i x_{i,k}^{\mathrm{Off}}N_{i,k}^{\mathrm{Trad}}
		\big(r_{i,k}^{\mathrm{ES}}-c_k\big)
		\ge 0 .
	\end{equation}
	Therefore, every winning AP obtains nonnegative utility. Combining the ES and AP sides, Off-AIC$^2$ is individually rational in the offline stage.
\end{proof}

\begin{Prop}[$\epsilon$-Near-Truthfulness of Off-AIC$^2$]
	Under quasi-linear utilities, independent private values/costs, deterministic tie-breaking, and bounded approximation error, Off-AIC$^2$ is $\epsilon$-near-truthful for both ESs and APs in the offline stage.
\end{Prop}

\begin{proof}
	We first clarify the role of $\epsilon$. If the ES--AP allocation rule were exactly monotone and the payment/reward were computed as the exact critical value, the standard single-parameter auction argument would imply dominant-strategy truthfulness. In Off-AIC$^2$, however, ES--AP allocation is coupled with route feasibility, AP capacity, service windows, and the eACO-VRP route-construction subroutine. These factors may introduce small deviations from exact monotonicity. In addition, the critical-value-like payment and reward are obtained through binary search and may contain a small approximation error. By Definition 4, the maximum utility gain caused by these two effects is bounded by $\epsilon_q$ for participant $q$.
	
	At a high level, Off-AIC$^2$ follows the same structural principle as classical truthful single-parameter auctions: it combines an approximately monotone allocation rule with critical-value-like payments. We prove the near-truthfulness property for ES buyers and AP sellers separately.
	
	\noindent
	\textbf{\textit{(i) ES side.}}
	Fix any ES $s_i$ and treat all other bids and asks as given. Under the single-parameter abstraction, $s_i$ reports a per-unit bid $b_i$, which determines its average bid $\overline{\mathrm{bid}}_i=b_i$ and its pairwise surplus values with APs.
	
	\textit{Approximate monotone allocation.}
	In Step~1, ESs are sorted in nonincreasing order of $\overline{\mathrm{bid}}_i=b_i$. Increasing $b_i$ can only move $s_i$ upward in the sorted list and therefore cannot worsen its admission rank. Since Off-AIC$^2$ admits the top-$I^\ast$ ESs into the subsequent matching stage, once $s_i$ belongs to the admitted ES set at bid $b_i$, it remains admitted for any higher bid $b_i'>b_i$.
	
	In Step~2, $s_i$ can be matched with AP $v_k$ only if
	\begin{equation}
		U_{i,k}^{\mathrm{Off}}
		=
		b_i-\mathrm{ask}_k
		>0,
		\qquad
		b_i\ge \mathrm{ask}_k .
	\end{equation}
	Thus, increasing $b_i$ weakly increases all pairwise surplus values involving $s_i$, making $s_i$ no less attractive in the candidate matching stage. If the route-aware allocation were exactly monotone, this would imply exact ES-side monotonicity. With eACO-VRP, the selected ES set may be slightly perturbed by route feasibility, service-window compatibility, travel cost, and AP capacity. By the bounded approximation-error definition, any additional utility gain caused by such route-induced non-monotonicity and pricing approximation is at most $\epsilon_i$.
	
	\textit{Critical-value-like payment.}
	For each winning ES $s_i$, Step~3 treats its effective bid as a variable and performs binary search over the feasible interval to find the smallest effective bid that keeps $s_i$ in the winning set. The resulting value is used to determine the payment $p_{i,k}^{\mathrm{ES}}$ within the settlement band. Under exact monotonicity and exact critical-value pricing, this payment would eliminate profitable unilateral bid deviations. Under approximate monotonicity and finite-precision binary search, all remaining deviation gains are bounded by $\epsilon_i$. Therefore, for any misreported bid $\hat b_i$,
	\begin{equation}
		U_i(b_i;b_i,\mathbf{b}_{-i})
		\ge
		U_i(b_i;\hat b_i,\mathbf{b}_{-i})
		-\epsilon_i .
	\end{equation}
	
	\noindent
	\textbf{\textit{(ii) AP side.}}
	Now fix any AP $v_k$ and treat all ES bids and other APs' asks as given. Under the single-parameter abstraction, $v_k$ reports a per-unit ask $a_k$.
	
	\textit{Approximate monotone allocation.}
	In Step~1, APs are sorted in nondecreasing order of $\mathrm{ask}_k=a_k$. Lowering $a_k$ can only move $v_k$ upward in the sorted list and therefore cannot worsen its admission rank. Since Off-AIC$^2$ admits the top-$K^\ast$ APs into the matching stage, once $v_k$ is admitted at ask $a_k$, it remains admitted for any lower ask $a_k'<a_k$.
	
	In Step~2, $v_k$'s competitiveness is determined by the pairwise surplus
	\begin{equation}
		U_{i,k}^{\mathrm{Off}}
		=
		\overline{\mathrm{bid}}_i-a_k .
	\end{equation}
	Lowering $a_k$ weakly increases all surplus values involving $v_k$, making it no less competitive in the matching stage. If the allocation were exactly monotone, lowering the ask would not reduce the chance that $v_k$ wins. With eACO-VRP, route-aware feasibility, travel cost, service windows, and AP capacity coupling may introduce small non-monotone perturbations. By the bounded approximation-error definition, any additional utility gain caused by such perturbations and pricing approximation is at most $\epsilon_k$.
	
	\textit{Critical-reward-like payment.}
	For each winning AP $v_k$, Step~3 treats its effective ask as a variable and performs binary search to find the largest effective ask at which $v_k$ still remains in the winning set. The resulting value is used to determine the reward $r_{i,k}^{\mathrm{ES}}$ within the settlement band. Under exact monotonicity and exact critical-value pricing, this reward would eliminate profitable unilateral ask deviations. Under approximate monotonicity and finite-precision binary search, all remaining deviation gains are bounded by $\epsilon_k$. Thus, for any misreported ask $\hat a_k$,
	\begin{equation}
		U_k(a_k;a_k,\mathbf{a}_{-k})
		\ge
		U_k(a_k;\hat a_k,\mathbf{a}_{-k})
		-\epsilon_k .
	\end{equation}
	
	Combining the ES and AP sides, for every participant $q$ and every unilateral misreport $\hat{\theta}_q$, we have
	\begin{equation}
		U_q(\theta_q;\theta_q,\mathbf{b}_{-q})
		\ge
		U_q(\theta_q;\hat{\theta}_q,\mathbf{b}_{-q})
		-\epsilon_q .
	\end{equation}
	Taking the supremum over all possible misreports and using $\epsilon=\max_q\epsilon_q$, we obtain
	\begin{equation}
		U_q(\theta_q;\theta_q,\mathbf{b}_{-q})
		\ge
		\sup_{\hat{\theta}_q\in\Theta_q}
		U_q(\theta_q;\hat{\theta}_q,\mathbf{b}_{-q})
		-\epsilon .
	\end{equation}
	Therefore, Off-AIC$^2$ is $\epsilon$-near-truthful for both ESs and APs in the offline stage. This result should be interpreted as an approximate incentive guarantee rather than strict dominant-strategy truthfulness. The numerical evidence in Figs.~\ref{fig:truth_agp}--\ref{fig:truth_es} further supports this analysis by showing that the realized utility curves are maximized near truthful reports and that unilateral deviations bring at most marginal gains in the tested instances.
\end{proof}

\begin{Prop}[Weak Budget Balance of Off-AIC$^2$]
	Under the settlement condition in~\eqref{eq:settlement_band}, Off-AIC$^2$ satisfies weak budget balance in the offline stage.
\end{Prop}

\begin{proof}
	To verify weak budget balance, it suffices to show that the auctioneer's utility in the offline stage is nonnegative. The auctioneer's utility is
	\begin{equation}
		U^{\mathrm{Auc,Off}}
		=
		\sum_{s_i\in\bm{\mathcal{S}}}
		\sum_{v_k\in\bm{\mathcal{V}}}
		x_{i,k}^{\mathrm{Off}}N_{i,k}^{\mathrm{Trad}}
		\big(p_{i,k}^{\mathrm{ES}}-r_{i,k}^{\mathrm{ES}}\big).
	\end{equation}
	Since $x_{i,k}^{\mathrm{Off}}\in\{0,1\}$ and $N_{i,k}^{\mathrm{Trad}}\ge 0$, it is enough to prove that every active trading pair satisfies
	\begin{equation}
		p_{i,k}^{\mathrm{ES}}-r_{i,k}^{\mathrm{ES}}\ge 0 .
	\end{equation}
	
	For every active pair $(i,k)$, the final settlement rule in~\eqref{eq:settlement_band} gives
	\begin{equation}
		r_{i,k}^{\mathrm{ES}}
		\le
		p_{i,k}^{\mathrm{ES}} .
	\end{equation}
	Therefore,
	\begin{equation}
		x_{i,k}^{\mathrm{Off}}N_{i,k}^{\mathrm{Trad}}
		\big(p_{i,k}^{\mathrm{ES}}-r_{i,k}^{\mathrm{ES}}\big)
		\ge 0 .
	\end{equation}
	Unmatched pairs have $x_{i,k}^{\mathrm{Off}}=0$ and thus contribute zero. Summing over all ES--AP pairs yields
	\begin{equation}
		U^{\mathrm{Auc,Off}}
		=
		\sum_{i,k}
		x_{i,k}^{\mathrm{Off}}N_{i,k}^{\mathrm{Trad}}
		\big(p_{i,k}^{\mathrm{ES}}-r_{i,k}^{\mathrm{ES}}\big)
		\ge 0 .
	\end{equation}
	Therefore, the total payment collected from winning ESs is no smaller than the total reward paid to winning APs, and Off-AIC$^2$ satisfies weak budget balance in the offline stage.
\end{proof}}

{\subsection{Evaluation and Limitations of eACO-VRP}
\label{app:eaco_rh}

\begin{table*}[b]
	\centering
	\small
	\caption{Rolling-horizon robustness and scalability of eACO-VRP. Normalized SW is computed relative to eACO-RH under the same scenario, seed, and trading round.}
    
	\label{tab:R2_eaco_scalability_robustness}
	\resizebox{\textwidth}{!}{
		\begin{tabular}{llcccc}
			\toprule
			\textbf{Scenario} 
			& \textbf{Method} 
			& \textbf{Runtime} 
			& \textbf{Feasible Route Ratio} 
			& \textbf{Route Utility} 
			& \textbf{Normalized SW} \\
			\midrule
			Small scale & eACO-RH & 0.0610 $\pm$ 0.0087s & 1.000 $\pm$ 0.000 & 18.26 $\pm$ 2.03 & 1.000 $\pm$ 0.000 \\
			Small scale & eACO-Static & 0.0008 $\pm$ 0.0003s & 1.000 $\pm$ 0.000 & -1.57 $\pm$ 3.16 & 0.785 $\pm$ 0.034 \\
			Small scale & Greedy-RH & 0.0001 $\pm$ 0.0000s & 1.000 $\pm$ 0.000 & 18.05 $\pm$ 2.05 & 0.996 $\pm$ 0.007 \\
			Small scale & Random-RH & 0.0001 $\pm$ 0.0000s & 1.000 $\pm$ 0.000 & 14.72 $\pm$ 2.19 & 0.962 $\pm$ 0.020 \\
			\midrule
			Medium scale & eACO-RH & 0.1901 $\pm$ 0.0164s & 1.000 $\pm$ 0.000 & 36.79 $\pm$ 3.14 & 1.000 $\pm$ 0.000 \\
			Medium scale & eACO-Static & 0.0051 $\pm$ 0.0007s & 1.000 $\pm$ 0.000 & -3.63 $\pm$ 9.37 & 0.786 $\pm$ 0.057 \\
			Medium scale & Greedy-RH & 0.0004 $\pm$ 0.0002s & 1.000 $\pm$ 0.000 & 36.12 $\pm$ 3.99 & 0.996 $\pm$ 0.007 \\
			Medium scale & Random-RH & 0.0001 $\pm$ 0.0000s & 1.000 $\pm$ 0.000 & 28.69 $\pm$ 3.15 & 0.956 $\pm$ 0.013 \\
			\midrule
			Demand drift & eACO-RH & 0.1145 $\pm$ 0.0070s & 1.000 $\pm$ 0.000 & 26.11 $\pm$ 1.60 & 1.000 $\pm$ 0.000 \\
			Demand drift & eACO-Static & 0.0023 $\pm$ 0.0004s & 1.000 $\pm$ 0.000 & -0.01 $\pm$ 5.72 & 0.810 $\pm$ 0.038 \\
			Demand drift & Greedy-RH & 0.0002 $\pm$ 0.0001s & 1.000 $\pm$ 0.000 & 26.00 $\pm$ 1.56 & 0.999 $\pm$ 0.006 \\
			Demand drift & Random-RH & 0.0001 $\pm$ 0.0000s & 1.000 $\pm$ 0.000 & 21.02 $\pm$ 2.09 & 0.962 $\pm$ 0.012 \\
			\midrule
			Topology perturbation & eACO-RH & 0.1115 $\pm$ 0.0094s & 1.000 $\pm$ 0.000 & 26.23 $\pm$ 2.60 & 1.000 $\pm$ 0.000 \\
			Topology perturbation & eACO-Static & 0.0023 $\pm$ 0.0005s & 1.000 $\pm$ 0.000 & -4.51 $\pm$ 5.84 & 0.775 $\pm$ 0.053 \\
			Topology perturbation & Greedy-RH & 0.0002 $\pm$ 0.0001s & 1.000 $\pm$ 0.000 & 25.48 $\pm$ 3.16 & 0.994 $\pm$ 0.011 \\
			Topology perturbation & Random-RH & 0.0001 $\pm$ 0.0000s & 1.000 $\pm$ 0.000 & 20.99 $\pm$ 2.50 & 0.959 $\pm$ 0.014 \\
			\bottomrule
	\end{tabular}}
\end{table*}

This appendix evaluates eACO-VRP as an offline or rolling-horizon AP route planner operating at the trading-round timescale. Slot-level SD--SP assignment is handled by PG-BRD. At the beginning of each trading round, Pro-LNN first predicts the ES-side demand over the upcoming service window. Given the predicted demand, AP states, mobility constraints, service-window constraints, and capacity constraints, eACO-VRP generates feasible AP service routes before online task execution starts. During the subsequent online service window, the generated AP routes and ES--AP contracts are treated as prepared service resources, while fast-timescale variations in HU/MU arrivals, channel states, SD locations, and congestion states are handled by the online PG-BRD scheduler.

This design follows the operational structure of air--ground service provisioning. AP/UAV routes cannot be arbitrarily redesigned at every online slot because vehicle movement, UAV dispatching, service-window feasibility, and UAV recovery operations impose physical constraints. In contrast, SD--SP task assignment can be updated more frequently according to the observed online states. Therefore, eACO-VRP is responsible for route-aware AP preparation at the trading-round timescale, whereas PG-BRD provides fine-grained online adaptability at the slot level.
To evaluate whether eACO-VRP is suitable for this intended role, we conduct a rolling-horizon robustness and scalability study. We compare four route-planning methods. \emph{eACO-RH} denotes the proposed rolling-horizon eACO-VRP, which replans AP routes at the beginning of each trading round based on the latest predicted demand. \emph{eACO-Static} denotes a one-shot planner that uses the initial eACO route plan and keeps the same routes in later rounds; later rounds only re-evaluate the fixed routes under current demand and perturbations. \emph{Greedy-RH} and \emph{Random-RH} are lightweight rolling-horizon baselines that construct routes greedily or randomly from feasible ES candidates.

We consider four representative scenarios. The small-scale setting contains 8 ESs, 4 APs, and 80 SDs, including 40 HUs and 40 MUs. The medium-scale setting contains 16 ESs, 8 APs, and 160 SDs, including 80 HUs and 80 MUs. The demand-drift setting contains 12 ESs, 6 APs, and 120 SDs, including 60 HUs and 60 MUs; it applies a drift intensity of 0.3 by blending the original demand vector with a spatially shifted demand vector and adding lognormal temporal perturbations. The topology-perturbation setting also contains 12 ESs, 6 APs, and 120 SDs, but applies a 20\% multiplicative perturbation to the route/travel-cost matrix. Each scenario is evaluated over 5 rolling trading rounds and 10 random seeds, yielding 50 seed-round observations for each method.
We report four metrics. \emph{Runtime} measures the route-planning wall-clock time per trading round. \emph{Feasible route ratio} measures the fraction of generated routes satisfying route, service-window, and capacity feasibility. \emph{Route utility} is the cumulative AP-side route utility induced by the selected ES visiting sequence. \emph{Normalized SW} is the downstream system social welfare after applying the generated routes, normalized by the eACO-RH result under the same scenario, seed, and trading round.

\begin{table*}[!b]
	\centering
	\caption{Roles of different components in the FUSION framework.}
    %\small
	\label{tab:app_fusion_component_roles}
	\resizebox{\textwidth}{!}{
		\begin{tabular}{p{0.13\textwidth}p{0.25\textwidth}p{0.27\textwidth}p{0.27\textwidth}}
			\toprule
			\textbf{Component} 
			& \textbf{Information state / timescale} 
			& \textbf{Decision role in FUSION} 
			& \textbf{Reason for adoption} \\
			\midrule
			Pro-LNN 
			& Before task arrivals; only historical records and aggregate ES-side demand traces are available. 
			& Predicts future ES-side aggregate demand for proactive offline planning. 
			& Purely online scheduling cannot anticipate future overloaded ESs or prepare AP resources in advance. \\
			\midrule
			eACO-VRP 
			& Before the online service window; predicted demand, AP states, mobility/resource constraints, and service-window information are available. 
			& Constructs route-feasible AP service plans and ES visiting sequences. 
			& The AP preparation problem has a combinatorial vehicle routing problem (VRP)-like structure with route feasibility, service-window, and capacity constraints. \\
			\midrule
			Off-AIC$^2$ 
			& Offline contracting stage; ES/AP bids, asks, capacities, predicted demands, and route-feasible AP plans are available. 
			& Forms incentive-aware ES--AP contracts and computes payments/rewards. 
			& Route planning alone does not determine economic incentives, while contracting without route feasibility may produce physically infeasible agreements. \\
			\midrule
			PG-BRD 
			& Online execution stage; actual HU/MU arrivals, SD locations, channel states, and congestion states are observed. 
			& Refines the real-time SD--SP assignment profile $\bm{\pi}$. 
			& Provides distributed congestion-aware online scheduling with convergence under the exact-potential-game structure. \\
			\bottomrule
	\end{tabular}}
\end{table*}

Table~\ref{tab:R2_eaco_scalability_robustness} shows that eACO-RH consistently achieves the largest route utility and the highest normalized SW across all tested scenarios. Its runtime remains below one second even in the medium-scale case, which is acceptable for trading-round-level offline planning. Compared with eACO-Static, eACO-RH achieves clear gains under both demand drift and topology perturbation, showing that rolling-horizon replanning is necessary when the predicted demand or travel-cost structure changes across trading rounds. Greedy-RH is competitive in some scenarios due to its very low runtime and simple local selection rule, but eACO-RH remains consistently better in route utility and provides a more systematic search over route-feasible ES combinations. Random-RH samples from feasible ES candidates without prioritizing route utility and achieves lower route utility and normalized SW in the tested scenarios.
These results support the use of eACO-VRP as the route-aware AP planning component in FUSION. The offline AP planning problem has a combinatorial VRP-like structure with route feasibility, service-window constraints, heterogeneous AP costs, and capacity limits. eACO-VRP directly targets this structured route-construction problem and produces feasible AP routes that can be naturally consumed by Off-AIC$^2$ for ES--AP contracting. More importantly, once the AP routes and reserved capacities are prepared, PG-BRD can adapt online SD--SP assignment to real-time arrivals and congestion states. Hence, the performance of FUSION comes from the complementary roles of eACO-VRP and PG-BRD: eACO-VRP prepares high-quality route-feasible service resources at the slow timescale, while PG-BRD performs congestion-aware task assignment at the fast timescale.

Nevertheless, eACO-VRP has limitations. If abrupt road blockages, emergency events, severe topology changes, or post-disaster infrastructure failures occur within a service window, the pre-planned AP route may become suboptimal or infeasible. In such cases, event-triggered replanning, dynamic ACO, model predictive control, or multi-agent reinforcement learning could be incorporated to improve real-time responsiveness. Accordingly, eACO-VRP updates AP routes between trading rounds, while PG-BRD adapts SD--SP assignments during online execution.
}

{
	\section{Rationale for the Hybrid Design of FUSION}
	\label{app:hybrid_rationale}

	FUSION is modularized according to the information structure and decision dependencies of the considered air--ground edge service provisioning problem. Before the online service window starts, the system only has access to historical demand records, predicted ES-side aggregate demand, AP states, mobility/resource constraints, and ES/AP-side bidding information. Therefore, decisions that must be prepared in advance, including AP route planning, service-capacity reservation, and ES--AP contracting, are handled in the offline stage. After actual HU/MU arrivals, SD locations, channel states, and congestion states are revealed, the online stage performs real-time SD--SP assignment through PG-BRD.
	
	The four components are connected through their decision dependencies. Pro-LNN supplies demand forecasts, eACO-VRP constructs route-feasible service plans, Off-AIC$^2$ determines ES--AP contracts and settlement terms, and PG-BRD updates SD--SP assignments after online states become available. Route feasibility and economic admissibility constrain the offline service commitments, whereas realized task and congestion states determine the online assignments. Table~\ref{tab:app_fusion_component_roles} summarizes these roles.
}

\section{Properties of the Exact Potential}\label{app:potential}

This appendix states the exact-potential property of PG-BRD for the potential-aligned payoff in~\eqref{eq:unified_utility_new}. Under this payoff, every unilateral feasible deviation produces the same change in the SD's payoff and in the game potential $\Phi^{\mathrm{G}}(\bm{\pi})$ defined in~\eqref{eq:potential_def_pi_final}. The result concerns the potential-aligned scheduling payoff, not the raw HU/MU utility alone.

\begin{Prop}[Exact potential and finite improvement]
	Consider the online SD-side scheduling game in which each SD $u_{\mathbbm{j}}^{\mathrm{On}}$ chooses one feasible action from $\mathcal{A}_{\mathbbm{j}}^{\mathrm{feas}}$ and evaluates each action according to the potential-aligned payoff in~\eqref{eq:unified_utility_new}. The posted prices $\{p_{\mathbbm{j},\mathbbm{i}}\}$, WPS weights $\{\gamma_{\mathbbm{j},\mathbbm{i}}\}$, coverage sets, and capacity limits are fixed throughout one PG-BRD run. Then $\Phi^{\mathrm{G}}(\bm{\pi})$ in~\eqref{eq:potential_def_pi_final} is an exact potential. Specifically, for any SD $u_{\mathbbm{j}}^{\mathrm{On}}$ and any unilateral deviation from action $\mathbbm{i}$ to action $\mathbbm{i}'$,
	\begin{equation}
		\widetilde U_{\mathbbm{j},\mathbbm{i}'}(\bm{\pi}_{-\mathbbm{j}})
		-
		\widetilde U_{\mathbbm{j},\mathbbm{i}}(\bm{\pi}_{-\mathbbm{j}})
		=
		\Phi^{\mathrm{G}}(\bm{\pi}_{-\mathbbm{j}},\mathbbm{i}')
		-
		\Phi^{\mathrm{G}}(\bm{\pi}_{-\mathbbm{j}},\mathbbm{i}).
	\end{equation}
	Therefore, every accepted payoff-improving update strictly increases $\Phi^{\mathrm{G}}(\bm{\pi})$. Since the feasible joint action space is finite, PG-BRD with $\varepsilon=0$ terminates at a pure-strategy NE. With a positive threshold $\varepsilon>0$, termination without accepted updates yields an $\varepsilon$-NE. If $T_{\max}$ is reached before the convergence condition is satisfied, PG-BRD returns the final feasible profile without an equilibrium guarantee.
\end{Prop}

The full proof is given in Appx.~E. In particular, Appx.~E.3 proves the exact-potential identity, Appx.~E.4 explains the active-regime MU soft-deadline form, Appx.~E.5 proves the finite-improvement and $\varepsilon$-NE claims, and Appx.~E.6 clarifies the distinction between transfer-free welfare and the game potential.

\section{Proof Details for the Exact Potential and PG-BRD}\label{app:tech_details}

This appendix proves the exact identity between unilateral changes in the potential-aligned scheduling payoff and changes in the game potential, followed by the finite-improvement and $\varepsilon$-NE results.

\subsection{Matrix Formulation Based on the Global Assignment Matrix}\label{app:matrixX}

For implementation and analysis, we represent the joint SD assignment profile by a global one-hot matrix
\begin{equation}
	\bm{X}\in\{0,1\}^{|\mathcal{U}^{\mathrm{On}}|\times(1+|\mathcal{S}^{\mathrm{On}}|)}.
\end{equation}
The $\mathbbm{j}$-th row is
\begin{equation}
	\bm{X}^{\mathrm{On}}_{\mathbbm{j}}
	=
	(x^{\mathrm{On}}_{\mathbbm{j},0},
	x^{\mathrm{On}}_{\mathbbm{j},1},
	\ldots,
	x^{\mathrm{On}}_{\mathbbm{j},|\mathcal{S}^{\mathrm{On}}|}),
\end{equation}
where $x^{\mathrm{On}}_{\mathbbm{j},0}=1$ denotes local execution and
$x^{\mathrm{On}}_{\mathbbm{j},\mathbbm{i}}=1$ denotes offloading to SP $s^{\mathrm{On}}_{\mathbbm{i}}$. The one-hot constraint is
\begin{equation}
	\sum_{\mathbbm{i}=0}^{|\mathcal{S}^{\mathrm{On}}|}
	x^{\mathrm{On}}_{\mathbbm{j},\mathbbm{i}}
	=1,\qquad
	\forall u_{\mathbbm{j}}^{\mathrm{On}}\in\mathcal{U}^{\mathrm{On}}.
\end{equation}
Equivalently, the scalar action $\pi_{\mathbbm{j}}$ and the one-hot vector $\bm{X}^{\mathrm{On}}_{\mathbbm{j}}$ represent the same decision.

For each online SP $s_{\mathbbm{i}}^{\mathrm{On}}$, the effective congestion load and physical workload induced by $\bm{X}$ are
\begin{equation}
	y_{\mathbbm{i}}(\bm{X})
	=
	\sum_{\mathbbm{j}\in\mathcal{U}^{\mathrm{On}}}
	\gamma_{\mathbbm{j},\mathbbm{i}}
	x^{\mathrm{On}}_{\mathbbm{j},\mathbbm{i}},
	\qquad
	\mathbbm{x}_{\mathbbm{i}}(\bm{X})
	=
	\sum_{\mathbbm{j}\in\mathcal{U}^{\mathrm{On}}}
	r_{\mathbbm{j}}
	x^{\mathrm{On}}_{\mathbbm{j},\mathbbm{i}}.
\end{equation}
These quantities coincide with $y_{\mathbbm{i}}(\bm{\pi})$ and $\mathbbm{x}_{\mathbbm{i}}(\bm{\pi})$ in~\eqref{eq:cong_vars_pi_final} whenever the rows of $\bm{X}$ are induced by the scalar profile $\bm{\pi}$.

\subsection{Marginal-Contribution Form of the Scheduling Payoff}\label{app:marginal_payoff}

Fix an SD $u_{\mathbbm{j}}^{\mathrm{On}}$ and the action profile of all other SDs, denoted by $\bm{\pi}_{-\mathbbm{j}}$. For each offloading action $\mathbbm{i}\in\mathcal{S}^{\mathrm{On}}$, define the load and workload generated by all SDs except $u_{\mathbbm{j}}^{\mathrm{On}}$ as
\begin{equation}
	y_{\mathbbm{i}}^{-\mathbbm{j}}
	=
	\sum_{\mathbbm{v}\neq\mathbbm{j}:\pi_{\mathbbm{v}}=\mathbbm{i}}
	\gamma_{\mathbbm{v},\mathbbm{i}},
	\qquad
	\mathbbm{x}_{\mathbbm{i}}^{-\mathbbm{j}}
	=
	\sum_{\mathbbm{v}\neq\mathbbm{j}:\pi_{\mathbbm{v}}=\mathbbm{i}}
	r_{\mathbbm{v}}.
\end{equation}
The marginal congestion contribution of assigning $u_{\mathbbm{j}}^{\mathrm{On}}$ to $s_{\mathbbm{i}}^{\mathrm{On}}$ is
\begin{equation}
	\mathcal{C}_{\mathbbm{j},\mathbbm{i}}(\bm{\pi}_{-\mathbbm{j}})
	=
	\int_{y_{\mathbbm{i}}^{-\mathbbm{j}}}^{y_{\mathbbm{i}}^{-\mathbbm{j}}+\gamma_{\mathbbm{j},\mathbbm{i}}}
	\tau_{\mathbbm{i}}(z)\,dz ,
\end{equation}
and the marginal SP-side operating cost is
\begin{equation}
	\mathcal{K}_{\mathbbm{j},\mathbbm{i}}(\bm{\pi}_{-\mathbbm{j}})
	=
	K_{\mathbbm{i}}\big(\mathbbm{x}_{\mathbbm{i}}^{-\mathbbm{j}}+r_{\mathbbm{j}}\big)
	-
	K_{\mathbbm{i}}\big(\mathbbm{x}_{\mathbbm{i}}^{-\mathbbm{j}}\big).
\end{equation}
For local execution, we use the convention
\begin{equation}
	S_{\mathbbm{j}}(0)=0,\quad
	p_{\mathbbm{j},0}=0,\quad
	\mathcal{C}_{\mathbbm{j},0}(\bm{\pi}_{-\mathbbm{j}})=0,\quad
	\mathcal{K}_{\mathbbm{j},0}(\bm{\pi}_{-\mathbbm{j}})=0.
\end{equation}
Then the potential-aligned payoff in~\eqref{eq:unified_utility_new} can be written compactly as
\begin{equation}
	\widetilde U_{\mathbbm{j},\mathbbm{i}}(\bm{\pi}_{-\mathbbm{j}})
	=
	S_{\mathbbm{j}}(\mathbbm{i})
	-
	\mathcal{C}_{\mathbbm{j},\mathbbm{i}}(\bm{\pi}_{-\mathbbm{j}})
	-
	\mathcal{K}_{\mathbbm{j},\mathbbm{i}}(\bm{\pi}_{-\mathbbm{j}})
	-
	p_{\mathbbm{j},\mathbbm{i}},
	\label{eq:app_payoff_compact}
\end{equation}
for every action $\mathbbm{i}\in\{0\}\cup\mathcal{S}^{\mathrm{On}}$.

\subsection{Single-User Deviation and Exact Potential Identity}\label{app:dev}

We now prove the exact-potential identity. Fix $\bm{\pi}_{-\mathbbm{j}}$ and consider a unilateral deviation of SD $u_{\mathbbm{j}}^{\mathrm{On}}$ from action $\mathbbm{i}$ to action $\mathbbm{i}'$. Denote the two resulting profiles by
\begin{equation}
	\bm{\pi}
	=
	(\bm{\pi}_{-\mathbbm{j}},\mathbbm{i}),
	\qquad
	\bm{\pi}'
	=
	(\bm{\pi}_{-\mathbbm{j}},\mathbbm{i}').
\end{equation}
All terms in $\Phi^{\mathrm{G}}$ that are unrelated to actions $\mathbbm{i}$ and $\mathbbm{i}'$ remain unchanged. Therefore, the change in the valuation term is
\begin{equation}
	\Delta\Phi_{\mathrm{val}}
	=
	S_{\mathbbm{j}}(\mathbbm{i}')
	-
	S_{\mathbbm{j}}(\mathbbm{i}).
\end{equation}
The congestion-integral change is
\begin{equation}
	\Delta\Phi_{\mathrm{cong}}
	=
	-\mathcal{C}_{\mathbbm{j},\mathbbm{i}'}(\bm{\pi}_{-\mathbbm{j}})
	+
	\mathcal{C}_{\mathbbm{j},\mathbbm{i}}(\bm{\pi}_{-\mathbbm{j}}).
\end{equation}
Indeed, assigning $u_{\mathbbm{j}}^{\mathrm{On}}$ to the new offloading action $\mathbbm{i}'$ increases the integrated congestion at $s_{\mathbbm{i}'}^{\mathrm{On}}$ by
$\mathcal{C}_{\mathbbm{j},\mathbbm{i}'}$, while removing it from the old offloading action $\mathbbm{i}$ decreases the integrated congestion at $s_{\mathbbm{i}}^{\mathrm{On}}$ by
$\mathcal{C}_{\mathbbm{j},\mathbbm{i}}$. The local-execution case is covered by the zero convention above.

Similarly, the SP-cost change is
\begin{equation}
	\Delta\Phi_{\mathrm{cost}}
	=
	-\mathcal{K}_{\mathbbm{j},\mathbbm{i}'}(\bm{\pi}_{-\mathbbm{j}})
	+
	\mathcal{K}_{\mathbbm{j},\mathbbm{i}}(\bm{\pi}_{-\mathbbm{j}}).
\end{equation}
This expression is exact for any cost function $K_{\mathbbm{i}}(\cdot)$ with finite values; it does not require differentiability or convexity. Hence, it remains valid when $K_{\mathbbm{i}}(\cdot)$ includes the activation cost
$\omega_2 c_{\mathbbm{i}}^{\mathrm{hard}}\mathbbm{1}_{\{\mathbbm{x}_{\mathbbm{i}}>0\}}$ in~\eqref{eq:sp_cost}.

Finally, the posted-price term changes by
\begin{equation}
	\Delta\Phi_{\mathrm{price}}
	=
	-p_{\mathbbm{j},\mathbbm{i}'}
	+
	p_{\mathbbm{j},\mathbbm{i}}.
\end{equation}
Adding the four components gives
\begin{equation}
	\begin{aligned}
		\Phi^{\mathrm{G}}(\bm{\pi}')
		-
		\Phi^{\mathrm{G}}(\bm{\pi})
		=&~
		S_{\mathbbm{j}}(\mathbbm{i}')
		-
		\mathcal{C}_{\mathbbm{j},\mathbbm{i}'}(\bm{\pi}_{-\mathbbm{j}})
		-
		\mathcal{K}_{\mathbbm{j},\mathbbm{i}'}(\bm{\pi}_{-\mathbbm{j}})
		-
		p_{\mathbbm{j},\mathbbm{i}'}\\
		&-
		\Big[
		S_{\mathbbm{j}}(\mathbbm{i})
		-
		\mathcal{C}_{\mathbbm{j},\mathbbm{i}}(\bm{\pi}_{-\mathbbm{j}})
		-
		\mathcal{K}_{\mathbbm{j},\mathbbm{i}}(\bm{\pi}_{-\mathbbm{j}})
		-
		p_{\mathbbm{j},\mathbbm{i}}
		\Big].
	\end{aligned}
\end{equation}
Using~\eqref{eq:app_payoff_compact}, we obtain
\begin{equation}
	\Phi^{\mathrm{G}}(\bm{\pi}_{-\mathbbm{j}},\mathbbm{i}')
	-
	\Phi^{\mathrm{G}}(\bm{\pi}_{-\mathbbm{j}},\mathbbm{i})
	=
	\widetilde U_{\mathbbm{j},\mathbbm{i}'}(\bm{\pi}_{-\mathbbm{j}})
	-
	\widetilde U_{\mathbbm{j},\mathbbm{i}}(\bm{\pi}_{-\mathbbm{j}}).
\end{equation}
This proves that $\Phi^{\mathrm{G}}$ is an exact potential function for the PG-BRD payoff.

\subsection{Active-Regime Form of the MU Soft-Deadline Utility}\label{app:softlinear}

We next clarify how the MU soft-deadline utility connects to the non-congestion valuation term $S_{\mathbbm{j}}(\mathbbm{i})$. For MU $u_{\mathbbm{j}}^{\mathrm{On}}$ offloading to $s_{\mathbbm{i}}^{\mathrm{On}}$,
\begin{equation}
	G^{\mathrm{M}}_{\mathbbm{j},\mathbbm{i}}
	=
	v_{\mathbbm{j}}
	\frac{\big[\tau_{\mathbbm{j}}^{\mathrm{ddl}}-T^{\mathrm{edge}}_{\mathbbm{j},\mathbbm{i}}\big]_+}
	{\tau_{\mathbbm{j}}^{\mathrm{ddl}}},
	\qquad
	T^{\mathrm{edge}}_{\mathbbm{j},\mathbbm{i}}
	=
	t^{\mathrm{tx}}_{\mathbbm{j},\mathbbm{i}}
	+
	\theta_{\mathbbm{j},\mathbbm{i}}\tau_{\mathbbm{i}}(y_{\mathbbm{i}}).
\end{equation}
In the active soft-deadline regime
$T^{\mathrm{edge}}_{\mathbbm{j},\mathbbm{i}}<\tau_{\mathbbm{j}}^{\mathrm{ddl}}$, the positive-part operator is active, and thus
\begin{equation}
	G^{\mathrm{M}}_{\mathbbm{j},\mathbbm{i}}
	=
	v_{\mathbbm{j}}
	-
	v_{\mathbbm{j}}
	\frac{t^{\mathrm{tx}}_{\mathbbm{j},\mathbbm{i}}}{\tau_{\mathbbm{j}}^{\mathrm{ddl}}}
	-
	\frac{v_{\mathbbm{j}}}{\tau_{\mathbbm{j}}^{\mathrm{ddl}}}
	\theta_{\mathbbm{j},\mathbbm{i}}\tau_{\mathbbm{i}}(y_{\mathbbm{i}}).
\end{equation}
After adding the MU energy-saving term
$\mathbb{V}^{\mathrm{M}}_e c^{\mathrm{M},\mathrm{save}}_{\mathbbm{j},\mathbbm{i}}$
and subtracting the posted price, the active-regime MU utility becomes
\begin{equation}
	U^{\mathrm{M}}_{\mathbbm{j},\mathbbm{i}}
	=
	\bar S_{\mathbbm{j}}(\mathbbm{i})
	-
	\kappa_{\mathbbm{j}}
	\theta_{\mathbbm{j},\mathbbm{i}}\tau_{\mathbbm{i}}(y_{\mathbbm{i}})
	-
	p_{\mathbbm{j},\mathbbm{i}},
\end{equation}
where
\begin{equation}
	\bar S_{\mathbbm{j}}(\mathbbm{i})
	=
	v_{\mathbbm{j}}
	-
	v_{\mathbbm{j}}
	\frac{t^{\mathrm{tx}}_{\mathbbm{j},\mathbbm{i}}}{\tau_{\mathbbm{j}}^{\mathrm{ddl}}}
	+
	\mathbb{V}^{\mathrm{M}}_e c^{\mathrm{M},\mathrm{save}}_{\mathbbm{j},\mathbbm{i}},
	\qquad
	\kappa_{\mathbbm{j}}
	=
	\frac{v_{\mathbbm{j}}}{\tau_{\mathbbm{j}}^{\mathrm{ddl}}}.
\end{equation}
Therefore, in the active regime, MUs share the same congestion-sensitive structure as HUs. Outside this regime, the slope of the positive-part operator becomes zero. In implementation, the original piecewise soft-deadline expression in~\eqref{eq:mu_soft_deadline} can be evaluated directly, while the active-regime form explains the congestion sensitivity used to construct the potential-aligned payoff.

\subsection{Finite-Improvement Property and $\varepsilon$-NE Guarantee}\label{app:fip}

The exact-potential identity in Appx.~E.3 implies the finite-improvement property. Each SD has a finite feasible action set because it can only choose local execution or one accessible SP satisfying the coverage and capacity constraints. Hence, the joint feasible action space is finite.

When $\varepsilon=0$, PG-BRD accepts only strict payoff-improving unilateral deviations. By the exact-potential identity, each accepted update strictly increases $\Phi^{\mathrm{G}}(\bm{\pi})$. Since the feasible profile space is finite, the algorithm cannot generate an infinite sequence of distinct profiles with strictly increasing potential values. It must terminate at a profile $\bm{\pi}^{\star}$ where no SD has a strictly profitable unilateral feasible deviation. By definition, $\bm{\pi}^{\star}$ is a pure-strategy NE of the potential-aligned scheduling game.

When PG-BRD uses a positive threshold $\varepsilon>0$, termination with $\mathrm{converged}=\textbf{true}$ means that for every SD $u_{\mathbbm{j}}^{\mathrm{On}}$ and every feasible action $a\in\mathcal{A}_{\mathbbm{j}}^{\mathrm{feas}}$,
\begin{equation}
	\widetilde U_{\mathbbm{j}}(a\mid\bm{\pi}_{-\mathbbm{j}}^{\star})
	-
	\widetilde U_{\mathbbm{j}}(\pi_{\mathbbm{j}}^{\star}\mid\bm{\pi}_{-\mathbbm{j}}^{\star})
	\le \varepsilon.
\end{equation}
Thus, no SD can improve its potential-aligned payoff by more than $\varepsilon$ through a unilateral feasible action change, and $\bm{\pi}^{\star}$ is an $\varepsilon$-NE. If the algorithm stops because $T_{\max}$ is reached before $\mathrm{converged}=\textbf{true}$, no equilibrium guarantee is claimed.

\subsection{Transfer-Free Welfare and Game Potential}\label{app:payments}

The transfer-free welfare metric and the game potential play different roles. The online welfare metric aggregates SD-side non-congestion valuations and subtracts integrated congestion costs and SP-side operating costs:
\begin{equation}
	\mathcal{W}(\bm{\pi})
	=
	\sum_{\mathbbm{j}:\pi_{\mathbbm{j}}\neq0}
	S_{\mathbbm{j}}(\pi_{\mathbbm{j}})
	-
	\sum_{\mathbbm{i}\in\mathcal{S}^{\mathrm{On}}}
	\int_0^{y_{\mathbbm{i}}(\bm{\pi})}
	\tau_{\mathbbm{i}}(z)\,dz
	-
	\sum_{\mathbbm{i}\in\mathcal{S}^{\mathrm{On}}}
	K_{\mathbbm{i}}\big(\mathbbm{x}_{\mathbbm{i}}(\bm{\pi})\big).
\end{equation}
This metric excludes monetary transfers between SDs and SPs and is used for online welfare evaluation.

In contrast, the game potential $\Phi^{\mathrm{G}}(\bm{\pi})$ includes posted prices because posted prices affect SD-side best responses. If prices were omitted from the game potential while SDs optimized price-dependent utilities, a price-driven unilateral improvement could fail to increase the potential. Including the fixed posted price term in~\eqref{eq:potential_def_pi_final} preserves the exact-potential identity while keeping $\mathcal{W}(\bm{\pi})$ available as the transfer-free welfare metric.

{\section{Performance Evaluation on DRL-Based Online Schedulers}
\label{app:drl_comparison}
\begin{table*}[t!]
	\centering
	\small
	\caption{Comparison with DRL-based online schedulers under the same offline planning results.}
	\label{tab:R2_drl_baseline_comparison}
	\resizebox{\textwidth}{!}{
		\begin{tabular}{lccccccc}
			\toprule
			\textbf{Method} 
			& \textbf{Social Welfare} 
			& \textbf{Timely Ratio} 
			& \textbf{Served Ratio} 
			& \textbf{Avg. Served Delay} 
			& \textbf{Energy Cost} 
			& \textbf{Online Runtime} 
			& \textbf{Training Time} \\
			\midrule
			FUSION 
			& 9362.50 $\pm$ 552.09 
			& 0.769 $\pm$ 0.048 
			& 0.876 $\pm$ 0.035 
			& 2.929 $\pm$ 0.084 
			& 1470.73 $\pm$ 308.62 
			& 0.0909 $\pm$ 0.0415s 
			& None \\
			
			FUSION-DQN 
			& 8946.71 $\pm$ 710.14 
			& 0.711 $\pm$ 0.056 
			& 0.921 $\pm$ 0.047 
			& 3.214 $\pm$ 0.113 
			& 1057.22 $\pm$ 248.76 
			& 0.0307 $\pm$ 0.0046s 
			& 7.32 $\pm$ 0.46s (100 ep) \\
			
			FUSION-PPO 
			& 8848.63 $\pm$ 731.95 
			& 0.700 $\pm$ 0.058 
			& 0.919 $\pm$ 0.045 
			& 3.265 $\pm$ 0.093 
			& 1068.20 $\pm$ 263.83 
			& 0.0397 $\pm$ 0.0043s 
			& 27.54 $\pm$ 2.46s (100 ep) \\
			
			FUSION-MADDPG 
			& 8904.83 $\pm$ 727.07 
			& 0.697 $\pm$ 0.054 
			& 0.921 $\pm$ 0.047 
			& 3.313 $\pm$ 0.084 
			& 1051.14 $\pm$ 260.71 
			& 0.0474 $\pm$ 0.0156s 
			& 52.52 $\pm$ 5.06s (100 ep) \\
			
			FUSION\_NoPG 
			& 7781.43 $\pm$ 749.96 
			& 0.642 $\pm$ 0.052 
			& 0.843 $\pm$ 0.056 
			& 3.310 $\pm$ 0.131 
			& 1517.06 $\pm$ 330.30 
			& 0.0040 $\pm$ 0.0003s 
			& None \\
			
			PurOnline 
			& 6565.13 $\pm$ 545.01 
			& 0.525 $\pm$ 0.018 
			& 0.625 $\pm$ 0.022 
			& 3.222 $\pm$ 0.192 
			& 3017.13 $\pm$ 227.11 
			& 0.0281 $\pm$ 0.0030s 
			& None \\
			\bottomrule
	\end{tabular}}
\end{table*}

This appendix provides an additional comparison between FUSION and representative DRL-based online scheduling baselines. The goal is to examine whether learning-based online policies can replace the PG-BRD scheduler when the same offline service preparation is available.
For a fair comparison, all DRL baselines replace only the online PG-BRD component, while using the same offline outputs as FUSION, including the Pro-LNN demand forecasts, eACO-VRP AP routes, Off-AIC$^2$ ES--AP contracts, and reserved service capacities. Thus, FUSION and its DRL variants operate under the same predicted demand, AP/ES resource configuration, mobility setting, channel model, online task-arrival process, and feasibility constraints. The performance difference therefore reflects the impact of the online scheduling policy rather than differences in offline planning.

We consider three representative DRL baselines. \emph{DQN} is a value-based method for discrete online assignment decisions. \emph{PPO} is a policy-gradient-based method with clipped policy updates and improved training stability. \emph{MADDPG} is a representative multi-agent DRL method for distributed decision-making with centralized training. For all DRL baselines, the state includes online task-arrival information, accessible SP sets, residual subcarrier capacity, estimated delay, energy-related features, and congestion/load indicators. The action corresponds to selecting local execution or one feasible SP. The reward is aligned with the online system objective by combining service utility, deadline/timeliness reward, delay penalty, energy cost, and capacity-violation penalty.

Table~\ref{tab:R2_drl_baseline_comparison} shows that FUSION achieves the highest social welfare and timely ratio among all compared methods. Compared with FUSION-DQN, FUSION improves social welfare from $8946.71$ to $9362.50$ and timely ratio from $0.711$ to $0.769$. Similar gains are observed over FUSION-PPO and FUSION-MADDPG. These results indicate that PG-BRD better balances service value, delay sensitivity, deadline satisfaction, congestion effects, and feasibility constraints under the same offline service preparation.
The DRL baselines achieve higher served ratios and lower energy costs than FUSION. This suggests that learned policies tend to admit or serve more tasks with lower energy expenditure. However, their average served delay is larger and their timely ratios are lower, indicating that they do not always prioritize delay-sensitive or deadline-critical tasks in a way that maximizes the overall system utility. In contrast, PG-BRD explicitly exploits the exact-potential structure of the online scheduling game and updates SD--SP assignments according to potential-aligned payoff-improving best responses, which leads to higher timely service performance and social welfare.
Another important difference lies in training and interpretability. DQN, PPO, and MADDPG require additional policy training before deployment, and their performance depends on the coverage and quality of training samples. When the task-arrival distribution, mobility pattern, channel condition, or AP/ES resource configuration changes, the learned policies may need to be retrained or fine-tuned. By contrast, PG-BRD does not require offline policy training. It directly operates on the currently observed online state and converges under the finite-potential-game structure. Although PG-BRD has a slightly higher online runtime than the trained DRL policies in this experiment, its runtime remains below $0.1$ seconds on average and avoids the additional training cost.

Under identical offline service preparation, PG-BRD achieves higher SW and timely ratios without online-policy training, whereas the tested DRL policies achieve higher served ratios, lower energy costs, and shorter online runtimes. All variants retain the same offline contracting component, so this comparison isolates the online scheduling tradeoff rather than differences in offline economic properties.}

\begin{table*}[t!]
	\centering
	\small
	\caption{Performance gap and runtime comparison between FUSION and small-scale full-information benchmarks.}
	\label{tab:R2_optimality_gap_runtime}
	\resizebox{\textwidth}{!}{
		\begin{tabular}{lccccccc}
			\toprule
			\textbf{Case} 
			& \textbf{Problem Size} 
			& \textbf{Benchmark Type} 
			& \textbf{Exact/UB SW} 
			& \textbf{FUSION SW} 
			& \textbf{Gap} 
			& \textbf{Exact/UB Runtime} 
			& \textbf{FUSION Runtime} \\
			\midrule
			1 
			& $|\mathcal{U}|=10, |\mathcal{S}|=3, |\mathcal{V}|=2$ 
			& 7/10 Exact, 3/10 UB 
			& 982.50 $\pm$ 172.22 
			& 974.64 $\pm$ 173.01 
			& 0.83 $\pm$ 1.13\% 
			& 0.9005 $\pm$ 0.9207s 
			& 0.0016 $\pm$ 0.0006s \\
			
			2 
			& $|\mathcal{U}|=16, |\mathcal{S}|=4, |\mathcal{V}|=2$ 
			& 2/10 Exact, 8/10 UB 
			& 1615.33 $\pm$ 223.16 
			& 1580.05 $\pm$ 224.43 
			& 2.25 $\pm$ 1.73\% 
			& 1.7800 $\pm$ 0.5200s 
			& 0.0028 $\pm$ 0.0006s \\
			
			3 
			& $|\mathcal{U}|=20, |\mathcal{S}|=5, |\mathcal{V}|=3$ 
			& 0/10 Exact, 10/10 UB 
			& 2090.58 $\pm$ 266.99 
			& 2041.00 $\pm$ 272.14 
			& 2.44 $\pm$ 0.74\% 
			& 2.0003 $\pm$ 0.0001s 
			& 0.0050 $\pm$ 0.0017s \\
			\bottomrule
	\end{tabular}}
\end{table*}
{\section{Decomposition-Induced Optimality Gap and Runtime}
\label{app:decomposition_gap}

This appendix compares FUSION with small-scale full-information benchmarks to evaluate the performance--runtime tradeoff of the two-stage design. FUSION prepares forecasts, AP routes, and ES--AP contracts before online task arrivals and updates SD--SP assignments after online states become available.

To evaluate the decomposition-induced gap, we construct small-scale oracle benchmark instances in which future task arrivals, channel states, route feasibility, and system states are assumed to be known. For each instance, an exact branch-and-bound solver is first attempted to solve the corresponding full-information benchmark. If the exact solver terminates within the time limit, the obtained optimum is used as the benchmark value. Otherwise, we use the relaxed upper bound returned by the solver as the benchmark value. Therefore, when the benchmark type contains upper-bound cases, the reported gap should be interpreted as a conservative upper estimate of the true optimality gap. The gap is computed as
\begin{equation}
	\mathrm{Gap}
	=
	\frac{
		U^{\mathrm{Exact/UB}}_{\mathrm{SW}}
		-
		U^{\mathrm{FUSION}}_{\mathrm{SW}}
	}{
		U^{\mathrm{Exact/UB}}_{\mathrm{SW}}
	}
	\times 100\% .
\end{equation}

Table~\ref{tab:R2_optimality_gap_runtime} shows that FUSION achieves a small exact/upper-bound gap, ranging from approximately 0.83\% to 2.44\% across the tested small-scale cases. At the same time, FUSION requires substantially lower runtime than the exact/upper-bound benchmark. For example, in Case 1, FUSION obtains an SW of 974.64 compared with the exact/upper-bound value of 982.50, while reducing the runtime from 0.9005s to 0.0016s. As the problem size increases, exact branch-and-bound becomes increasingly difficult: the benchmark changes from 7/10 exact solutions in Case 1 to 0/10 exact solutions in Case 3, where all benchmark values rely on relaxed upper bounds. This trend illustrates the rapid growth of the full-information joint problem complexity.
The tested small-scale instances show limited gaps relative to exact optima or solver upper bounds and substantially lower runtime for FUSION. These comparisons characterize the tradeoff under the specified benchmark settings; they do not establish global optimality or a general optimality-gap bound for larger instances.
}

\end{document}